\pdfoutput=1
\documentclass[acmsmall,screen,nonacm,appendix]{acmart}

\usepackage[capitalize]{cleveref}
\crefname{line}{Line}{Lines}
\crefname{rule}{rule}{rules}
\Crefname{rule}{Rule}{Rules}
\crefrangelabelformat{line}{#3#1#4--#5#2#6}

\usepackage{microtype}
\usepackage{wrapfig}
\usepackage{natbib}
\usepackage[inline]{enumitem}
\usepackage{mathpartir}
\let\ifdraft\iffalse
\usepackage{style/notation}
\usepackage{style/appendix}

\additionalmaterial[The extended version~\cite{fullversion}]{\cite{fullversion}}

\usepackage[figuresleft]{rotating}
\ifappendix
  \setcopyright{none}
\else
  \setcopyright{rightsretained}
  \acmPrice{}
  \acmDOI{10.1145/3563298}
  \acmYear{2022}
  \copyrightyear{2022}
  \acmSubmissionID{oopslab22main-p226-p}
  \acmJournal{PACMPL}
  \acmVolume{6}
  \acmNumber{OOPSLA2}
  \acmArticle{135}
  \acmMonth{10}
\fi

\newenvironment{mathfig}[1][]{\adjustfigure[#1]\begin{figure}[tp]}{\end{figure}}

\begin{document}

\title{Proving Hypersafety Compositionally}
\ifappendix
  \subtitle{(Extended Version)}
\fi

\author{Emanuele D'Osualdo}
\email{dosualdo@mpi-sws.org}
\orcid{0000-0002-9179-5827}
\affiliation{\institution{MPI-SWS}
  \city{Saarland Informatics Campus}
  \country{Germany}
}
\author{Azadeh Farzan}
\orcid{0000-0001-9005-2653}
\email{azadeh@cs.toronto.edu}
\affiliation{\institution{University of Toronto}
  \city{Toronto}
  \country{Canada}
}
\author{Derek Dreyer}
\orcid{0000-0002-3884-6867}
\email{dreyer@mpi-sws.org}
\affiliation{\institution{MPI-SWS}
  \city{Saarland Informatics Campus}
  \country{Germany}
}

\begin{abstract}
  Hypersafety properties of arity~$n$ are program properties that relate~$n$ traces of a program (or, more generally, traces of~$n$ programs).
  Classic examples include determinism, idempotence, and associativity.
  A number of \emph{relational program logics} have been introduced
  to target this class of properties.
  Their aim is to construct simpler proofs by capitalizing on
  structural similarities between the~$n$ related programs.
We propose unexplored, complementary proof principles
  that establish hyper-triples (i.e.~hypersafety judgments)
  as a unifying compositional building block for proofs,
  and we use them to develop a
  \emph{Logic for Hyper-triple Composition}~(\thelogic),
  which supports forms of proof compositionality
  that were not achievable in previous relational logics.
We prove \thelogic{} sound and apply it to a number of challenging examples.
\end{abstract}

\ifnoappendix

\fi

\maketitle

\section{Introduction}
\label{sec:intro}

Many properties of interest about programs are properties not of individual program traces but rather of multiple program traces. 
For example, stipulating a bound on mean response time over all executions of a program cannot be specified as a property of individual traces, because the acceptability of delays in a trace depends on the magnitude of delays in all other traces. 
\citet{ClarksonS08} formally studied this class of properties, coining the term \emph{hyperproperties}. In this paper, we focus on the verification problem of a class of generalized hyperproperties, which we will refer to as \emph{\pre n-safety properties}:
these are \emph{safety} hyperproperties (i.e.~only concerned with partial correctness) that govern~$n$ executions of \emph{potentially different} programs.
More formally, these \pre n-safety properties have the form
$
  \forall (s_1,s_1') \in \sem{t_1}
  \dots
  \forall (s_n,s_n') \in \sem{t_n}\st
  \phi(s_1,s_1',\dots,s_n,s_n')
$,
where $\sem{t}$ denotes the set of input/output states of the traces of~$t$.
Examples of such properties are\footnote{Intuitively, for instance,
  associativity $ \p{f}(a,\p{f}(b,c)) = \p{f}(\p{f}(a,b),c) $
  compares 4 runs of \p{f}.
}
commutativity ($n=2$), associativity ($n=4$), determinism ($n=2$), noninterference ($n=2$) and transitivity ($n=3$).
Input-output equivalences between two programs can also be proved
as 2-safety properties (e.g.~a compiler optimization preserves I/O behaviour).

One approach to proving \pre n-safety properties is to
obtain from~$t$ a precise mathematical characterization~$R$
of its functionality~$\sem{t}$---i.e.~$t$'s strongest postcondition---and then prove using this mathematical characterization that the desired \pre n-safety property~$\phi$ is satisfied. This approach has a major drawback, however:
functional correctness is in general a much stronger---thus more difficult to prove---property of~$t$ than what proving~$\phi$ requires.

Thus, the trend in research on this problem has been instead towards reasoning \emph{directly} about \pre n-safety, through a number of so-called \emph{relational program logics}, e.g.~\cite{Benton04,BartheCK16,SousaD16,Yang07}. 
These logics move away from the traditional Hoare-triple judgment and introduce \pre n-ary relational Hoare-style judgments, which we will refer to as \emph{hyper-triples}.
In our syntax, a hyper-triple (of arity~$n$) is a judgment of the form
\[
  \J |- {P} {\m[\I 1: t_1,\dots,\I n: t_n]} {Q}
\]
which means
$
  \forall (s_1,s_1') \in \sem{t_1}
  \dots
  \forall (s_n,s_n') \in \sem{t_n}\st
    P(s_1,\dots,s_n) \implies Q(s_1',\dots,s_n')
$.
For example, determinism of~$t$ can be expressed as
$
  \J |- {\p{vars}_{\I1} = \p{vars}_{\I2}} {\m[\I 1: t, \I 2: t]} {\p{vars}_{\I1} = \p{vars}_{\I2}}
$, which asserts that if~$t$ is run from two states which agree on the values
of some relevant variables $\p{vars}$, then, if the two runs terminate,
the output states will agree on the values of~$\p{vars}$ as well.

The key idea underlying relational program logics is that,
by \emph{exploiting the similarity in the program structure}
of $t_1,\ldots,t_n$,
they avoid the need to specify and prove
functional specifications for these programs,
thus reducing the complexity of the work needed
to prove the \pre n-safety property.
In particular, a recurring motif of these logics is the idea of ``lockstep''
proof rules, which decompose all the programs in the hyper-triple at the same time in the same way according to a common structure. For example, a binary lockstep rule for sequential composition might look as follows:
\begin{proofrule}
  \infer*[lab=LockStepSeq]{
  \label{rule:bin-lockstep-seq}
    \J |- {P} {\m[\I 1: t_1, \I 2: t_2]} {P'}
    \\
    \J |- {P'} {\m[\I 1: t_1', \I 2: t_2']} {Q}
  }{
    \J |- {P} {\m[\I 1: (t_1\p;t_1'), \I 2: (t_2\p;t_2')]} {Q}
  }
\end{proofrule}
This lockstep rule is advantageous when the effects of~$t_1$ and~$t_2$
(and of~$t_1'$ and $t_2'$ respectively)
are correlated,
so that the intermediate assertion~$P'$ does not need to reveal
their individual effect, but only the relation between their effects.
In this fortunate case, the rule logically \emph{aligns} the two traces
and propagates simple relational assertions,
like $\p{vars}_{\I1} = \p{vars}_{\I2}$,
through the aligned pair of traces.
For instance, it allows one to prove determinism of~$(t\p;t')$
by proving determinism of~$t$ and of~$t'$.

On their own, however, lockstep rules are too rigid:
in general, such a perfect alignment might not be attainable.
They are also too fragile:
small syntactic differences between the two programs may make them inapplicable.
Much of the effort in previous work has thus focused on overcoming this rigidity of lockstep rules.
In particular, the state-of-the-art solution to the problem is to extend the logic with
rules that replace terms with other, semantically equivalent, terms
in a hyper-triple
(see \eg \cite{BartheCK11,BartheGHS17,SousaD16}).
The replacement can bring back the syntactic similarity needed to resume
a lockstep proof.
The resulting proof strategy,
which we call \emph{enhanced lockstep},
is to first rewrite the terms so that their structure reflects the
desired alignment, and then proceed with a lockstep derivation.

Despite its simplicity, enhanced lockstep reasoning has
produced relational program logics that
have already been deployed successfully in a variety of applications.
For example~\cite{BartheCK16,BartheGHS17} use them for translation validation tasks and even verification of cryptographic algorithms.
\citet{SousaD16} demonstrate how higher-arity hypersafety can encode
correctness requirements of comparators for user-defined data structures.

Here, however, we would like to call attention to
an orthogonal, unexplored, and very limiting dimension of rigidity in lockstep-based relational proofs.
Consider for instance a MapReduce library.
We ought to be able to prove that this library ensures the high-level guarantee of determinism for its operations, \emph{provided that} the user-supplied parameters to the library---\ie the lower-level operations on which it depends---are deterministic, associative and commutative.
Seeing as the latter properties are all hyper-safety properties,
it would be desirable to be able to take the heterogeneous collection of hyper-triples (of different arities!) representing the assumptions on the user-supplied parameters, and produce from them a determinism hyper-triple for the MapReduce implementation.

Unfortunately, existing relational logics do not support hyper-triple composition in this way. It is easy to observe this at a shallow level: these logics do not provide any rules for mixing and matching properties of different arities (i.e.\ using hyper-triples of one arity to prove hyper-triples of another arity). Even in CHL~\cite{SousaD16}, which is parametrized over the arity~$n$ of its hyper-triples, a fixed~$n$ must be used throughout an entire hyper-triple derivation. In order to compose hypersafety proofs as a library like MapReduce would require, we need more expressive ways of composing hyper-triples than what lockstep-style rules provide.
To see the issue concretely, consider the following example.

\begin{example}
\label{ex:commut-op}
\allowdisplaybreaks
Imagine a library provides a binary operation \p{op}$(a,b)$
that has been proven deterministic and commutative,
in the form of two hyper-triples:\footnote{Note that \code{op} can have non-deterministic side-effects on the state
  (determinism is restricted to the return value).
}
\begin{align}
  &\J|-
  {\True}
  {\m[
    \I 1: \code{r_1:=op}(a,b),
    \I 2: \code{r_2:=op}(a,b)
  ]}{
    \p{r}_1(\I 1)=\p{r}_2(\I 2)
  }
  \tag{\textsc{Det}$_{\p{op}}$}
  \label{spec:op-det}
  \\
  &\J|-
  {\True}
  {\m[
    \I 1: \code{r_1:=op}(a,b),
    \I 2: \code{r_2:=op}(b,a)
  ]}{
    \p{r}_1(\I 1)=\p{r}_2(\I 2)
  }
  \tag{\textsc{Comm}$_{\p{op}}$}
  \label{spec:op-comm}
\end{align}
In assertions, we write $\p x(i)$ to refer to the value
of the program variable~$\p x$ in the store at index~$i$.

Assuming \p{op} does not modify \p{x} and \p{y},
we ought to be able to use the hyper-triples above in a proof of the following:
\begin{equation}
  \J|-
  {\True}
  {\m<
    \I 1: \code{x:=op$(a$,$b)$;}
\code{z:=op$($x,x$)$},
    \I 2: \code{x:=op$(a$,$b)$;y:=op$(b$,$a)$;z:=op$($x,y$)$}
  >}{
    \p{z}(\I 1)=\p{z}(\I 2)
  }
  \tag{\textsc{Goal}$_{\p{op}}$}
  \label{ex:commut-op:goal}
\end{equation}
The unfortunate surprise is that there is no
alignment of the commands of the two components,
that can transform the goal into a lockstep proof
that reuses \eqref{spec:op-det} and \eqref{spec:op-comm}.
No matter which pairs of calls of \p{op} we align across the two components,
there will always be one call in component~\I2 that remains unmatched;
the abstract specification of \p{op} is relational and we have no hyper-triple
that specifies the effect of a single call.
Yet, the hyper-triple follows semantically from the assumptions.
\end{example}

We argue that enhanced lockstep reasoning is fundamentally incomplete,
and does not allow compositions of proofs that are needed for modular reasoning.
In this specific example,
even though the arities of the hyper-triples match,
we lack an appropriate way to compose them.
Intuitively, one should be able to prove that
the assignment to \p{x} in \I1 is equivalent to the first assignment in \I2,
\emph{and} to the second assignment in \I2, and put together these two facts
to conclude $\p{x}(\I2)=\p{y}(\I3)$.

In this paper, we ask more generally: what can and should composition of hypersafety proofs look like?
We provide an answer to this question
in the form of a new program logic,
\emph{Logic for Hyper-triple Composition (\thelogic)}.
\thelogic\ enables the reasoning about a given hypersafety property,
formally encoded as a hyper-triple,
to be composed from subproofs of hyper-triples of potentially different arities,
and with richer and more structurally complex means of composition
than are afforded by traditional lockstep reasoning.

The key observation behind the design of \thelogic\ is that
traditional relational rules are of two main varieties:
they either hold generically w.r.t.\ the hyper-terms involved (like the classic Hoare-logic consequence rule),
or they decompose the goal by matching on the structure of the related programs,
(like the lockstep sequence rule).
As we illustrate in \cref{sec:motivation},
\thelogic\ identifies a \emph{third} missing set of rules which decompose the goal
by matching on \emph{the structure of the hyper-term itself}.
That is, by viewing hyper-terms as maps from indices to terms, \thelogic's new rules allow splitting, joining, and re-indexing of these maps.
In so doing, they provide a new dimension of compositionality in hypersafety reasoning, making it possible for the first time to verify a range of interesting examples (such as \cref{ex:commut-op}) that were beyond the reach of prior relational logics.

The remainder of the paper is structured as follows:
\begin{itemize}
\item In \cref{sec:motivation},
  we illustrate the main new reasoning principles of \thelogic,
  by means of a series of examples that cannot be handled
  with previous relational logics.
\item In \cref{sec:prelim,sec:logic},
  we formalize the model and formulate the rules in full generality.\item In \cref{sec:discussion},
  we include a discussion of various features of \thelogic{} and contrast it against the previous state of the art.
\item In \cref{sec:related-work},
  we give a high-level survey of the literature in the problem space of hypersafety verification, and in \cref{sec:conclusion}, we discuss exciting future directions for this problem space.
\end{itemize}
Additional case studies, omitted details and the soundness proofs can be found in
\ifappendix Appendix.
\else the extended version of this paper~\cite{fullversion}.
\fi
 \section{Overview of \thelogic}
\label{sec:motivation}

In this section, we introduce the core new reasoning principles of \thelogic,
illustrating them through a series of simple examples.
None of these examples can be handled with existing relational logics.
More extensive case studies can be found in \appendixref{sec:case-studies}.

\subsection{A Weakest-Precondition-Based Calculus}
\label{sec:overview:wp}

Using \thelogic{},
we can overcome the limitations of pure lockstep reasoning---without relying on functional specifications for subprograms,
nor meta-level reasoning---by embracing hyper-triples as building blocks for proving other hyper-triples.
The first step to achieve this is to move from a calculus
based on triples precondition/hyper-term/postcondition,
to a calculus based on a weakest-precondition~(WP) predicate over hyper-terms.
This allows for a minimal and flexible logic.

More formally,\footnote{The notation used in this section will be made fully precise in \cref{sec:prelim,sec:logic}.}
our judgments are of the form~$ \V R |- P $ where both
$R$ and~$P$ are assertions over hyper-stores,
\ie maps from indices to variable stores.
The judgment asserts that $R \implies P$ holds on every hyper-store.
To state properties of hyper-terms, we introduce the WP assertion
$ \WP{\m{t}}{Q} $
that holds on a hyper-store~$\m{s}$ if running $\m{t}$ from~$\m{s}$
results in a hyper-store~$\m{s}'$ which satisfies the assertion~$Q$.\footnote{In general~$Q$ may also predicate over the return values of~$\m{t}$,
  but we ignore this for simplicity in this section.
}
A hyper-term is ``run'' from a hyper-store by running each of its components
on the store at the corresponding index.
On the indices where the hyper-term is undefined, the store is simply preserved.
Hyper-triples $ \T{P}{\m{t}}{Q} $ are notation for $ P \implies \WP{\m{t}}{Q} $.

Standard unary laws for weakest preconditions hold for our WP as well;
for example, the usual (unary) rule for sequence
$
  \WP {\m[\I1: t]}[\big]{
    \WP {\m[\I1: t']}{Q}
  }
  \lequiv
  \WP {\m*[\I1: (t\p;t')]} {Q}
$
\setlabel{\textsc{wp-seq$_1$}}\label{rule:wp-seq-1}
is valid in \thelogic.
It is easy to provide WP-based versions of lockstep rules.
For example, \ref{rule:bin-lockstep-seq} can be captured~by:
\begin{proofrule}
  \infer*[lab=wp-seq$_2$]{}{
    \WP {\m[\I1: t_1, \I2: t_2]}[\big]{
      \WP {\m[\I1: t_1', \I2: t_2']}{Q}
    }
    \lequiv
    \WP {\m*[\I1: (t_1\p;t_1'), \I2: (t_2\p;t_2')]} {Q}
  }
  \label{rule:wp-seq-2}
\end{proofrule}

The \ref{rule:wp-seq-2} rule helps us start a proof\/\footnote{Here and in following derivations,
  we omit applications of the rule of consequence.}
of the goal \eqref{ex:commut-op:goal} of \cref{ex:commut-op},
by aligning the components at the assignments to \p{z},
as in \cref{fig:example-align}.
The right-hand premise is an instance of \eqref{spec:op-det}.
To derive the left-hand premise, however,
we need genuinely new reasoning principles.

\begin{figure}[tb]
  \adjustfigure[]
  \begin{derivation}
    \infer*[right=\labelstep{ex:comm-op-seq}]{\V|- \WP{\m<
        \I 1: \code{x:=op$(a$,$b)$},
        \I 2: \code{x:=op$(a$,$b)$;y:=op$(b$,$a)$}
      >}*{
      \begin{conj*}
        \p{x}(\I 1)=\p{x}(\I 2)
        \and
        \p{y}(\I 2)=\p{x}(\I 1)
      \end{conj*}
      }
      \and \V
      \begin{conj*}
        \p{x}(\I 1)=\p{x}(\I 2)
        \and
        \p{y}(\I 2)=\p{x}(\I 1)
      \end{conj*}
      |-
      \WP{\m<
        \I 1: \code{z:=op$($x,x$)$},
        \I 2: \code{z:=op$($x,y$)$}
      >}[\big]{
        \p{z}(\I 1)=\p{z}(\I 2)
      }
    }{\V|- \WP{\m<
        \I 1: \code{x:=op$(a$,$b)$;}
              \code{z:=op$($x,x$)$},
        \I 2: \code{x:=op$(a$,$b)$;y:=op$(b$,$a)$;z:=op$($x,y$)$}
      >}[\big]{
        \p{z}(\I 1)=\p{z}(\I 2)
      }
    }
  \end{derivation}
  \caption{A first step in the proof of \cref{ex:commut-op},
    aligning the programs at the assignment to \p{z}.
    The right-hand premise is an instance of the determinism assumption
    \eqref{spec:op-det}.}
  \label{fig:example-align}
\end{figure}

\subsection{Conjunction and Nesting}
\label{sec:overview:conj-nest}

What characterises lockstep rules is that they match
on the structure of the \emph{program terms}
in multiple components simultaneously.
The new rules introduced by \thelogic\ focus instead on the structure of the
\emph{hyper-term} as a map from indices to terms.
The question we ask is: which combinations of WPs are induced by
operations on hyper-terms as maps?
We start with the ones we need
to finish the proof of \cref{ex:commut-op}:
two rules induced by the union operation.
The union~$ \m{t}_1 \m+ \m{t}_2 $ of two hyper-terms $\m{t}_1$ and $\m{t}_2$
is a hyper-term if $\m{t}_1$ and $\m{t}_2$ coincide,
on the indices they have in common.
For example~$ \m[\I1: t_1, \I2: t_2] \m+ \m[\I2: t_2', \I3: t_3] $
is well-defined if $t_2 = t_2'$.
When the two hyper-terms do not have indices in common,
the result is disjoint union, which we write~$ \m{t}_1 \m. \m{t}_2 $.

Now we can ask,
how can we combine two WPs to obtain a WP on the union of their hyper-terms?
\thelogic's answer is the \ref{rule:wp-conj-simpl} rule:
\begin{proofrule}
\infer*[lab=wp-conj$_0$]{
    \idx(Q_1) \subs \supp(\m{t}_1)
    \\
    \idx(Q_2) \subs \supp(\m{t}_2)
  }{
    \V
    \WP{\m{t}_1}{Q_1}
    \land
    \WP{\m{t}_2}{Q_2}
    |-
    \WP{(\m{t}_1 \m+ \m{t}_2)}{Q_1 \land Q_2}
  }
  \label{rule:wp-conj-simpl}
\end{proofrule}
Read from right to left,
the rule states that,
to prove that we are in a hyper-state from which
$\m{t}_1 \m+ \m{t}_2$ takes us to $Q_1 \land Q_2$,
it is sufficient to prove separately that~$\m{t}_1$ takes us to~$Q_1$
and~$\m{t}_2$ to~$Q_2$.
This is sound as long as the postconditions only predicate over indices pertaining to their corresponding hyper-term, as mandated by the two premises.

\thelogic\ proposes a second way of decomposing a WP by seeing
its hyper-term as a \emph{disjoint} union,
with the \ref{rule:wp-nest-simpl} rule:
\begin{proofrule}
\infer*[lab=wp-nest$_0$]{}{
    \WP{\m{t}_1}{\WP{\m{t}_2}{Q}}
    \lequiv
    \WP{(\m{t}_1 \m. \m{t}_2)}{Q}
  }
  \label{rule:wp-nest-simpl}
\end{proofrule}
The rule establishes an equivalence between two nested WP on disjoint
hyper-terms, and a single one on their union.
The idea is that since the semantics of WP preserves the stores at indices
not belonging to the WP's hyper-term,
the inner WP on the left receives as input
on the indices of~$\m{t}_2$ the initial stores,
just like on the right;
and it preserves the outputs of~$\m{t}_1$.
The postcondition~$Q$ is free to predicate on the indices of both hyper-terms,
and applies to the same hyper-stores on both sides of the equivalence.

This rule unlocks a very useful proof pattern.
It is very common to have a goal with many components,
but wanting to temporarily focus the proof on a subset of the components.
The \ref{rule:wp-nest-simpl} rule, applied from right to left,
would correspond to ``shelving'' the components in $\m{t}_2$,
allowing the proof to operate on $\m{t}_1$, \eg in lockstep.
When $\m{t}_2$ becomes relevant again, applying the rule from left to right
would ``unshelve'' the components for the rest of the proof.

\begin{sidewaysfigure}
\[
    \infer*[right=\ref{rule:wp-nest-simpl}]{
    \infer*[Right=\ref{rule:wp-seq-1}]{
    \infer*[Right=\ref{rule:wp-nest-simpl}]{
      \infer*[Right=\ref{rule:wp-const}]{
        \infer*[Right=\ref{rule:wp-conj-simpl}]{
          \infer*{}{
          \V |-
          \WP{\m<
            \I1: \code{x:=op$(a$,$b)$},
            \I2: \code{x:=op$(a$,$b)$}
          >}[\big]{
            \p{x}(\I1) = \p{x}(\I2)
          }
          }
\and
\infer*[Right=\ref{rule:wp-nest-simpl}]{
            \infer*[Right=\ref{rule:wp-cons}]{
              \V |-
              \WP{\m<
                \I1: \code{x:=op$(a$,$b)$},
                \I2: \code{y:=op$(b$,$a)$}
              >}[\big]{
                \p{y}(\I2) = \p{x}(\I1)
              }
            }{\V |-
              \WP{\m[
                \I2: \code{x:=op$(a$,$b)$}
              ]}*{
                \WP{\m<
                  \I1: \code{x:=op$(a$,$b)$},
                  \I2: \code{y:=op$(b$,$a)$}
                >}{
                  \p{y}(\I2) = \p{x}(\I1)
                }
              }
            }
          }{\V |-
            \WP{\m<
              \I1: \code{x:=op$(a$,$b)$},
              \I2: \code{x:=op$(a$,$b)$}
            >}[\big]{
              \WP{\m[
                \I2: \code{y:=op$(b$,$a)$}
              ]}{
                \p{y}(\I2) = \p{x}(\I1)
              }
            }
          }
        }{\V |-
          \WP{\m<
            \I1: \code{x:=op$(a$,$b)$},
            \I2: \code{x:=op$(a$,$b)$}
          >}*{
            \begin{conj*}
              \p{x}(\I1) = \p{x}(\I2) \land{} \\
              \WP{\m[
                \I2: \code{y:=op$(b$,$a)$}
              ]}*{
                \p{y}(\I2) = \p{x}(\I1)
              }
            \end{conj*}
          }
        }
      }{\V |-
        \WP{\m<
          \I1: \code{x:=op$(a$,$b)$},
          \I2: \code{x:=op$(a$,$b)$}
        >}*{
          \WP{\m[
            \I2: \code{y:=op$(b$,$a)$}
          ]}*{
          \begin{conj*}
            \p{x}(\I1) = \p{x}(\I2) \\
            \p{y}(\I2) = \p{x}(\I1)
          \end{conj*}
          }
        }
      }
      }{\V |-
        \WP{\m[\I1: \code{x:=op$(a$,$b)$}]}*{
          \WP{\m[\I2: \code{x:=op$(a$,$b)$}]}*{
            \WP{\m[\I2: \code{y:=op$(b$,$a)$}]}*{
\p{x}(\I1) = \p{x}(\I2) \land
                \p{y}(\I2) = \p{x}(\I1)
}
          }
        }
      }
      }{\V |-
        \WP{\m[\I1: \code{x:=op$(a$,$b)$}]}*{
          \WP{\m[\I2: \code{x:=op$(a$,$b)$;y:=op$(b$,$a)$}]}*{
\p{x}(\I1) = \p{x}(\I2) \land
            \p{y}(\I2) = \p{x}(\I1)
}
        }
      }
    }{\V |-
      \WP{\m<
        \I1: \code{x:=op$(a$,$b)$},
        \I2: \code{x:=op$(a$,$b)$;y:=op$(b$,$a)$}
      >}*{
      \begin{conj*}
        \p{x}(\I1) = \p{x}(\I2) \\
        \p{y}(\I2) = \p{x}(\I1)
      \end{conj*}
      }
    }
  \]
  \captionsetup{width=.7\textwidth}
  \caption{\thelogic{} proof of \cref{ex:commut-op}:
    derivation of left-hand premise of step~\eqref{ex:comm-op-seq}.
    The leaf on the left is discharged by the assumption of determinism
    of \p{op} \eqref{spec:op-det}.
    The one on the right is discharged by the assumption of commutativity
    of \p{op} \eqref{spec:op-comm}.
  }
  \label{fig:commut-op-proof}
\end{sidewaysfigure}

The seemingly simple \ref{rule:wp-nest-simpl} and \ref{rule:wp-conj-simpl},
are powerful enough to close the proof of \cref{ex:commut-op},
if combined with the standard principles of consequence and frame.
\Cref{fig:commut-op-proof} shows the \thelogic\ proof of the
left-hand premise of step~\eqref{ex:comm-op-seq} in \cref{fig:example-align}.
The derivation starts with an application of \ref{rule:wp-nest-simpl}
which allows us to apply the (unary) sequence rule \ref{rule:wp-seq-1}
only on~\I2.
Another application of \ref{rule:wp-nest-simpl} obtains the overall effect
of shelving the assignment to \p{y} in \I2 for later,
aligning the two assignments to \p{x}.
Then, the application of the rule of frame
(see \ref{rule:wp-const} in \cref{fig:wp-struct-laws}),
borrowed from Hoare logic,
is justified by the assumption that~\p{op} does not modify~\p{x}.
This means that if we prove $\p{x}(\I1) = \p{x}(\I2)$ in the postcondition
of the outer WP,
it would hold after running $\m[\I2: \code{y:=op$(b$,$a)$} ]$ too.

Since the postcondition is now a conjunction,
we can apply the \ref{rule:wp-conj-simpl}.
Note that the two hyper-terms in the premises overlap
in \emph{all} the components.
The first premise coincides with our determinism assumption \eqref{spec:op-det}.
For the second premise, we use \ref{rule:wp-nest-simpl} twice:
once to push component~\I1 in the postcondition,
and another time to fuse it with the nested~\I2.
We obtain a top-level WP with a postcondition that
coincides with our commutativity assumption \eqref{spec:op-comm}.
The final application of consequence (\ref{rule:wp-cons})
amounts to say that if the WP in the premise holds on every state,
then it should hold on the state resulting
from running $\m[\I2: \code{x:=op$(a$,$b)$}]$.

The combination of \ref{rule:wp-conj-simpl} and \ref{rule:wp-nest-simpl}
allowed us to carry out the intuitive proof strategy,
which relates both assignments of~\p{x} and~\p{y} in~\I2 with the single
assignment to~\p{x} in~\I1.

\subsection{Projection}
\label{sec:overview:proj}

The rules in the previous section were induced by union of hyper-terms.
We next present a rule that is induced by the operation of
removing components from a hyper-term.
To motivate the rule, we present another simple example
that lockstep reasoning cannot handle.

Imagine we are given
some deterministic terminating\footnote{We will elaborate on the assumption of termination later.}
$t_1$ and $t_2$ satisfying
the specification
\begin{equation}
  \J |- {\pv{x}(\I1)=\pv{x}(\I2)}
        {\m[\I1: t_1\p;t_2, \I2: t_2\p;t_1]}
        {\pv{x}(\I1)=\pv{x}(\I2)}
  \tag{\textsc{Swap}$_{t_1,t_2}$}
  \label{ex:swap-t}
\end{equation}
Given a vector of pairwise distinct program variables~$
 \pv{x} = \p x_1 \dots \p x_n
$ we write $ \pv{x}(i) $ for the vector
$ \p x_1(i) \dots \p x_n(i) $;
the assertion $\pv{x}(\I1)=\pv{x}(\I2)$
states pointwise equality between the two vectors.
The specification above then states that,
relative to some relevant variables~\pv{x},
sequencing the two commands~$t_1$ and $t_2$
in one order or the other generates the same result.
It should be possible to derive, within the logic, that
\[
  \J |- {\pv{x}(\I1)=\pv{x}(\I2)}
        {\m<\I1: t_1\p;t_2\p;t_2, \I2: t_2\p;t_2\p;t_1>}
        {\pv{x}(\I1)=\pv{x}(\I2)}
\]
It is however not possible to prove this by aligning the two programs.
The only alignment that generates a match with the assumption~\eqref{ex:swap-t}
would leave an unmatched~$t_2$ at the start of \I2 and one at the end of \I1.

The natural proof strategy here is to introduce an auxiliary term
$ t_2\p;t_1\p;t_2 $, and show that both components~\I1 and \I2 above
are equivalent to it.
In derivation form, the desired proof looks as follows:
\begin{derivation}
  \infer*[right=\ref{rule:wp-proj-simpl}]{
  \infer*[Right=\ref{rule:wp-conj-simpl}]{
    \J |- {\pv{x}(\I1)=\pv{x}(\I3)}
          {\m<\I1: t_1\p;t_2\p;t_2, \I3: t_2\p;t_1\p;t_2>}
          {\pv{x}(\I1)=\pv{x}(\I3)}
    \and
    \J |- {\pv{x}(\I2)=\pv{x}(\I3)}
          {\m<\I3: t_2\p;t_1\p;t_2, \I2: t_2\p;t_2\p;t_1>}
          {\pv{x}(\I2)=\pv{x}(\I3)}
  }{
    \J |- {\pv{x}(\I1) = \pv{x}(\I3) = \pv{x}(\I2)}
          {\m< \I1: t_1\p;t_2\p;t_2
             , \I3: t_2\p;t_1\p;t_2
             , \I2: t_2\p;t_2\p;t_1
             >}
          {\pv{x}(\I1) = \pv{x}(\I3) = \pv{x}(\I2)}
  }}{
    \J |- {\pv{x}(\I1)=\pv{x}(\I2)}
          {\m<\I1: t_1\p;t_2\p;t_2, \I2: t_2\p;t_2\p;t_1>}
          {\pv{x}(\I1)=\pv{x}(\I2)}
  }
\end{derivation}Read from top to bottom,
the two premises express the equivalence between the original components
and the auxiliary component~\I3.
They are easily proved by the obvious lockstep proof,
using \eqref{ex:swap-t} and the assumptions of determinism (as hyper-triples).

The crucial last step removes component~\I3
from the hyper-triple,
an operation that we call ``projecting out''.
Note that this step is not just a by-product of the principle of consequence:
the precondition of the goal does not imply the precondition of the premise,
which additionally constrains the store at~\I3.
To justify such a step, \thelogic\ provides the following \emph{projection} rule:
\begin{proofrule}
  \infer*[lab=wp-proj$_i$, right=$ \proj(t) $]{
    \V P |- \WP{(\m{t}'\m.\m[i: t])}{Q}
  }{
    \V \P i.P |- \WP {\m{t}'} {\P i. Q}
  }
  \label{rule:wp-proj-simpl}
\end{proofrule}
The rule introduces a novelty in the assertion language,
the projection modality~$\P i.P$.
A hyper-store~$\m{s}$ satisfies~$ \P i.P $
if there is some store~$s'$ that can be placed at~$i$ so that the resulting
hyper-store~$\m{s}\m[i: s']$ satisfies~$P$.
Intuitively, $ \P i.P $ removes the constraints imposed by $P$
on the store at index~$i$,
while keeping the explicit and implied constraints
imposed on the other components.
For example,
$
  \P\I2.(\p{x}(\I1) < \p{x}(\I2) < \p{x}(\I3))
$ is logically equivalent to $
  \p{x}(\I1) + 1 < \p{x}(\I3)
$:
if an hyper-store satisfies $\p{x}(\I1) + 1 < \p{x}(\I3)$
we can replace the store at~\I2 with one satisfying
$\p{x}(\I2)= \p{x}(\I1) + 1$ and obtain a hyper-store
satisfying $\p{x}(\I1) < \p{x}(\I2) < \p{x}(\I3)$.

Since the hyper-term in the conclusion of the rule has had its
$i$-th component removed, both the precondition and postcondition of
the conclusion are subject to the $\Pi_i$ projection.
The rule has an important side condition, however:
$\proj(t)$, which requires~$t$ to have terminating traces from any input store
(in \cref{sec:logic} we relax this condition).
In our example this is discharged by the assumption that~$t_1$ and~$t_2$ are terminating.
Perhaps surprisingly, the rule without side conditions is unsound:
since the semantics of WP only constrains the terminating runs,
if~$t$ does not terminate, $Q$ can be $\False$ in the premise,
which would result in an invalid triple in the conclusion.

The combination of projection and conjunction that we used in the example
corresponds roughly to using transitivity in a refinement-based proof.
It is striking that previous relational logics for hypersafety do not admit
this principle.
A notable exception is RHL~\cite{Benton04} which has a ``transitivity'' rule
for \pre2-properties that could handle our example.
The soundness of the rule relies however on a semantics of triples
that goes beyond hypersafety and constrains the non-terminating behaviour too,
in a way that is incompatible with some of \thelogic's rules
(\eg \ref{rule:wp-nest-simpl}).
We discuss the trade-off involved in \cref{sec:discuss-design}.

\subsection{Reindexing and Indirect-Style Hyper-triples}
\label{sec:overview:reindex}

A third way in which lockstep reasoning is unable to
reuse hyper-triples,
is the case of what we call ``indirect-style'' specifications.
Consider the case of idempotence,
that is, that~$t$ and $t\p;t$ achieve the same effect.
A naive encoding of idempotence of~$t$
with respect to some relevant variables \pv{x},
is:\begin{equation}
  \J |- { \pv{x}(\I1) = \pv{x}(\I2) }
        {\m[\I 1: t, \I 2: (t\p;t)]}
        { \pv{x}(\I1)=\pv{x}(\I2) }
  \tag{\textsc{IdemSeq}$_{t}$}
  \label{spec:t-idemp-seq}
\end{equation}

In \cite{SousaD16}, however, idempotence is specified instead as follows:
\begin{equation}
  \J |- { \pv{x}(\I2)=\vec{v} }
        {\m[\I 1: t, \I 2: t]}
        { \pv{x}(\I1)=\vec{v} \implies \pv{x}(\I2)=\vec{v} }
  \tag{\textsc{Idem}$_{t}$}
  \label{spec:t-idemp}
\end{equation}The input of~\I1 is unconstrained; the input of~\I2 is assumed to be~$\vec{v}$;
after a run of~$t$ in each component, the postcondition asserts that,
if the input of~\I2 happened to coincide with the output of \I1,
then the output of~\I2 is the same as the output of \I1.
The specification uses a slightly unintuitive pattern,
which we call ``indirect-style'':
the output of~\I1 is fed as input to~\I2 indirectly,
through an implication in the postcondition.

Let us compare the two specifications.
The hyper-triple \eqref{spec:t-idemp-seq} is certainly easier to interpret;
it also is in a form that seems readily applicable:
one would expect to use idempotence precisely to relate the two components
of the hyper-triple.
In contrast, the hyper-triple \eqref{spec:t-idemp} is not directly applicable
in lockstep proofs that need to relate~$t$ and~$t\p;t$.
It has the advantage, however, of making the two components syntactically coincide, which can be exploited by a tool to produce a simple lockstep idempotency proof, if one exists.
On the other hand, attacking directly a proof of \eqref{spec:t-idemp-seq}
might require one to characterize very precisely the effect of the first run of~$t$
so that the second one can be shown equivalent to a \code{skip}.

A difference which might not be immediate is that
\eqref{spec:t-idemp-seq} holds in fact for strictly fewer programs than
\eqref{spec:t-idemp}.
For an example, consider \code{if x=0 then (if $\;*\;$ then x:=1 else x:=2)}:
the result of a run from a store with $\p{x}=0$ would output a store with
a random value~$\p{x} \in \set{1,2}$;
executing the same program again would preserve the store,
therefore the program satisfies~\eqref{spec:t-idemp}.
It does not however satisfy~\eqref{spec:t-idemp-seq},
as the programs at~\I{1} and~\I{2}
may store different values in \p{x}.

The two encodings of idempotence are, however, equivalent for \emph{deterministic} programs.
Given the tension between the two---one is easier to prove but harder to use and vice versa---we would like to be able to derive one from the other,
\emph{within} the logic.
In CHL it is impossible to use \eqref{spec:t-idemp} and the hyper-triple encoding determinism of~$t$ to prove \eqref{spec:t-idemp-seq}.
This is, again, because the logical alignment proof strategy is fundamentally limited.

In order to support these derivations within the logic, \thelogic\ provides one final way of manipulating the hyper-term structure:
\emph{reindexing}, \ie the ability to substitute indices for indices in hyper-triples.

To illustrate the idea,
let us derive \eqref{spec:t-idemp-seq} from \eqref{spec:t-idemp} plus determinism in \thelogic.
There are two main obstacles to overcome:
the first is that \eqref{spec:t-idemp-seq} has three occurrences of~$t$
but \eqref{spec:t-idemp} has only two;
the second, and more relevant, is that the two runs of~$t$
in \eqref{spec:t-idemp-seq} that correspond to the two runs in
\eqref{spec:t-idemp} are \emph{sequenced} in the same component \I2.

The first step of the proof can be dealt with using the rules we have already introduced.
We want to appeal to determinism to establish that,
in $\m[\I 1: t, \I 2: (t\p;t)]$,
whatever we can say about the output of the first run of~$t$ in~\I2,
will hold for the output at \I1 too,
and reduce the goal to proving a triple only involving component \I2.
This can be achieved with the help of \ref{rule:wp-conj-simpl}:
\begin{derivation}[\normalsize]
  \infer*[right=\labelstep{ex:idemp:step-seq}\qquad]{
  \infer*[Right=\labelstep{ex:idemp:step-decouple}]{
  \infer*[Right=\ref{rule:wp-conj-simpl}]{
    \V
    \pv{x}(\I1) = \pv{x}(\I2)
    |-
    \WP{\m[\I 1: t, \I 2: t]}{ \pv{x}(\I1)=\pv{x}(\I2) }
    \and \V |-
    \WP{\m[\I 2: t]}{
        \E \vec{v}.
          \pv{x}(\I2) = \vec{v}
          \land
          \WP{\m[\I 2: t]}{\pv{x}(\I2) = \vec{v}}
    }
  }{
    \V
    \pv{x}(\I1) = \pv{x}(\I2)
    |-
    \WP{\m[\I 1: t, \I 2: t]}*{
      \pv{x}(\I1)=\pv{x}(\I2)
      \land
      \E \vec{v}.(
        \pv{x}(\I2) = \vec{v}
        \land
        \WP{\m[\I 2: t]}{\pv{x}(\I2) = \vec{v}}
      )
    }
  }}{
    \V
    \pv{x}(\I1) = \pv{x}(\I2)
    |-
    \WP{\m[\I 1: t, \I 2: t]}{
      \WP{\m[\I 2: t]}
         { \pv{x}(\I1)=\pv{x}(\I2) }
    }
  }}{
    \V
    \pv{x}(\I1) = \pv{x}(\I2)
    |-
    \WP{\m[\I 1: t, \I 2: (t\p;t)]}
       { \pv{x}(\I1)=\pv{x}(\I2) }
  }
\end{derivation}
The derivation starts just like \cref{fig:commut-op-proof}:
component~\I2 is decomposed at the sequential composition,
with the help of \ref{rule:wp-nest-simpl} and \ref{rule:wp-seq-1}.
Then, again similarly to \cref{fig:commut-op-proof},
consequence and frame are used to decouple the assertions on~\I1 and \I2
in the postcondition:
we can prove (by determinism) that~$\p{x}(\I1)=\p{x}(\I2)$ holds
already after the first runs of~$t$ in \I1 and \I2;
and separately we can prove the second run in \I2 will preserve
whatever value is stored in $\p{x}(\I2)$.
The \ref{rule:wp-conj-simpl} rule allows us to discharge the first conjunct
of the postcondition by appealing to the assumption of determinism of~$t$.

We are now left with our goal being the right-hand premise,
which has only two occurrences of~$t$,
although they are both at~\I2.
This identifies the crux of the problem:
our new goal involves two \emph{sequenced} runs of the same term $t$,
but the indirect-style hyper-triple we take as our assumption
(\ie \eqref{spec:t-idemp}) relates two runs in separate indices of a hyper-triple.

To close the proof we need two new rules of \thelogic\ that
deal precisely with re-assigning indices to components in a WP.
\thelogic\ uses the \emph{reindexing} notation $P\isub{j->i}$
to denote the assertion that is true on a hyper-store~$\m{s}$
if~$P$ is true on the hyper-store $ \m{s}\m[j: \m{s}(i)] $.
For a more syntactic intuition, $P\isub{j->i}$ is the
assertion~$P$ where
the occurrences of the index~$j$ are replaced with~$i$.
\thelogic\ proposes two main rules (simplified here slightly for clarity)
explaining the interaction between reindexing and WP:
\begin{proofrules}
  \infer*[lab=wp-idx-post$_0$,
    right={\footnotesize$j \notin \supp(\m{t})$}
  ]{
    \V |- \WP{\m{t}}{Q}
  }{
    \V |- \WP{\m{t}}{Q\isub{j->i}}
  }
  \label{rule:wp-idx-post-simpl}

  \infer*[lab=wp-idx-swap$_0$,
    right={\footnotesize$i \notin \idx(Q)$}
  ]{ }{
    \V
    \bigl(\WP{\m[j: t]}{Q}\bigr)\isub{j->i}
    |-
    \WP{\m[i: t]}{Q\isub{j->i}}
  }
  \label{rule:wp-idx-swap-simpl}
\end{proofrules}
\Cref{rule:wp-idx-post-simpl} states that it is possible
to substitute~$j$ for~$i$ in the postcondition of a WP
if $\m{t}$ does not have a component at index~$j$.
\Cref{rule:wp-idx-swap-simpl} shows the effect of applying the substitution
to a WP that does contain a component at~$j$.
(We explain the need for the side conditions on these rules in \cref{sec:logic}.)

With these two rules we can now complete the derivation:
\begin{derivation}[\normalsize]
\infer*[Right=\labelstep{ex:idemp-step-apply23}]{
\infer*[Right=\ref{rule:wp-idx-swap-simpl}]{
\infer*[Right=\labelstep{ex:idemp-step-ex2v}]{
\infer*[Right=\labelstep{ex:idemp-step-prop23}]{
\infer*[Right=\ref{rule:wp-idx-post-simpl}]{
\infer*[Right=\labelstep{ex:idemp-step-massage}]{
  \V
  \pv{x}(\I3)=\vec{v}
  |-
  \WP{\m[\I 2: t, \I 3: t]}
     { \pv{x}(\I2)=\vec{v} \implies \pv{x}(\I3)=\vec{v} }
}{
  \V |-
  \WP{\m[\I 2: t]}{
    \A \vec{v}.
      (\pv{x}(\I2) = \vec{v}
      \land
      \pv{x}(\I3) = \vec{v})
      \implies
      \WP{\m[\I 3: t]}{\pv{x}(\I3) = \vec{v}}
  }
}}{
  \V |-
  \WP{\m[\I 2: t]}*{
    \bigl(
      \A \vec{v}.
        (\pv{x}(\I2) = \vec{v}
        \land
        \pv{x}(\I3) = \vec{v})
        \implies
        \WP{\m[\I 3: t]}{\pv{x}(\I3) = \vec{v}}
    \bigr)
    \isub{\I3->\I2}
  }
}}{
  \V |-
  \WP{\m[\I 2: t]}*{
      \A \vec{v}.
        \pv{x}(\I2) = \vec{v}
        \implies
    \bigl(
        \WP{\m[\I 3: t]}{\pv{x}(\I3) = \vec{v}}
    \bigr)
    \isub{\I3->\I2}
  }
}}{
  \V |-
  \WP{\m[\I 2: t]}*{
    \E \vec{v}.
    \pv{x}(\I2) = \vec{v}
    \land
      \bigl(\WP{\m[\I 3: t]}{\pv{x}(\I3) = \vec{v}}\bigr)
      \isub{\I3->\I2}
  }
}}{
  \V |-
  \WP{\m[\I 2: t]}*{
    \E \vec{v}.
    \pv{x}(\I2) = \vec{v}
    \land
      \WP{\m[\I 2: t]}{(\pv{x}(\I3) = \vec{v})\isub{\I3->\I2}}
  }
}}{
  \V |-
  \WP{\m[\I 2: t]}*{
    \E \vec{v}.
    \pv{x}(\I2) = \vec{v}
    \land
      \WP{\m[\I 2: t]}{\pv{x}(\I2) = \vec{v}}
  }
}
\end{derivation}
We start at the top with \eqref{spec:t-idemp};
we renamed~\I1 to~\I2 and \I2 to \I3,
which is possible as
\thelogic's specifications are closed under unambiguous renaming of indices.\footnote{Formally this is handled by rules \ref{rule:reindex} and \ref{rule:wp-idx},
  presented in \cref{sec:logic}.}
Step~\eqref{ex:idemp-step-massage} uses \ref{rule:wp-nest-simpl}
and some simple logical manipulations
(mainly commutation laws between implication and WP)
to put the postcondition in a suitable form.
An application of \ref{rule:wp-idx-post-simpl} applies the $\isub{\I3->\I2}$
substitution to the postcondition,
which is propagated using consequence in step~\eqref{ex:idemp-step-prop23},
until it reaches the inner WP.
Step~\eqref{ex:idemp-step-ex2v} uses consequence and the tautology
$ \E \vec{v}.\pv{x}(\I2) = \vec{v} $ to apply the universal quantification
in the postcondition.
We then use \ref{rule:wp-idx-swap-simpl} to propagate the reindexing
to the postcondition of the inner WP.
Step \eqref{ex:idemp-step-apply23} finally applies the reindexing to the
postcondition of the inner WP obtaining the desired judgment.

\subsection{Loop Invariants}
\label{sec:overview:loops}

The treatment of loops is a thorny issue in relational logics.
Loops are where all the shortcomings of the alignment proof strategy
get amplified.
Consider the following simple consequence of idempotence and determinism of~$t$
(picking \p{i} to be a fresh variable):
\begin{equation}
  \J |- { \pv{x}(\I1) = \pv{x}(\I2) }
        {\m[\I 1: t, \I 2: (t\code{;while i>0 do ($t$;i--)})]}
        { \pv{x}(\I1)=\pv{x}(\I2) }
  \tag{\textsc{IdemLoop}$_{t}$}
  \label{spec:t-idemp-loop}
\end{equation}
The idea of lockstep proofs underlies all the relational rules for loops in the literature: they work well when all the components in a hyper-triple are loops
which can be aligned at the boundaries of their bodies.
Since the relational specifications of idempotence and determinism of~$t$
do not give us direct information about single runs of~$t$
(only of \emph{pairs} of runs)
this strategy immediately falls apart
when attempting a proof of \eqref{spec:t-idemp-loop}.
We only have one component with a loop,
containing runs of~$t$ which we fail to align to any other runs of~$t$ in the other component.

Let us inspect the issue more closely, by sketching a derivation:
\begin{equation*}
\infer{
  \J |- { \pv{x}(\I1) = \pv{x}(\I2) }
        {\m[\I 1: t, \I 2: t]}
        { P }
  \\
  \infer*{
  \J |- { P }
        {\m[\I 2: (\code{$t$;i--})]}
        { P }
  }{
  \J |- { P }
        {\m[\I 2: (\code{while i>0 do ($t$;i--)})]}
        { P }
  }
}{
  \J |- { \pv{x}(\I1) = \pv{x}(\I2) }
        {\m[\I 1: t, \I 2: (t\code{;while i>0 do ($t$;i--)})]}
        { \pv{x}(\I1)=\pv{x}(\I2) }
}
\end{equation*}
We start by considering the first runs of~$t$ in \I1 and \I2 in lockstep,
and separately the loop, from the resulting state.
As in a standard proof, proving the loop boils down to
finding a suitable loop invariant.
In all the previous relational logics the most precise information
we can obtain for the left-hand premise, reusing the specifications of~$t$,
is $P=\bigl(\pv{x}(\I1) = \pv{x}(\I2)\bigr)$.
Unfortunately, the right-hand premise cannot be proven with this~$P$,
since~$t$ does not in general preserve it:
the assertion is too weak in that it does not record the fact that
a run of~$t$ already happened before the loop starts.
Short of this information, the verification of the loop needs to consider the
case where the body executes~$t$ for the first time ever, which in general
would not preserve the equivalence
(as there is no corresponding run of~$t$ in \I1).

This is where the extended expressivity of \thelogic\ is crucial
to obtain a modular and relational proof.
In \thelogic\ we can set
$
  P =
    \E \vec{v}.
    \bigl(
      \pv{x}(\I1) = \pv{x}(\I2) = \vec{v}
      \land
      \WP{\m[\I 2: t]}{\pv{x}(\I2) = \vec{v}}
    \bigr)
$
which correctly records the fact that~$t$ has run already,
by asserting that a further run of~$t$ at~\I2 would not modify~$\pv{x}$.
The left-hand premise can be established with this~$P$
by reusing the derivation of the previous section---it is a direct consequence of the premise
of step~\eqref{ex:idemp:step-decouple}.
The verification of the body of the loop can now use the WP in~$P$
to justify why~$t$ has no effect on~$\pv{x}$.
The full derivation,
shown for a variant of this example in \appendixref{sec:ex-loop-hoisting},
crucially relies on \thelogic's ability to handle nested WPs.

Through this series of examples we introduced the main novel
reasoning principles of \thelogic,
and showed how they enable compositional, modular proofs,
embracing hyper-triples as the fundamental building block of
hypersafety proofs.
 \section{Preliminaries}
\label{sec:prelim}

In this section we fix notation,
and define assertions over hyper-stores and their basic laws.

\begin{definition}[Finite maps]
\label{def:finmaps}
  Given a type $A$, we define the type $\opt{A} \is A \dunion \set{\bot}$.
Given a function $ f \from A \to \opt{B} $ we define $\supp(f) \is \set{a : A | \isdef{f(a)} }$. We say $f$ is a \emph{finite map from $A$ to $B$},
  written $f \from A \pto B$,
  if $f \from A \to \opt{B}$ and $\supp(f)$ is finite.

  We write $\m[a_1: b_1, \dots, a_n: b_n]$ for the finite map associating each $a_i$ to $b_i$ (and everything else to~$\bot$).
  Similarly to set comprehensions, we use the notation
  $ \m[a: b | \phi(a,b)] $,
  e.g.~$\m[i: j | i\in \Nat, i = 2j \leq 4 ] = \m[0: 0, 2: 1, 4: 2]$.
Given $f,g \from A \pto B$ such that
  $ \forall x \in \supp(f) \inters \supp(g) \st f(x) = g(x) $,
  the union of~$f$ and~$g$, written
  $ f \m+ g \from A \pto B $ is defined as:
  \[
    (f \m+ g)(x) =
      \begin{cases}
        f(x) \CASE x \in \supp(f) \setminus \supp(g) \\
        g(x) \CASE x \in \supp(g) \\
        \bot \OTHERWISE
      \end{cases}
  \]
  We leave $f \m+ g$ undefined if $f(x) \neq g(x) $ for some
  $ x \in \supp(f) \inters \supp(g)$.
The \emph{disjoint union} of~$f$ and~$g$,
  written $ f \m. g  \from A \pto B $,
  is defined as $f \m. g \is f \m+ g$ if ${\supp(f) \inters \supp(g) = \emptyset}$,
  undefined otherwise.
For any function $f \from A \to B$, we define
  $f\m[a: b] \is \fun x.\kw{if}\ x=a\ \kw{then}\ b\ \kw{else}\ f(x)$.
\end{definition}

Both operations $(\m+)$ and $(\m.)$ on maps are commutative and associative
(where defined) and have the empty map $\m[]$ as neutral element.
Indices are natural numbers $ i\in \Idx \is \Nat $.
We make extensive use of finite maps from~$\Idx$:
if $A$ is the type of \emph{things} then
$\Idx \to \opt{A}$ is the type of \emph{hyper-things}.
As a notational convention,
if a meta-variable $a$ ranges over~$A$
we use $\m{a}$ to range over $\Idx \to \opt{A}$.
For a set of indices~$I\subseteq\Idx$ and some $x \in A$,
we write $\const{x}{I}$ for $ \m[i: x | i \in I] $.
Given $I\subseteq \Idx$,
we lift a relation~$\sim$ on~$A$ to the relation~$\sim_I$ on $\Idx \pto A$ as follows;
for $\m{a}_1,\m{a}_2 \from \Idx \pto A$,
$\m{a}_1 \sim_I \m{a}_2$ holds when
$I \subseteq \supp(\m{a}_1),\supp(\m{a}_2)$ and
$\forall i\in\Idx\st \m{a}_1(i) \sim \m{a}_2(i)$.
When $a_0 \in A$,
we write $\m{a} \sim_I a_0$ to mean
$\forall i\in\Idx\st \m{a}(i) \sim a_0$.

\begin{definition}[Reindexing]
  A function $\pi\from\Idx\to\Idx$ is called a \emph{reindexing}.
  Given ${\m{a} \from \Idx \to A}$,
  we write $\m{a}\isub*{\pi}$ to apply the reindexing $\pi$ to the map:
  $\m{a}\isub*{\pi} \is \m[i: \m{a}(\pi(i)) | i \in \Idx]$.
We write $\m{a}\isub{i_1->j_1,,i_n->j_n}$ to denote $\m{a}\isub*{\pi}$ where
  $\pi(i_k) = j_k$ and $\pi(i)=i$
  if~$i \in\Idx\setminus\set{i_1,\dots,i_n}$.
\end{definition}

\subsection{Hyper-programs}
\label{sec:programs}

\paragraph{Syntax}
We use a minimal untyped imperative language to formalize our ideas.
We will assume an enumerable set of \emph{program variables}~$\p{x} \in \PVar$.
The set of \emph{values}  $v \in \Val \is \Int$ is the set of integers, for simplicity.
Booleans are represented using~$0$ for false, and any other integer for true.
\begin{grammar}
\Term \ni t,g,e \is
      v | \code{x} | {*}
    | e\mathbin{\p{+}}e | e\mathbin{\p{-}}e | e \leq e | \dots
    \\&\;
    | \code{skip}
    | \code{x:=$e$}
    | t\p;t
| \code{if$\;g\;$then$\;t\;$else$\;t\;$}
| \code{while$\;g\;$do$\;t\;$}
\end{grammar}
Every term evaluates to some value and mutates a first-order store.
We use the meta-variables~$g$ for terms that are meant to evaluate
to a boolean (i.e.~guards), and~$e$ for terms that are meant to evaluate
to an integer (i.e.~expressions).
The~$*$ expression chooses some integer non-deterministically, and returns it.
Commands like \code{skip} and \code{while} have an irrelevant
(and thus arbitrary) return value.
A simplifying (but inessential) assumption is that evaluation never faults.
If needed, one can model expressions like $n\p{/}0$
as returning a special \p{NaN} value.

To keep the development as simple as possible,
we also don't model scoping and function calls.
Including them can be handled using standard techniques.
Function calls in our examples should be understood as simply naming code blocks.

The functions $\pvar(t)$ and $\mods(t)$ return the set of program variables
occurring in~$t$ and modified by~$t$ respectively.
Their definition is standard.
We extend them to hyper-terms by setting
$\pvar(\m{t}) = \set{ (\p{x}, i) | i\in\supp(\m{t}), \p{x} \in \pvar(\m{t}(i)) }$
(and similarly for $\mods$).

\paragraph{Semantics}
A \emph{store} is an element of $\Store \is \PVar \to \Val$.
For simplicity, we adopt a big-step semantics for our language.
The judgment $ \bigstep{t}{s}{v}{s'} $ indicates that
the term~$t$ starting from input store~$s$ may terminate
with the return value~$v \in \Val$ and output store~$s'$.
The definition of $ \bigstep{t}{s}{v}{s'} $ is in \appendixref{sec:program-semantics}.
We define $\bigsome{t}{s} \is \exists v,s'\st \bigstep{t}{s}{v}{s'}$.
Note that $\bigsome{t}{s}$ is equivalent to termination
only if~$t$ does not have non-deterministic steps.
For example, $ \bigsome{\code{while$\;*\;$do skip}}{s} $ holds even though,
in a standard small-step semantics, the program has a diverging execution.

\begin{definition}
\label{def:hyper-bigstep}
  A \emph{hyper-term} is a finite partial function
  $
    \m{t} \from \Idx \pto \Term.
  $
  A \emph{hyper-store} is a function
  $
    \m{s} \from \Idx \to \Store.
  $
  A \emph{hyper-return-value} is a finite partial function
  $
    \m{r} \from \Idx \pto \Val.
  $
  The big-step semantics judgment
  is lifted to hyper-terms as follows:
  \begin{align*}
    \bigstep{\m{t}}{\m{s}}{\m{r}}{\m{s}'} &\is
      \forall i\in\Idx \st
      \begin{cases}
        \bigstep{\m{t}(i)}{\m{s}(i)}{\m{r}(i)}{\m{s}'(i)}
        \CASE i\in\supp(\m{t})
        \\
        \m{s}(i) = \m{s}'(i) \land \isundef{\m{r}(i)}
        \OTHERWISE
      \end{cases}
\end{align*}
  Note that on all indices where~$\m{t}$ is undefined,
  the store is preserved untouched.
\end{definition}

\subsection{Hyper-assertions}
\label{sec:assertions}

\begin{definition}[Hyper-assertion]
\label{def:assrt}
\label{def:hyper-assertion}
  \emph{Assertions} are of type
  $
    A \in \Assrt \is \Store \to \Prop.
  $
A \emph{hyper-assertion} is a function of type
  $
    P \in \HAssrt \is (\Idx \to \Store) \to \Prop.
  $
A \emph{post hyper-assertion} has the type
  $ Q \from (\Idx \pto \Val) \to \HAssrt $.
  They are used in postconditions, where the hyper-value argument is bound to the return
  value of the hyper-term in the triple.
  We require post hyper-assertions to be
  \emph{upward-closed} on the hyper-return-value:
  if $ Q(\m{v})(\m{s}) $ and
     $ \forall i\in\supp(\m{v})\st \m{v}(i) = \m{v}'(i) $
  then $ Q(\m{v}')(\m{s}) $.
\end{definition}

\begin{mathfig}[\small]
  \begin{align*}
(\p{x}(i) = v) &\is
      \fun \m{s}.
        \m{s}(i)(\p{x}) = v
&
    P_1 \implies P_2 &\is
      \fun \m{s}.
        P_1(\m{s}) \implies P_2(\m{s})
\\
    P_1 \land P_2 &\is
      \fun \m{s}. P_1(\m{s}) \land P_2(\m{s})
&
    P_1 \lor P_2 &\is
      \fun \m{s}. P_1(\m{s}) \lor P_2(\m{s})
\\
    \E x.P(x) &\is
      \fun \m{s}.\exists x\st P(x)(\m{s})
&
    \A x.P(x) &\is
      \fun \m{s}.\forall x\st P(x)(\m{s})
\\
    P\isub*{\pi} &\is
      \fun \m{s}.
P(\m{s}\isub*{\pi})
    &
    Q\isub*{\pi} &\is
      \ret.Q(\ret\isub*{\pi})\isub*{\pi}
\\
    \P I.P &\is
      \fun \m{s}.
        \exists \pr{\m{s}}\st
          P(\m{s}\m[i: \m{s}'\!(i) | i \in I])
    &
    \PP I.Q &\is \ret.\E \m{v}.\P I.Q(\ret\m[i: \m{v}(i) | i\in I])
\\
    \at{A}{I} &\is
      \fun \m{s}. \LAnd_{i\in I} A(\m{s}(i))
&
    Q_1 \land Q_2 &\is \ret. Q_1(\ret) \land Q_2(\ret)
\end{align*}
  \caption{Hyper-assertions}
  \label{fig:hyper-assertions}
\end{mathfig}

The set of indices that are relevant for~$P$ is the set $\idx(P)$ defined below. For example, we have 
$
  \idx\bigl( (\p{x}(\I1)=\p{y}(\I2)) \lor \p{z}(\I3)=0 \bigr) = \set{\I1,\I2,\I3}
$.
\begin{definition}[\/$\idx$]
\label{def:assrt-idx}
  The \emph{indices of} a (post) hyper-assertion~$P$ (resp.~$Q$) are the set:
  \begin{align*}
    \idx(P) &\is
      \Idx \setminus
        \set{ i \in \Idx |
                \forall \m{s},s'\st P(\m{s}) \iff P(\m{s}\m[i: s']) }
    \\
    \idx(Q) &\is
      \Idx \setminus
        \set{ i \in \Idx |
                \forall \ret,\m{s},s'\st
                  Q(\ret)(\m{s}) \iff Q(\ret\m[i: \bot])(\m{s}\m[i: s']) }
  \end{align*}
\end{definition}

Similarly, we define $\pvar(P)$ as the set of (indexed) program variables of an hyper-assertion such that
$(\p{x}, i) \in \pvar(P)$ if arbitrarily changing the value
of \p{x} at~$i$ may affect whether~$P$ holds.

\begin{definition}[\/$\pvar$]
  The \emph{program variables of} a (post) hyper-assertion~$P$ (resp.~$Q$) are the set:
  \begin{align*}
    \pvar(P) &\is
      ( \PVar \times \Idx )
      \setminus
      \set*{
        (\p{x}, i) |
          \forall \m{s}, v\st
            P(\m{s}) \iff P\bigl(\m{s}\m[i: \m{s}(i){\m[\p{x}: v]} ]\bigr)
      }
    \\
    \pvar(Q) &\is\textstyle
      \Union_{\ret \from \Idx \pto \Val} \pvar(Q(\ret))
  \end{align*}
\end{definition}
Even though we defined $\idx$ and $\pvar$ semantically,
any over-approximation would
suffice to preserve the soundness of the side conditions of our rules.

We will be using a number of hyper-assertions,
summarized in \cref{fig:hyper-assertions}.
Pure meta-level propositions~$\phi$ lift to
\emph{pure} hyper-assertions~$\fun\wtv.\phi$,
which hold independently of the state.
An assertion $A \of \Assrt$, which predicates over single stores,
can be lifted to a hyper-assertion, which predicates over hyper-stores,
by $\at{A}{I}$ which specifies on which indices~$I$, $A$ is required to hold.
In addition to the usual logical connectives, we introduce some notation
to deal with hyper-stores.
To refer to the value of a program variable \p{x} at index~$i$
in a hyper-store, we write $\p{x}(i)$.

The hyper-assertion $P\isub*{\pi}$
uses the reindexing $\pi$ to change the indices of the hyper-store
before checking~$P$ holds on it.
The notation is extended to post hyper-assertions by reindexing the
hyper-return-value too.
Note that while reindexing an assertion is intuitively akin to applying a substitution to indices, we are giving it here a \emph{semantic} definition.
On simple assertions, reindexing does propagate
just like a substitution of indices.
For example $
  \bigl( 5 \leq \p{x}(\I1) \land \p{x}(\I2) \leq \p{x}(\I3) \bigr)
  \isub{1->2}
  $ is logically equivalent to $
  \bigl( 5 \leq \p{x}(\I2) \leq \p{x}(\I3) \bigr)
$.
We will however see in \cref{sec:idx-rules} that its interaction with WPs
is not trivial.

The \emph{projection modality} \(\Pi_I\) projects out the information on the
components in~$I$ of the hyper-store under consideration.
$\P I.P$ holds on a hyper-store~$\m{s}$
if~$P$ holds on some hyper-store that coincides with~$\m{s}$
on every index except for the ones in~$I$.
Notice that~$ \idx(\P I.P) \inters I = \emptyset $.
For brevity, when~$I=\set{i}$ is a singleton, we omit the braces
and write~$\P i.P$.
Projection of a post hyper-assertion $\PP I.Q$,
projects out the components at~$I$ of the hyper-return-value too.

\begin{definition}
\label{def:entailment}
  An assertion~$A$ is \emph{valid}, written~$ \V |- A $,
  just if~$ \forall s\st A(s) $.
  Similarly, we define
  \emph{validity}, \emph{entailment}, and \emph{logical equivalence}
  of hyper-assertions:
  \begin{align*}
    \V & |- P \is \forall \m{s}\st P(\m{s})
    &
    \V \Gamma & |- P \is (\V |- \LAnd \Gamma \implies P)
    &
    P &\lequiv P' \is (\V P |- P') \land (\V P' |- P)
  \end{align*}
  where~$ \Gamma $ is a list of hyper-assertions.
  When $\Gamma$ is empty, $\LAnd \Gamma = \True$.
  Since~$\Gamma$ is always interpreted as the conjunction of its items
  we will not make a distinction between~$\Gamma$ and $\LAnd \Gamma$.
\end{definition}

\begin{mathfig}[\small]
  \begin{proofrules}
    \V \Gamma, P |- P

    \infer{
      \V \Gamma |- P'
      \\
      \V \Gamma, P' |- P
    }{
      \V \Gamma |- P
    }

    \infer*[fraction={===}]{
      \V \Gamma, P |- R
    }{
      \V \Gamma |- P \implies R
    }

    \infer{
      \forall x\st (\V \Gamma, P(x) |- R)
    }{
      \V \Gamma, \E x.P(x) |- R
    }

    \infer{
      \V \Gamma |- P_1
      \\
      \V \Gamma |- P_2
    }{
      \V \Gamma |- P_1 \land P_2
    }

\end{proofrules}
  \caption{Basic hyper-assertion laws (selection).}
  \label{fig:basic-laws}
\end{mathfig}

\begin{mathfig}[\small]
  \begin{proofrules}
    \infer*[lab=idx]{
      \V \Gamma |- P
}{
      \V \Gamma\isub*{\pi} |- P\isub*{\pi}
    }
    \label{rule:reindex}
    \label{rule:idx-intro}

    \infer*[lab=idx-pvar]{}{
      (\p{x}(i)=v)\isub*{\pi} \lequiv \p{x}(\pi(i))=v
    }
    \label{rule:idx-pvar}

    \infer*[lab=idx-ex]{}{
      (\E x.P(x))\isub*{\pi} \lequiv \E x.(P(x)\isub*{\pi})
    }
    \label{rule:idx-ex}

    \infer*[lab=idx-conj]{}{
      (P_1 \land P_2)\isub*{\pi} \lequiv (P_1\isub*{\pi} \land P_2\isub*{\pi})
    }
    \label{rule:idx-conj}

    \infer*[lab=idx-impl]{}{
      (P_1 \implies P_2)\isub*{\pi} \lequiv (P_1\isub*{\pi} \implies P_2\isub*{\pi})
    }
    \label{rule:idx-impl}

    \infer*[lab=idx-irrel]{
      \forall i \in \idx(P)\st \pi(i)=i
    }{
      P\isub*{\pi} \lequiv P
    }
    \label{rule:idx-irrel}

    \infer*[lab=proj]{
      \V \Gamma |- P
    }{
      \V \P I.\Gamma |- \P I.P
    }
    \label{rule:proj-intro}

    \infer*[lab=proj-intro]{}{
      \V P |- \P I. P
    }
    \label{rule:proj-weak}

    \infer*[lab=proj-merge]{}{
      \P I_1.\P I_2. P \lequiv \P I_1\union I_2. P
    }
    \label{rule:proj-merge}

    \infer*[lab=proj-irrel]{
      \idx(P) \inters I = \emptyset
    }{
      \V \P I. P |- P
    }
    \label{rule:proj-irrel}

   \infer*[lab=proj-store]{
     i \notin \idx(P)
     \\
     \card{\vec{v}} = \card{\pv{x}}
   }{
     \V \E \vec{v}.P(\vec{v})
     |- \P i.\E \vec{v}.(P(\vec{v}) \land \pv{x}(i) = \vec{v})
   }
   \label{rule:proj-store}
  \end{proofrules}
  \caption{Hyper-assertion laws for reindexing and projection.}
  \label{fig:idx-proj-laws}
\end{mathfig}

In \cref{fig:basic-laws} we show a selection of laws
on entailments of basic hyper-assertions.
These laws are standard and mirror the laws for standard connectives.
In \cref{fig:idx-proj-laws} we show the core laws that apply to the two
constructs that are special to hyper-assertions:
reindexing and projection.

\Cref{rule:idx-intro} allows the application of arbitrary reindexing
on both sides of the turnstile.
\Cref{rule:idx-pvar,rule:idx-ex,rule:idx-conj,rule:idx-impl} show how
reindexing distributes over the other connectives.
To eliminate a reindexing one can use these rules to push the reindexing
down the structure of the assertion until either \ref{rule:idx-pvar} applies
or the reindexing does not affect the indices of the assertion
and \ref{rule:idx-irrel} allows to remove it.

\Cref{rule:proj-intro} allows projection to be introduced
on both sides of the turnstile.
Notice that the rule is not an instance of consequence:
the assumption~$ \P I.\Gamma $ in the conclusion
does not imply the assumption~$\Gamma$ in the premise.

\Cref{rule:proj-weak} is the most obvious way to introduce a projection,
but it is not the most useful;
typically when introducing an assertion~$\P i.P$
we want to assert in~$P$ facts that
are not true for the current store at index~$i$,
but would be true if we reassigned the store at~$i$ appropriately.
\Cref{rule:proj-store} supports this common scenario:
for instance, we can use it to prove
$
  {(\p{x}(\I1)=\p{x}(\I2))} \proves
  \E v.\p{x}(\I1)=\p{x}(\I2)=v \proves
  \P\I3.\bigl(\E v.\p{x}(\I1)=\p{x}(\I2)=v \land \p{x}(\I3)=v\bigr) \proves
  \P\I3.(\p{x}(\I1)=\p{x}(\I2)=\p{x}(\I3)).
$

\Cref{rule:proj-irrel} is the main mean to eliminate projection:
starting from an assertion~$P_0$ with $ \idx(P_0) \inters I \ne \emptyset $
one first proves~$\V P_0 |- P$ for some suitably strong~$P$ with
$ \idx(P_0) \inters I = \emptyset $;
then an application of \ref{rule:proj-intro} and \ref{rule:proj-irrel}
give us $ \V \P I.P_0 |- P $.

 \section{The program logic}
\label{sec:logic}

To make our logic capable of asserting properties of programs,
we introduce a \emph{weakest precondition}~(WP) modality
$ \WP{\m{t}}{Q} $,
where
$ \m{t} \from \Idx \pto \Term $ is a hyper-term, and
$ Q \from (\Idx \pto \Val) \to \HAssrt$
is the postcondition.\footnote{We
omit the binder in the postcondition
  when the return values are simply ignored,
  i.e.~$ \WP {\m{t}} {P} \is \WP {\m{t}} {\fun\_.P} $.
}
It intuitively holds on the hyper-states from which
$ \m{t} $ runs yielding an hyper-return-value~$\m{v}$ and an hyper-state $\m{s}'$
satisfying the postcondition~$Q$.
The \emph{arity} of $\WP{\m{t}}{Q}$ is~$\card{\supp(\m{t})}$.

\begin{definition}[Weakest precondition]
\label{def:wp}
  $
    \WP{\m{t}}{Q} \is
      \fun \m{s}. \big(
        \forall \m{s}', \m{v} \st
          \bigstep{\m{t}}{\m{s}}{\m{v}}{\m{s}'}
            \implies
              Q(\m{v})(\m{s}')
      \big).
  $
\end{definition}

Hyper-triples are defined in terms of weakest-preconditions:
$
  \T{P}{\m{t}}{Q} \is P \implies \WP{\m{t}}{Q}.
$
This definition is consistent with the one adopted by~\cite{BartheCK11,SousaD16}.
Other relational program logics, notably~\cite{Benton04,Yang07},
insist that the hyper-terms either all diverge or all terminate.
We elaborate on the trade-offs implied by this choice in \cref{sec:discuss-design}.

One possible variation is to require \emph{safety} of the terms,
\ie that the terms do not fault
(as in e.g.~ \cite{Yang07})
which is a common choice in Hoare/Separation logic.
For simplicity we do not model faults in our language.
The most flexible extension would model faults
as a special return value~$\fault$.
This way, the postcondition can decide what should happen when each component faults. For some applications, requiring safety ($\ret.\A i.\ret[i]\ne\fault$) may be appropriate, for others it may be sufficient to show that if one component faults then the others do too
(${\ret.\A i,j.\ret[i]=\fault \iff \ret[j]=\fault}$).

Another important point of departure from the literature is that we consider
a non-deterministic programming language
(in contrast with e.g.~\cite{Yang07,Benton04,BartheCK11}).
An hyper-triple of shape
$ \J |- {\pv{x}(\I1)=\pv{x}(\I2)} {\m[\I1: t_1, \I2: t_2]} {\pv{x}(\I1)=\pv{x}(\I2)} $,
for example,
encodes semantic equivalence of~$t_1$ and~$t_2$ in a deterministic language.
But in general it represents a stronger property, which implies
determinism of~$t_1$ and~$t_2$ (assuming they terminate).
The flexibility of our logic allows us to consider non-deterministic programs,
and obtain the stronger proofs that the hypothesis of determinism enables
by simply involving in the derivations the \emph{hyper-triple encoding}
of determinism.

\paragraph{Overview of the proof rules}
\label{sec:rules-overview}

We take a semantic approach in formalizing WP,
therefore proof rules are just lemmas involving WP.
We divide \thelogic's rules in four groups.
The \emph{structural} rules (in \cref{fig:wp-struct-laws})
apply regardless of the hyper-terms in the WP,
and are mostly adaptations of standard Hoare Logic rules.
Then we have the \emph{lockstep} rules (in \cref{fig:wp-lockstep-laws}),
which match on the structure of the program terms
of all components at the same time.
These are minor adaptations of rules that are virtually present
(most often in the special case of arity~2)
in all the relational program logics in the literature.
The \emph{hyper-structure} rules (in \cref{fig:wp-hyper-laws})
provide the basic reasoning principles needed to
compose hyper-triples of varying arity.
Finally, the \emph{reindexing} rules (in \cref{fig:wp-idx-laws})
allow the sound merging of components,
which underpins the manipulations needed to handle indirect-style triples.

All rules apply only if all the components are defined.
This implies some implicit side conditions.
For example, whenever a rule involves an expression $\m{t}_1 \m. \m{t}_2$,
it only applies if $\supp(\m{t}_1) \inters \supp(\m{t}_2) = \emptyset$.
Similarly, an expression $\m{t}_1 \m+ \m{t}_2$ implies the constraint that
$
  \forall i \in \supp(\m{t}_1) \inters \supp(\m{t}_2) \st
    \m{t}_1(i) = \m{t}_2(i).
$

\subsection{Structural Rules}
\label{sec:struct-rules}

\begin{mathfig}[\small]
  \begin{proofrules}
    \infer*[lab=wp-triv]{}{
  \V |- \WP{\m{t}}{\True}
}     \label{rule:wp-triv}

    \infer*[lab=wp-cons]{
  \forall \m{v}\st\ \V Q(\m{v}) |- Q'(\m{v})
}{
  \V \WP{\m{t}}{Q} |- \WP{\m{t}}{Q'}
}     \label{rule:wp-cons}

    \infer*[lab=wp-all]{}
{
  \A x.\WP{\m{t}}{Q(x)} \lequiv \WP{\m{t}}{\A x.Q(x)}
}
     \label{rule:wp-all}

    \infer*[lab=wp-frame]{
\pvar(P) \inters \mods(\m{t}) = \emptyset
}{
  P \land \WP{\m{t}}{Q}
  \proves
  \WP{\m{t}}{P \land Q}
}     \label{rule:wp-const}
    \label{rule:wp-frame}

    \infer*[lab=wp-impl-r]{
  \pvar(P) \inters \mods(\m{t}) = \emptyset
}{
  P \implies \WP{\m{t}}{Q}
  \lequiv
  \WP{\m{t}}{\ret. P \implies Q(\ret)}
}     \label{rule:wp-impl-r}

    \infer*[lab=wp-subst]{
  \p{x} \notin \mods(t)
}{
  \V
  \p{x}(i) = v
  \land
  \WP{\bigl(\map[\big]{i: t\subst{\p{x}->v}}\m.\m{t}'\bigr)}{Q}
  |-
  \WP{(\m[i: t]\m.\m{t}')}{Q}
}
     \label{rule:wp-subst}

    \infer*[lab=wp-idx]{
  \pi\ \mathrm{bijective}
}{
  \V (\WP{\m{t}}{Q})\isub*{\pi} |- \WP{\m{t}\isub*{\pi}}{Q\isub*{\pi}}
}     \label{rule:wp-idx}
  \end{proofrules}
  \caption{Weakest precondition laws: structural rules.}
  \label{fig:wp-struct-laws}
\end{mathfig}

The rules in \cref{fig:wp-struct-laws}
represent basic inferences that are typically available
in Hoare-style program logics.
They hold generically on the hyper-term of the relevant WP:
they apply independently of which components and which terms the hyper-term contains.
The rules can be understood as axiomatic accounts of how
the basic connectives commute with WP.
\Cref{rule:wp-triv} states that a WP ``commutes'' with $\True$:
the WP with trivial postcondition is trivially true.
This rule is sound for our model of WP that does not insist on
safety of~$\m{t}$.
\Cref{rule:wp-cons} states that WP preserves entailment.
This encodes the usual rule of consequence:
if one can prove the WP with postcondition~$Q$,
the same WP with a weaker~$Q'$ is also provable.
\Cref{rule:wp-all} states that WP commutes with~$\Forall$.

\Cref{rule:wp-const} generalises Hoare's \emph{frame} rule,
also known as \emph{constancy}.
It shows that WP commutes with~$(P \land)$
when~$P$ does not depend on the variables modified by~$\m{t}$.
Note that~$P$ does not depend on the return values.
\Cref{rule:wp-impl-r} is the analog of \ref{rule:wp-const}
for~$(P\implies)$.
\Cref{rule:wp-subst} allows the sound subsitution of values for program variables.

\Cref{rule:wp-idx} explains how a bijective reindexing propagates through
a WP, by applying the reindexing both to the term and to the postcondition.
Combined with \ref{rule:idx-intro},
this effectively closes the proofs under all renamings,
revealing an underlying symmetry of the entailment judgments
(\ie no index is treated specially).
We typically use the rule in concert with \ref{rule:idx-intro} and \ref{rule:wp-cons} to uniformly rename the components across a judgment.

\subsection{The Lockstep Rules}
\label{sec:lockstep-rules}

\begin{mathfig}[\small]
  \begin{proofrules}
    \infer*[lab=wp-seq$_I$]{}{
  \WP {\m[i: t_i | i \in I]}[\big]{
    \WP {\m*[i: \smash{t'_i} | i\in I]} {Q}
  }
  \lequiv
  \WP {\m[i: (t_i\code{;}\ t'_i) | i \in I]} {Q}
}     \label{rule:wp-seq}

    \infer*[lab=wp-assign$_I$]{
  \forall i\in I\st
    (\p{x}_i,i) \not\in \pvar(Q)
}{
  \V
  \WP {\m[i: e_i | i\in I]} {Q}
  |-
  \WP {\m[i: \p{x}_i \code{:=}\, e_i | i \in I]}
      {\ret.Q(\ret) \land \LAnd_{i\in I} \ret[i] = \p{x}_i(i)}
}     \label{rule:wp-assign}

    \infer*[lab=wp-if$_I$]{}{
  \WP {\m[i: g_i | i\in I]}*{
    \fun\m{b}.
    \WP {
    \begin{pmatrix*}[l]
      {\m[i: t_i | i\in I, \m{b}(i) \neq 0]}
      {\m.}\\
      {\m[i: t'_i | i\in I, \m{b}(i) = 0]}
    \end{pmatrix*}
    }[\big]{Q}
  }
  \lequiv
  \WP {\m[i: \code{if}\ g_i\ \code{then}\ t_i\ \code{else}\ t'_i | i \in I]}
      {Q}
}     \label{rule:wp-if}

    \infer*[lab=wp-while$_I$]{
  \V
  P
  |-
  \WP {\m[i: g_i | i\in I]}[\big]{
    \fun\m{b}.
    (\m{b} =_I 0 \land R)
    \lor
    (\m{b} \ne_I 0 \land \WP {\m[i: t_i | i \in I]} {P})
  }
}{
  \V
  P
  |-
  \WP {\m[i: \code{while}\ g_i\ \code{do}\ t_i | i \in I]} {R}
}     \label{rule:wp-while}

    \infer*[lab=wp-refine]{}{
  \V
  \at{(t_1 \semleq t_2)}{i},
  \WP {(\m*[i: t_2] \m. \m{t})} {Q}
  |-
  \WP {(\m[i: t_1] \m. \m{t})} {Q}
}     \label{rule:wp-refine}
  \end{proofrules}
  \caption{Weakest precondition laws: lockstep rules.}
  \label{fig:wp-lockstep-laws}
\end{mathfig}

The rules in~\cref{fig:wp-lockstep-laws} are all straightforward extensions
of the corresponding Hoare logic rules, to \pre k-ary hyper-triples.
Variations of these rules appear in virtually all other relational logics.
In fact, one can recover Hoare logic by instantiating our rules
to the \pre 1-ary hyper-triple case.
This embedding is more awkward to obtain in other relational logics
that insist on a fixed arity greater than~1 for relational triples,
e.g.~\cite{Benton04}.

The \cref{rule:wp-while} rule follows the same pattern.
The rule applies if all the components are while loops,
and verifies their guards and their bodies as if they executed in lockstep.
Exactly like its Hoare logic counterpart, the rule is based on loop invariants:
here~$P$ is the (relational) loop invariant.
The premise asks to prove, assuming the loop invariant holds initially,
that after the evaluation of all the guards,
we only have two cases.
The first case is where all the guards evaluated to false ($\m{b} =_I 0$)
and the overall postcondition~$R$ holds.
The second case is where all the guards evaluated to true ($\m{b} \ne_I 0$).
In that case we also have to prove that running all the loop bodies once
results in re-establishing the loop invariant~$P$.
Note that the disjunction does not allow for the guards to go ``out of sync'':
the loops execute exactly the same number of times.

The lockstep principle is very advantageous when it applies,
but its applicability is very restricted:
most often the control paths taken in two components
will differ to a point where this strategy cannot be used
or becomes counterproductive.
Overcoming the rigidity of the basic lockstep proof strategy
has been a goal of many proposals in the literature.
For example, \cite{BartheCK11,BartheGHS17} include a number of semantic-preserving transformations, like loop unrolling and loop splitting,
that can be applied to terms so that they are brought to a shape
amenable to application of lockstep rules.
Similarly, \cite{SousaD16} provide a set of rules,
dubbed ``Cartesian Loop Logic'',
that perform a limited set of such transformations to terms.
In~\cite{BartheGHS17} a generalized while rule allows
stuttering in the alignment of the (two) loops considered,
without having to syntactically rewrite the terms.
To express these non-trivial alignments,
we include in our logic the \ref{rule:wp-refine} rule,
which allows to replace a term~$t_1$ in a WP,
with another term $t_2$ if every behaviour of~$t_1$ is also a behaviour of~$t_2$.

\begin{definition}[Refinement]
  The refinement $t_1 \semleq t_2$ holds
  when~$t_2$ has all the behaviours of~$t_1$:
  \[
    t_1 \semleq t_2
    \is
      \fun s.
      \bigl(
      \forall s',v\st
        \bigstep{t_1}{s}{v}{s'}
        \implies
        \bigstep{t_2}{s}{v}{s'}
      \bigr)
  \]
  Semantic equivalence is defined as~$ (t_1 \semeq t_2) \is (t_1 \semleq t_2 \land t_2 \semleq t_1) $.
\end{definition}
Proving $t_1 \semleq t_2$ generally requires meta-level reasoning about
the semantics of the terms.
The rule should therefore be considered as a last resort,
and be used by instantiating the side condition with known generic refinements.
Thanks to the flexibility of \thelogic,
many challenging examples that in other logics would require
ad-hoc refinements can still be handled within the logic,
by only including the special case of the \ref{rule:wp-refine} rule
instantiated with loop unfolding: $
  \code{while$\;g\;$do$\;t$}
  \semeq
  \code{if$\;g\;$then$\;(t$;while$\;g\;$do$\;t)$}
$.
We show an example of this pattern in \cref{sec:discuss-refinement}.

\subsection{The Hyper-structure Rules}
\label{sec:hyper-str-rules}

\begin{mathfig}[\small]
  \begin{proofrules}
    \infer*[lab=wp-nest]{}{
\WP{\m{t}_1}{\fun\m{v}.
    \WP{\m{t}_2}{\fun\m{w}.
      Q(\m{v}\m.\m{w})}}
  \lequiv
  \WP{(\m{t}_1 \m. \m{t}_2)}{Q}
}     \label{rule:wp-nest}

    \infer*[lab=wp-conj]{
\idx(Q_1) \inters \supp(\m{t}_2) \subs \supp(\m{t}_1)
  \\
\idx(Q_2) \inters \supp(\m{t}_1) \subs \supp(\m{t}_2)
}{
  \V
  \WP{\m{t}_1}{Q_1}
  \land
  \WP{\m{t}_2}{Q_2}
  |-
  \WP{(\m{t}_1 \m+ \m{t}_2)}{Q_1 \land Q_2}
}
     \label{rule:wp-conj}

    \infer*[lab=wp-proj,
  Right={$ I = \supp(\m{t}_1) $}
]{
}{
  \V
  \P I. \bigl(
    \proj(\m{t}_2) \implies \proj(\m{t}_1)
    \land
    \WP{(\m{t}_1 \m. \m{t}_2)}{Q}
  \bigr)
  |-
  \WP{\m{t}_2}{\PP I.Q}
}
     \label{rule:wp-proj}
  \end{proofrules}
  \caption{Weakest precondition laws: hyper-structure rules.}
  \label{fig:wp-hyper-laws}
\end{mathfig}

The rules in~\cref{fig:wp-hyper-laws}
are the key new rules of \thelogic.
We call them \emph{hyper-structure} rules because they decompose the goal
by breaking up the components of a hyper-term.
They represent three key reasoning principles available for WP on hyper-terms:
how nested WPs can be merged into one (\ref{rule:wp-nest}),
how a conjunction of WPs can be merged into one (\ref{rule:wp-conj}),
and how to soundly remove components from a WP (\ref{rule:wp-proj}).
They naturally arise from studying how WP commutes with other constructs,
specifically, with another WP, with conjunction and with projection.

The \ref{rule:wp-nest} rule states that a WP on a hyper-term that can be
split into two disjoint hyper-terms
$ \m{t}_1 $ and $ \m{t}_2 $,
can be equivalently expressed as the nesting of a WP for~$\m{t}_2$ in the postcondition of the WP for $\m{t}_1$.
As we saw in \cref{sec:overview:conj-nest},
this rule increases the flexibility of the logic.
As seen in \cref{sec:overview:loops},
the nested WP pattern becomes essential when using the \ref{rule:wp-while}
rule with a loop invariant that needs ``side-computation'' to be stated.

The \ref{rule:wp-conj} rule states that the conjunction of two WPs entails
a single WP where the two postconditions are conjoined,
and the hyper-terms are unioned.
On an index $i \notin \supp(\m{t}_1) \inters \supp(\m{t}_2) $,
the postcondition of the WP on $\m{t}_1$ and $\m{t}_2$ would
observe different stores.
The side conditions make sure~$Q_1$ and~$Q_2$ ignore the problematic indices.

In practice, \ref{rule:wp-conj} allows us to break down the current goal
as a conjunction of two smaller hyper-triples.
It is not mandatory for
the hyper-terms~$\m{t}_1$ and~$\m{t}_2$ to have components in common,
but the application of the rule is more powerful when they do.
As a simple example, when $\m{t}_1 = \m{t}_2$,
the rule allows the combination of multiple WPs on the same hyper-term.
We have seen instances of this pattern in \cref{sec:motivation}.

The embedding of Hoare logic together with \ref{rule:wp-conj}
immediately entails a relative completeness result for our logic:
given an oracle for derivations of Hoare triples,
we can prove any hyper-triple by using \ref{rule:wp-conj} to compose
the derivations for the strongest unary Hoare triple for each component,
and then \ref{rule:wp-cons} to imply the original goal
(see \appendixref{sec:completeness} for details).
This proof strategy, however, is the one that removes all opportunities
for relational proofs, negating the benefits of working in a relational logic.

The last hyper-structure rule is \ref{rule:wp-proj},
that can be seen as a commutation rule between WP and projection.
By using \cref{rule:wp-proj,rule:proj-intro}
one can derive \ref{rule:wp-proj-simpl} (shown in \cref{sec:overview:proj}),
which explains how to soundly remove some components from a WP.
As illustrated in \cref{sec:overview:proj},
and further discussed in
\cref{sec:discuss-divide-and-conquer,sec:discuss-refinement},
\ref{rule:wp-proj-simpl} is typically used in combination
with \ref{rule:wp-conj} to introduce an auxiliary hyper-term~$\m{t}$,
so that a goal that relates some $\m{t}_1 \m. \m{t}_2$ can be broken
into a goal that relates~$\m{t}_1 \m. \m{t}$,
and one that relates~$\m{t}_2 \m. \m{t}$.

If we ignored the conjunct regarding the $\proj(\hole)$ assertion,
the \ref{rule:wp-proj} rule would be unsound.
The reason is that when even
a single component of a WP diverges,
the WP definition does not require the output stores
of the other (terminating) components to satisfy the postcondition.
For example,
$ \V |- \WP {\m[\I1:\code{skip}, \I2: \code{while$\;$1$\;$do skip}]} {\False} $
holds.
Projecting out component~\I{2} blindly, however,
would produce the invalid triple
$ \V |- \WP {\m[\I1:\code{skip}]} {\False} $.
To ensure soundness,
\ref{rule:wp-proj} uses the
\emph{projectability} assertion.
\begin{definition}[Projectable]
  The assertion~$\proj(t)$ holds on states where~$t$
  is \emph{projectable}, \ie can yield a result:
  $ \proj(t) \is \fun s.\bigsome{t}{s} $.
  We also define an analogous hyper-assertion parameterized over hyper-terms:
  $ \proj(\m{t}) \is \fun \m{s}.\bigsome{\m{t}}{\m{s}} $.\end{definition}
To be able to apply \ref{rule:wp-proj},
one needs to establish that,
if all the terms in $\m{t}_2$ can terminate
so can all the terms in~$\m{t}_1$ which we are projecting out.
This would rule out our counterexample,
and in fact is sufficient to prove soundness.

The projectability condition
is strictly weaker than requiring termination of~$t'$,
in two ways.
First, it corresponds to ``may'' termination:
it would hold if~$t'$ has both a diverging and a terminating trace from some
initial store. For example, $\code{while$\;*\;$do skip}$ is projectable:
from every initial store it has both diverging and terminating traces.
Second, it is conditional on termination of the terms in~$\m{t}_2$.
This can be useful when~$\m{t}_1(i)$ is a sub-term of some other component~$\m{t}_2(j)$, causing~$\m{t}_1(i)$ and~$\m{t}_2(j)$ to terminate under the same conditions.

\subsection{The Reindexing Rules}
\label{sec:idx-rules}

\begin{mathfig}[\small]
  \begin{proofrules}
    \infer*[lab=wp-idx-post]{
  \V \Gamma |- \WP{\m{t}}{Q}
  \\
  j \notin \supp(\m{t}) \union \idx(\Gamma)
}{
  \V \Gamma |- \WP{\m{t}}{Q\isub{j->i}}
}
     \label{rule:wp-idx-post}

    \infer*[lab=wp-idx-swap]{
  i \notin \idx(Q)
}{
  \V
  \bigl(\WP{(\m[j: t]\m.\m{t}')}{Q}\bigr)\isub{j->i}
  |-
  \WP{(\m[i: t]\m.\m{t}')}{Q\isub{j->i}}
}
     \label{rule:wp-idx-swap}

    \infer*[lab=wp-idx-pass]{
  i,j \notin \supp(\m{t})
}{
  \V (\WP{\m{t}}{Q})\isub{j->i} |- \WP{\m{t}}{Q\isub{j->i}}
}
     \label{rule:wp-idx-pass}

    \infer*[lab=wp-idx-merge]{}{
  \V
  \bigl(\WP{(\m[i: t, j: t]\m.\m{t}')}{Q}\bigr)\isub{j->i}
  |-
  \WP{(\m[i: t]\m.\m{t}')}{Q\isub{j->i}}
}
     \label{rule:wp-idx-merge}
  \end{proofrules}
  \caption{Weakest precondition laws: reindexing rules.}
  \label{fig:wp-idx-laws}
\end{mathfig}

Reindexing is a useful tool in relational proofs.
One way of introducing a reindexing is by using \cref{rule:idx-intro}.
While propagating the effects of a reindexing is straightforward for
assertions with basic connectives
(through the rules of \cref{fig:idx-proj-laws}),
its interaction with WPs is much more interesting.

\Cref{rule:wp-idx} already handles bijective reindexing.
The rules in \cref{fig:wp-idx-laws} represent the sound interactions
between \emph{non-bijective} reindexing and WPs.
These reindexings boil down to compositions of reindexings of the form
$\isub{j->i}$ where~$i\ne j$.
The objective is to understand how a WP and a reindexing ``commute'':
given $ \WP{\m{t}}{Q}\isub{j->i} $, how does the reindexing propagate
to~$\m{t}$ and~$Q$?
There are four cases to consider, depending on whether~$i$ and~$j$
are indices of~$\m{t}$ or not.

\Cref{rule:wp-idx-pass} deals with the simple case where $i,j \notin \supp(\m{t})$, in which the reindexing has no effect on the hyper-term,
and can be simply propagated to the postcondition.

\Cref{rule:wp-idx-swap} handles the case where $j \in \supp(\m{t})$
but $i\notin \supp(\m{t})$.
In this case $\m{t} = \m[j: t] \m. \m{t}'$ for some~$t$ and $\m{t}'$,
with $i \notin \m{t}'$
(the latter constraint is implied by well-formedness of the judgment in the rule).
The rule then states that the reindexing is applied to the hyper-term
by exchanging index~$j$ for index~$i$ and the reindexing is propagated to the postcondition.
To be sound, the rule requires $Q$ not to predicate on the index~$i$:
indeed in the starting WP, references to index~$i$ would refer to the initial store at~$i$ (since the hyper-term does not affect it),
while in the resulting WP the store at~$i$ is modified by the hyper-term.

\Cref{rule:wp-idx-merge} deals with the case where both $i,j \in \supp(\m{t})$.
In this case the effect on the hyper-term should be of ``merging'' the two components, which is only meaningful if they are mapped to the same term.
Therefore the rule matches on $\m{t} = \m[i: t, j: t]\m.\m{t}'$.
The reindexing is then propagated to the postcondition.
The intuition is that we have proved that $Q$ holds on the results
of any two runs of~$t$, so it will hold when both runs are the same run (at a single index~$i$).

The only missing case is when~$i \in \supp(\m{t})$ and~$j \notin \supp(\m{t})$.
Since the reindexing does not affect the indices of~$\m{t}$,
one might be tempted to write a rule like \ref{rule:wp-idx-pass}:
$
  \V (\WP{\m{t}}{Q})\isub{j->i} |- \WP{\m{t}}{Q\isub{j->i}}.
$
Such rule would however be unsound.
For example, we could start with the valid
$
  \V \p{x}(\I2)=0 |- \WP{\m[\I1: \code{x:=1}]}{\p{x}(\I1)=1 \land \p{x}(\I2)=0},
$
apply \ref{rule:idx-intro} with $ \isub*{\pi} = \isub{2->1} $
to obtain the valid
$
  \V
  \p{x}(\I1)=0
  |-
  \bigl(\WP{\m[\I1: \code{x:=1}]}{\p{x}(\I1)=1 \land \p{x}(\I2)=0}\bigr)
  \isub{2->1}
  .
$
An application of our tentative rule would yield
$
  \V
  \p{x}(\I1)=0
  |-
  \WP{\m[\I1: \code{x:=1}]}{\p{x}(\I1)=1 \land \p{x}(\I1)=0}
$
which has an unsatisfiable postcondition (and is thus invalid).
One simple side condition that would make the rule sound is
$ j \notin \idx(Q) $; this situation is however already handled by the
\ref{rule:idx-irrel} rule.

\Cref{rule:wp-idx-post} represents a more useful way of dealing with
the reindexing $\isub{j->i}$ when~$i \in \supp(\m{t})$
and~$j \notin \supp(\m{t})$.
The rule states that in such case it is possible to \emph{introduce}
the reindexing in the postcondition,
provided the assumptions~$\Gamma$ do not constrain the store at~$j$.

The soundness argument of \ref{rule:wp-idx-post} goes as follows.
Let~$\m{s}$ be an input hyper-store satisfying~$\Gamma$,
and let~$\m{s}'\m[i: s]$ be the hyper-store resulting from running
$\m{t}$ on $ \m{s} $.
To establish $Q\isub{j->i}$ on~$\m{s}'\m[i: s]$,
we need to prove~$Q$ holds on $\m{s}'\m[i: s, j:s]$.
Since~$\Gamma$ does not constrain the store at~$j$,
we know it holds on~$\m{s}\m[j: s]$.
Moreover, since~$j\notin\supp(\m{t})$,
the hyper-term run from $\m{s}\m[j: s]$
will output $\m{s}'\m[i: s, j:s]$.
The premise therefore implies~$Q$ holds on the output hyper-store.
The crucial step in the proof is the act of feeding the output store at~$i$
in the conclusion, as the input store at~$j$ in the premise.
This precisely captures the conversion from indirect-style triples to direct
triples.

 \section{Discussion}
\label{sec:discussion}
Here, we highlight some key features of \thelogic{} through some examples and a discussion of how certain design choices make this logic stand out.
In the interest of space, the examples are presented in broad strokes.
\Appendixref{sec:case-studies}
contains their full proofs and additional case studies.

\subsection{Modularity}
\label{sec:discuss-modularity}
\begin{wrapfigure}[7]{R}{25ex}\vspace*{-1.8em}
  \begin{sourcecode*}[gobble=4]
    f(x, y) $\is$
      while (i < x) do
        r := op(r, y);
        i := i + 1
  \end{sourcecode*}
  \vspace*{-1em}
  \caption{A program that is parametric on \p{op}.}
  \label{fig:f-code}
\end{wrapfigure}

One of the main goals of \thelogic\ is allowing the construction of
truly modular proofs.
Consider for example the code in \cref{fig:f-code}:
the code uses a library-provided function \p{op};
\p{f} accumulates its output in~\p{r} (assumed initially~0).
Suppose we want to prove that \p{f} distributes
over \p{op} in the second argument:
$
  \p{f}(a,\p{op}(b,c)) = \p{op}(\p{f}(a,b),\p{f}(a,c)).
$
This property holds if \p{op} is:
\begin{enumerate*}[label=(\roman*)]
  \item not modifying variables of~\p{f},
    \label{prop:op-mods}
  \item a total function (\ie~projectable and deterministic)
    with $\p{op}(0,0)=0$ (1-safety),
    \label{prop:op-fun}
  \item associative (4-safety), and
    \label{prop:op-assoc}
  \item commutative (2-safety).
    \label{prop:op-comm}
\end{enumerate*}
Seeing as both the goal and these assumptions are hyper-safety properties,
we would want to build a modular proof of distributivity of \p{f}:
one that does not rely on the specific implementation of \p{op},
but only on the properties listed above expressed as hyper-triples.

In this example, the intuitive proof strategy is a vanilla lockstep alignment:
if we consider the hyper-term
$
 \m[
   \I1: \p{f}(a,\p{op}(b,c)),\allowbr
   \I2: \p{f}(a,b),\allowbr
   \I3: \p{f}(a,c)
 ]
$
all the components execute~$a$ iterations of their loops;
intuitively, we should be able to prove
the relational loop invariant
$ \p{r}(\I1) = \p{op}(\p{r}(\I2), \p{r}(\I3)) $
with a lockstep proof.
However, even though the high-level proof strategy of this example is a lockstep alignment,
the lockstep-based logics are unable to provide a modular proof.
There are two main obstacles.
First, the loop invariant sketched above applies \p{op} and thus cannot be
readily represented as an assertion.
Second, the properties of \p{op} we want to rely on have mixed arities,
and we would need to apply them to runs of \p{op} in
the bodies of the loops, the initial call to $\p{op}(b,c)$
and the one in the loop invariant, simultaneously.

By contrast, the ability of manipulating nested and mixed arity WPs
of \thelogic{} provides a modular proof of this example.
The goal can be expressed as:
\begin{equation*}
  \begin{conj}[c]
    \WP{\m[\I 4: \p{op}(b\code{,} c)]}{\ret.\ret[\I 4] = d}
    \and
    \p{r}(\I1) = \p{r}(\I2) = \p{r}(\I3) = 0
    \and
    \p{i}(\I1) = \p{i}(\I2) = \p{i}(\I3)
  \end{conj}
  \proves
    \WP
    {\m< \I 1: \p{f}(a\code{,} d)
       , \I 2: \p{f}(a\code{,} b)
       , \I 3: \p{f}(a\code{,} c)
    >}*{
      \begin{array}{l}
      \E v_1,v_2,v_3.\\
      \quad
      \begin{conj}
        \p{r}(\I1) = v_1 \land
        \p{r}(\I2) = v_2 \land
        \p{r}(\I3) = v_3
        \\
        \WP{\m[\I 5: \p{op}(v_2,v_3)]}{
          \ret. \ret[\I 5] = v_1
        }
      \end{conj}
      \end{array}
    }
\end{equation*}
Note the use of nested WPs to invoke \p{op} on~$b$ and~$c$,
and on the results of components~\I2 and~\I3.

Using the same style, we can represent (and prove!)
the desired loop invariant:
\[
    \begin{conj}[c]
      \p{i}(\I1) = \p{i}(\I2) = \p{i}(\I3)
      \and
      \WP{\m[\I 4: \p{op}(b\code{,} c)]}{\ret.\ret[\I 4] = d}
    \end{conj}
    \land
    \E v_1,v_2,v_3.
\begin{conj}[c]
    \p{r}(\I1) = v_1 \land
    \p{r}(\I2) = v_2 \land
    \p{r}(\I3) = v_3
    \and
      \WP{\m[\I 5: \p{op}(v_2,v_3)]}{
        \ret. \ret[\I 5] = v_1
      }
    \end{conj}
\]

\subsection{Parsimonious Use of Refinement}
\label{sec:discuss-refinement}

Relational logics in the literature try to overcome some of the limitations
of lockstep reasoning by using refinement
to allow richer alignments between programs.
Our observation is that the required refinements can be recast in many cases
as hypersafety proofs.
\thelogic\ is flexible enough to prove these refinements \emph{within} the logic,
and to apply them using the hyper-structure rules.

Consider again the program
of \cref{fig:f-code} specialized with $\p{op}(x,y)\is x + y$.
We can prove that the program enjoys distributivity
on the \emph{first} argument as well.
In the form of a hyper-triple:\begin{equation}
  \J |- {\LAnd_{j\in\set{\I1,\I2,\I3}} \p{r}(j)=\p i(j)=0}
        {\m[\I 1: \p{f}(a+b, c),
          \I 2: \p{f}(a, c),
          \I 3: \p{f}(b, c)
        ]}
        {\p{r}(\I 1) = \p{r}(\I 2) + \p{r}(\I 3)}
  \label{spec:mult-distrib1}
\end{equation}
All three components in the judgment are while loops,
but they do not iterate the same number of times.
We therefore cannot apply the lockstep while rule.
Intuitively, we want to argue that the first~$a$ iterations
of component~\I{1} are matched exactly by component~\I{2},
and the remaining~$b$ iterations by component~\I{3}.
To formally implement the above strategy,
we replace component~\I{1} with an equivalent program
(here at index~\I4)
that splits the while loop into two loops:
\begin{equation}
\J |- {\LAnd_{j\in\set{\I1,\I2,\I4}} \p{r}(j)=\p i(j)=0}
        {\m[\I 4: \p{f}(a, c)\p;\p{f}(a{+}b, c),
          \I 2: \p{f}(a, c),
          \I 3: \p{f}(b, c)
        ]}
        {\p{r}(\I 4) = \p{r}(\I 2) + \p{r}(\I 3)}
  \label{spec:mult-distrib1-split}
\end{equation}
The term at~\I{4} may look like a typo, but it is not:
\p{f} does not initialize its variables,
so the second call will find in \p{i} the value left by the previous call.
It is easy to discharge this version of the hyper-triple
with a lockstep derivation.

To close the gap between
\eqref{spec:mult-distrib1} and~\eqref{spec:mult-distrib1-split}
we need to prove the equivalence between component~\I{1} and~\I{4},
which can be encoded by the auxiliary hyper-triple:
\begin{equation}
\J |- {\p{r}(\I1)=\p{r}(\I4) \\ \p{i}(\I1)=\p{i}(\I4)}
        {\m<\I 1: \p{f}(a+b\code{,} c),
          \I 4: \p{f}(a\code{,} c)\p;\p{f}(a+b\code{,} c)
        >}
        {\p{r}(\I1)=\p{r}(\I4) \\ \p{i}(\I1)=\p{i}(\I4)}
  \label{spec:mult-split}
\end{equation}

Then we would be able to apply~\ref{rule:wp-proj} and \ref{rule:wp-cons}
to replace component~\I{4} for component~\I{1} in our original goal:
\[\infer*[right=\ref{rule:wp-proj}]{
  \infer*[Right=\ref{rule:wp-conj}]{
    \eqref{spec:mult-distrib1-split}\;
    \J |- {\dots} {\m[\I4: \p f(a,c)\p;\p f(a{+}b,c),\I2: \p f(a,c),\I3: \p f(b,c)]} {\dots}
    \\
    \eqref{spec:mult-split}\;
    \J |- {\dots} {\m[\I1: \p f(a{+}b,c),\I4: \p f(a,c)\p;\p f(a{+}b,c)]} {\dots}
  }{
    \J |- {\dots}
          {\m[\I1: \p f(a{+}b,c),\I2: \p f(a,c),\I3: \p f(b,c),\I 4: \p f(a,c)\p;\p f(a{+}b,c)]}
          {\dots}
  }}{
    \eqref{spec:mult-distrib1}\;
    \J |- {\dots}
          {\m[\I1: \p f(a{+}b,c),\I2: \p f(a,c),\I3: \p f(b,c)]}
          {\dots}
  }
\]

We have thus reduced the original goal to proving~\eqref{spec:mult-split}.
The derivation, shown in \appendixref{sec:ex-loop-splitting},
uses the \ref{rule:wp-refine} rule only once
using the most basic refinement for while loops,
loop unfolding:
$
  \code{while$\;g\;$do$\;t$}
  \semeq
  \code{if$\;g\;$then$\;(t$;while$\;g\;$do$\;t)$}
$.

As we noted in \cref{sec:logic},
the hyper-triple encoding of equivalence for non-deterministic programs
is stronger than semantic equivalence.
This means that for some programs
the above proof pattern might not allow replacing some~$t_1$ for some~$t_2$,
even when $t_1 \semleq t_2$ holds.
A precise characterization of the limits of this proof pattern is an interesting
direction of future research.
For instance, one could ask which proofs are possible in our logic
if the \ref{rule:wp-refine} rule is \emph{replaced} by the special case
of unfolding one iteration of a loop.

\subsection{Divide and Conquer}
\label{sec:discuss-divide-and-conquer}
The hyper-structure rules of \thelogic{} are a useful tool to
decompose proofs into more tractable sub-goals.
The advantages of this are two-fold.
First, in some cases, they permit the proof engineer
to perform the formal decomposition following the same line
of a natural informal argument.
Second, smaller/simpler subgoals
are more likely provable using off-the-shelf automatic provers.

Take again the example in \cref{fig:f-code},
fixing $\p{op}(x,y)\is x + y$.
We have shown how to prove the two distributivity properties,
which we can informally summarize as
$\p{f}(a,b+c) = \p{f}(a,b) + \p{f}(a,c)$
and
$\p{f}(a+b,c) = \p{f}(a,c) + \p{f}(b,c)$.
These two properties imply distributivity on \emph{both} arguments:
$\p{f}(a+b,c+d) = \p{f}(a,c) + \p{f}(b,c) + \p{f}(a,d) + \p{f}(b,d)$.
The intuitive argument breaks the property as the conjunction of
$
  \p{f}(a+b,c+d) = \p{f}(a+b,c) + \p{f}(a+b,d)
$,
$
  \p{f}(a+b,c) = \p{f}(a,c) + \p{f}(b,c)
$, and
$
  \p{f}(a+b,d) = \p{f}(a,d) + \p{f}(b,d)
$,
which are instances of the distributivity properties we already established.
\thelogic\ can replicate the same simple argument,
on the hyper-triple encodings of these informal equations.
The conjunction of the three equations above can be handled using \ref{rule:wp-conj}.
\Cref{rule:wp-proj} can remove the auxiliary components
running the terms~$\p{f}(a+b,c)$ and $\p{f}(a+b,d)$
(see \appendixref{sec:distr-both}).
This illustrates how our logic allows the top-level goal
to be decomposed into simpler goals within the logic,
before the programs in the hyper-triples are even considered.

\subsection{Hyper-triples as Goals and as Assumptions}
\label{sec:discuss-idempotence}

Existing work on relational logics in the literature does not discuss the general problem of how to encode an arbitrary verification goal as a hyper-triple.
Conceptually, the high-level goal can often be formulated as some
semi-formal assertion involving ``calls'' to code
(e.g.~$ \p{f}(x) \leq \p{g}(x) $);
then a formalization in terms of hyper-triples can be obtained by
introducing a component for each ``call''
(e.g.~$ \V |- \WP{\m[\I1: \p{f}(x), \I1: \p{g}(x)]}{\ret.\ret[\I1]\leq\ret[\I2]} $).
While this procedure seems obvious enough in simple instances,
it can be surprisingly ambiguous in general.
Consider the case of idempotence,
which we already partially addressed in \cref{sec:overview:reindex}.
Informally, the goal looks like $t \sim (t\p; t)$,
which can plausibly be formalized in (at least) three different ways:
\begin{align}
  &\J |- { \pv{x}(\I1)=\pv{x}(\I2) }
        {\m[\I 1: t, \I 2: (t\p; t)]}
        { \pv{x}(\I1)=\pv{x}(\I2) }
  \tag{\ref{spec:t-idemp-seq}}
  \label{spec:t-idemp2-seq}
  \\
  &\J |- { \pv{x}(\I1)=\pv{x}(\I2) \land \pv{x}(\I3)=\vec{v} }
        {\m[\I 1: t, \I 2: t, \I 3: t]}
        { \pv{x}(\I2)=\vec{v} \implies
            \pv{x}(\I1)=\pv{x}(\I3) }
  \tag{\textsc{Idem}$^3_t$}
  \label{spec:t-idemp3}
  \\
  &\J |- { \pv{x}(\I2)=\vec{v} }
        {\m[\I 1: t, \I 2: t]}
        { \pv{x}(\I1)=\vec{v} \implies \pv{x}(\I2)=\vec{v} }
  \tag{\ref{spec:t-idemp}}
  \label{spec:t-idemp2-ndet}
\end{align}

The most verbatim translation of the informal goal is~\eqref{spec:t-idemp2-seq}.
The hyper-triple~\eqref{spec:t-idemp3} uses indirect-style to feed the output of \I2 as the input of \I3,
and considers each occurrence of~$t$
in the informal equation $t \sim (t\p; t)$ as a separate component.
\citet{SousaD16} propose~\eqref{spec:t-idemp2-ndet}
which also uses indirect-style,
but conflates components \I1 and \I2 of~\eqref{spec:t-idemp3}
into the single component \I1.

A natural question arises about the differences between these formulations.
Thanks to the generality of \thelogic{} rules,
this question can be investigated within the logic.
For example, we can prove:
\begin{enumerate*}[label=(\roman*)]
  \item $\eqref{spec:t-idemp3} \implies \eqref{spec:t-idemp}$
  \item $\eqref{spec:t-idemp3} \implies \eqref{spec:t-idemp-seq}$
  \item $
    \bigl( \eqref{spec:t-idemp} \land (\textsc{Det}_t) \bigr)
      \implies \eqref{spec:t-idemp3}
  $
  \item $
    \bigl( \eqref{spec:t-idemp3} \land \proj(t) \bigr)
      \implies (\textsc{Det}_t)
  $
\end{enumerate*}
where~(\textsc{Det}$_t$) is the hyper-triple asserting
determinism of~$t$ on $\pv{x}$.
All the implications in the other directions do not hold.
All the details are in \appendixref{sec:specs-for-idemp}.

As we explained in \cref{sec:overview:reindex},
even in the cases where the specifications are equivalent,
there is a tension between which hyper-triples are easier to \emph{prove},
\ie they are better as goals, and the ones which are easier to \emph{use},
\ie they are better as assumptions.
There is no silver bullet: different proofs may need idempotence
to be expressed as relating occurrences of~$t$ in the same component, or two or three components.
If the logic cannot perform the inter-derivations we just examined,
choosing one specification over the other
may make it impossible to prove some valid goals.
\thelogic{} is unique among relational logics in supporting the derivation of one specification from the other.

\subsection{Hyper-triple Semantics and Termination}
\label{sec:discuss-design}

Some relational logics, notably~\cite{Benton04,Yang07},
propose an ``equi-termination'' semantics:
given a hyper-store satisfying the precondition,
a component can diverge only if all the others do.
In a language with non-determinism, the more appropriate semantics
would be ``\emph{may} equi-termination'':
if a component \emph{may} terminate, then all the others \emph{may} terminate as well.
Notice that the (may) equi-termination semantics encodes a property that
goes beyond \pre k-safety: it quantifies over traces existentially.

The equi-termination semantics has some useful consequences for the rules.
For example, the \ref{rule:wp-proj} rule would not need the projectability side condition.
The \ref{rule:wp-conj} rule, on the other hand, would only be sound if
$\m{t}_1$ and $\m{t_2}$ overlap on at least one component.
In fact the [R-Tr] rule of \cite{Benton04} can be seen as a binary special case
of a combination of \ref{rule:wp-proj} and overlapping \ref{rule:wp-conj}
which is only sound in the equi-termination semantics.

The equi-termination semantics nicely supports the soundness
of the lockstep \ref{rule:wp-while} rule:
if the guards are always either all true or all false,
the while loops are equi-terminating;
note that this is weaker than proving termination of the loops
(no variant-based reasoning is needed).

The equi-termination semantics, however, does not validate
some important rules like the left to right direction of \ref{rule:wp-nest},
and the \ref{rule:wp-idx-post} rule.

A more powerful extension of our logic with support for equi-termination WPs would annotate hyper-triples with an equivalence relation~$\approx$
on its indexes.
An assertion $ \wpsymb_\approx\,\m{t}\,\{Q\}$ would---in addition to the constraints of the current definition---require that, for every hyper-store~$\m{s}$ satisfying~$P$,
if $\bigsome{\m{t}(i)}{\m{s}(i)}$ and $i \approx j$,
then $\bigsome{\m{t}(j)}{\m{s}(j)}$.
This would make it feasible to derive the most precise equi-termination
guarantees each rule can provide,
given the equivalence relations that hold in the premises.
In \ref{rule:wp-proj}, the projectability side condition
could be then replaced by $\forall i\in I\st \exists j \notin I. i\approx j$.
More ambitiously, as we speculate in \cref{sec:conclusion},
termination-related properties could become
first-class citizens in a logic for \emph{hyperliveness}.
 \section{Related work}
\label{sec:related-work}

There is a large body of work on relational verification.
Here, we survey and compare against the most closely related work.

\paragraph{Relational Hoare Logic and product programs}
The seminal Relational Hoare Logic~(RHL) of~\cite{Benton04}
is a program logic with judgments on pairs of programs supporting
variations of lockstep proofs.
RHL and its extension to Separation Logic~\cite{Yang07}
have been used to prove a wide range of relational properties
of first-order imperative programs, from correctness of optimisations
to information-flow.
\citet{BartheDR04}
introduced self-composition
as a relatively complete method to reason about properties like
non-interference and information flow,
an approach further developed in
\citet{DarvasHS05,BeringerH07,Beringer11,TerauchiA05}.
Building on this idea,
\citet{BartheCK11} proposed product programs as a way to
reduce relational verification of \pre 2-properties to standard Hoare logic reasoning on a single program, without losing the relational flavour of the proofs.
The technique has been incorporated in program logics and extended to
handle probabilistic properties~\cite{BartheGHS17,BartheKOB13}
like differential privacy or security of cryptographic code.
\citet{KovacsSF13} incorporates this approach in abstract interpretation.
Cartesian Hoare Logic~(CHL)~\cite{SousaD16}
is an extension of RHL to handle \pre k-safety properties,
for fixed~$k$.
The utility of going beyond \pre 2-properties is shown by
providing hyper-triple-like specifications for common correctness checks
that arise in application domains like comparators of user-defined types,
\eg transitivity, associativity, and commutativity.
The problem of alignment of while loops is addressed through
a limited set of rules that implicitly construct product programs through derivation sequences. \citet{GodlinS13} proposed \emph{regression verification} for proving equivalence between two \emph{similar} programs using automatically inferred alignments.

We build on these logics by introducing patterns of reasoning
that involve hyper-triples of mixed arities---i.e.~transforming a hyper-triple goal of arity~$k$ into hyper-triple subgoals with arities potentially different from~$k$.
No prior relational logics support such reasoning.
See \cref{sec:logic} and \cref{sec:discuss-design}
for a more technical distinction between \thelogic's design decisions
and other logics.

We briefly commented on relative completeness of \thelogic\ in
\cref{sec:hyper-str-rules}.
\citet{NagasamudramN21} studies completeness of various combinations
of enhanced lockstep binary rules in terms of the classes of product programs
(called alignments) they can represent.
Characterising, in the same spirit, the
new classes of alignments expressible through \thelogic,
is interesting future work.

\paragraph{Logics for proving refinement}
Another set of program logics commonly labelled as ``relational''
are the ones which aim to prove refinement between two programs,
often with a focus on concurrency,
e.g.~\cite{LiangFF12,FruminKB18}.
Refinement is a property that quantifies both universally
and existentially over traces, and is thus not a hypersafety property (since the latter quantifies universally over all traces).
Moreover, all refinement judgments are between exactly two programs.
However, since proving refinement often retains the ``relational'' flavour of proofs that we seek in proving hypersafety,
exploring a single logic encompassing
both kinds of relational reasoning is an interesting path for future research---see~\cref{sec:conclusion}.

\paragraph{Higher-order logics}
Relational proofs can also be conducted in a higher-order logic framework.
Approaches like Relational Higher-Order Logic~(RHOL)~\cite{AguirreBGGS19}
propose rules that can derive relational proofs of \pre 2-safety properties
of pure higher-order functional programs.
Notably, RHOL can embed various quantitative/probabilistic extensions
as special cases of the general framework.
This approach tends to work well for programs that can be expressed directly
as pure expressions in the host meta-logic (e.g.~Coq).
\thelogic\ obtains compositional relational reasoning in a \emph{first-order}
program logic that supports direct treatment of non-deterministic impure code.
We establish WPs on hyper-terms as the right building block to embed
such programs in assertions.

\paragraph{Automation}
While we do not consider automation in this paper,
part of the motivation for considering relational proofs
is that they are more feasible for automation,
 whereas functional-specification-based proofs of the same properties
would be completely out of reach of current algorithms. 
Consider, as an example, the relational proof of distributivity
of multiplication (e.g.~our running example~\p{f} when~\p{op} is addition).
The proof can be done by exclusively using conjunctions of equalities and additions, which is within the power of Presburger arithmetic, a decidable theory.
A functional-specification-based proof of the same property
would require the use of mathematical multiplication in assertions,
which requires the use of undecidable Peano arithmetic. 

This observation has been exploited in various works
that use product programs~\cite{BartheCK11,EilersMH20}.
In~\cite{FarzanV19} this idea is taken even further by showing
that it is possible to automate the \emph{search of appropriate product programs}
itself, even in the more general case of \pre k-safety.
\citet{ShemerGSV19} and \citet{UnnoTK21} push the boundaries of automatically inferrable enhanced lockstep proofs, with the latter supporting hyper-liveness too.

All these works are based on whole-program enhanced lockstep proofs,
leading to searches that are global and do not scale well beyond small programs. Our logic outlines how a monolithic task like this can be effectively decomposed.
It would be possible to integrate these tools in our logic
with the goal of discharging as many sub-hypertriples involved in a proof as possible.
The compositionality achieved by our proof system can make 
these tools applicable to smaller, more tractable instances,
increasing the scalability of the overall method.

\paragraph{Model checking}
HyperLTL and HyperCTL$^*$~\cite{ClarksonFKMRS14}
are temporal logics which have been proposed
as specification formalisms for hyperproperties
(with \pre k-safety as a special case). MultiLTL~\cite{GoudsmidGS21} extends the idea to the case where a hyperproperty
property involves~$k$ different finite-state models. They target finite-state systems, and therefore substantially differ in scope compared to this work. Temporal logics specify hyperproperties as global properties of sets of traces,
which does not lead naturally to the kind of decomposition of proofs
that we seek to obtain with \thelogic. The hope is that compositional program logics like ours can inspire interesting decomposition heuristics for these monolithic verification tasks in the same style that compositional concurrent program logics inspired algorithms for model checking of concurrent systems \cite{GuptaPR11,McMillan99a,FlanaganQS02}.

 \section{Conclusion and Future work}
\label{sec:conclusion}

We introduced \thelogic, proved it sound,
and showed how a handful of new rules
can unlock important compositional proof principles for hypersafety.

\thelogic{} is just a first step in exploring
the change of perspective of embracing hyper-triples as proof building blocks.
A first interesting line for future work is to investigate how the
compositionality and locality dimension added to proofs by Separation Logic~\cite{Yang07}
interacts with the hyper-triple compositions of \thelogic.
In addition to supporting proofs of heap-manipulating programs,
such an extension could pave the way to support for parallel composition.

Hypersafety properties such as determinism, non-interference, and correctness of parallelizing optimizations are examples of important applications
of an extension of \thelogic{} supporting concurrency.
The key challenge would be to blend rely/guarantee-style reasoning
together with the relational flavour of hypersafety verification.
In particular, obtaining rely/guarantee proofs that avoid the use of strong functional specifications looks like a very intriguing research problem.

An important next step is exploring automation.
The rules of \thelogic{} generate a search space that is different from those of
previous automated work~(\eg \cite{FarzanV19,UnnoTK21}) in two key senses:
  (1) As a consequence of expressivity of \thelogic{},
      the search space is larger; and
  (2) as a counter effect to (1),
      its structure offers enticing opportunities for proof reuse and
      compositional arity-changing proof patterns.
None of the heuristics used in previous relational logics
handle this dimension of the search.

An even more ambitious direction for future research is the extension
of \thelogic{} to reason about hyperproperties beyond hypersafety,
such as \emph{hyperliveness}~\cite{ClarksonS08}.
For example, it would be interesting to extend hyper-triples
from hypersafety statements of the form $\forall^n$ to
hyperliveness statements of the form $\forall^n\exists^m$.
As an enticing by-product, such a logic would embed both
\pre n-safety triples \emph{and refinement}
(which is a $\forall\exists$ hyperproperty)
in a single generalized judgment.
Understanding the proof principles available for such a judgment
would be a way to unify hypersafety-based and refinement-based relational reasoning.
 
\begin{acks}
  The authors would like to thank
  Deepak Garg and Viktor Vafeiadis for their useful comments
  on an early version of the paper.

  This work was supported by a
  \grantsponsor{ERC}{European Research Council}{https://erc.europa.eu/} (ERC)
  Consolidator Grant for the project ``PERSIST'' under the European Union's
  Horizon 2020 research and innovation programme
  (grant agreement No.~\grantnum{ERC}{101003349}).
\end{acks}

\setlabel{LAST}
\label{paper-last-page}

\appendix

\section{Program semantics}
\label{sec:program-semantics}

We specify the programming language using a big-step semantics
with judgment $\bigstep{t}{s}{v}{s'}$.

\begin{mathpar}
  \let\RightTirName\RuleNameLbl
  \infer*{ }{
    \bigstep{\code{skip}}{s}{v}{s}
  }

  \infer*{ }{
    \bigstep{v}{s}{v}{s}
  }

  \infer*{ }{
    \bigstep{\p{x}}{s}{s(\p{x})}{s}
  }

  \infer*{ }{
    \bigstep{*}{s}{v}{s}
  }

  \infer*{
    \bigstep{e_1}{s}{v_1}{s'}
    \\
    \bigstep{e_2}{s'}{v_2}{s''}
  }{
    \bigstep{e_1 \oplus e_2}{s}{\sem{\oplus}(v_1, v_2)}{s''}
  }

  \infer*{
    \bigstep{e}{s}{v}{s'}
  }{
    \bigstep{x=e}{s}{v}{s'\m[\p{x}: v]}
  }

  \infer*{
    \bigstep{t_1}{s}{\_}{s'}
    \\
    \bigstep{t_2}{s'}{\_}{s''}
  }{
    \bigstep{t_1\p; t_2}{s}{v}{s''}
  }

  \infer*{
    \bigstep{g}{s}{b}{s'}
    \\
    b \ne 0
    \\
    \bigstep{t_1}{s'}{v}{s''}
  }{
    \bigstep{\code{if$\;g\;$then$\;t_1\;$else$\;t_2\;$}}{s}{v}{s''}
  }

  \infer*{
    \bigstep{g}{s}{b}{s'}
    \\
    b = 0
    \\
    \bigstep{t_2}{s'}{v}{s''}
  }{
    \bigstep{\code{if$\;g\;$then$\;t_1\;$else$\;t_2\;$}}{s}{v}{s''}
  }

  \infer*[right=ws$_0$]{
    \bigstep{g}{s}{0}{s'}
  }{
    \bigstep{\code{while$\;g\;$do$\;t$}}{s}{v}{s'}
  }
  \label{rule:while-sem-0}

  \infer*[right=ws$_1$]{
    \bigstep{g}{s_0}{b}{s_1}
    \\
    b \ne 0
    \\
    \bigstep{t}{s_1}{\_}{s_2}
    \\
    \bigstep{\code{while$\;g\;$do$\;t$}}{s_2}{v}{s_3}
  }{
    \bigstep{\code{while$\;g\;$do$\;t$}}{s_0}{v}{s_3}
  }
  \label{rule:while-sem-1}
\end{mathpar}

The primitive operations~$\oplus$ have standard meaning:
\begin{align*}
  \sem{\p+}(v_1,v_2) & \is v_1 + v_2
  &
  \sem{\p-}(v_1,v_2) & \is v_1 - v_2
  &
  \sem{\p<}(v_1,v_2) & \is
    \begin{cases}
      1 \CASE v_1 < v_2\\
      0 \OTHERWISE
    \end{cases}
\end{align*}

 \section{Omitted rules}
\label{sec:omitted-rules}

We present here some omitted valid rules.
\Cref{fig:wp-primitives} shows rules to handle primitive operations;
they are entirely standard.

\begin{figure}[!htbp]
  \adjustfigure[\small]
  \begin{proofrules}
  \infer*[lab=wp-prim$_{\oplus}$]{}{
  \V
  \WP {\m[i: e_1]}{
    \fun\m{v}_1.
    \WP {\m[i: e_2]}{
      \fun\m{v}_2.
      \WP {\m[i: \m{v}_1(i) \oplus \m{v}_2(i)]}{Q}
    }
  }
  |-
  \WP {\m[i: e_1 \oplus e_2]}{Q}
}   \label{rule:wp-primop-eval}

  \infer*[lab=wp-prim$_{\oplus}$]{}{
  \V |- \WP {\m[i: v_1 \oplus v_2]}
            {\ret. \ret[i] = (v_1 \mathbin{\sem{\oplus}} v_2)}
}   \label{rule:wp-primop}

  \infer*[lab=wp-val]{}{
  \V |- \WP {\m[i: v]}{\ret.\ret[i] = v}
}   \label{rule:wp-val}

  \infer*[lab=wp-var]{}{
  \V |- \WP { \m[i: \p{x}] } { \ret. \ret(i) = \p{x}(i) }
}   \label{rule:wp-var}
  \end{proofrules}
  \caption{Weakest precondition laws: primitives.}
  \label{fig:wp-primitives}
\end{figure}

\Cref{fig:proj-t-laws} show two rules that explain the relation
between $\proj(\m{t})$ and WPs on~$\m{t}$.
\Cref{rule:proj-elim} says that when proving $\WP{\m{t}}{Q}$,
the assumption of~$\proj(\m{t})$ is redundant:
when $\m{t}$ is not projectable, the WP is vacuously true for any~$Q$.
By using projectability, we can also provide \ref{rule:wp-elim},
which can be seen as a bi-directional version of frame.
If~$P$ does not predicate on variables that are modified by~$\m{t}$,
whether~$P$ is asserted on the current hyper-store or in the output
hyperstore of~$t$ should not matter; provided, that is, that~$\m{t}$ is projectable.

\begin{figure}[!htbp]
  \adjustfigure[\small]
  \begin{proofrules}
    \infer*[lab=proj-elim]{}{
  \proj(\m{t}) \implies \WP {\m{t}} {Q}
  \proves
  \WP {\m{t}} {Q}
}
     \label{rule:proj-elim}

    \infer*[lab=wp-elim]{
  \pvar(P) \inters \mods(\m{t}) = \emptyset
}{
  \WP {\m{t}} {P}
  \lequiv
  \proj(\m{t}) \implies P
}
     \label{rule:wp-elim}
  \end{proofrules}
  \caption{Projectability rules}
  \label{fig:proj-t-laws}
\end{figure}

In fact, as shown in \cref{sec:soundness},
the rules in \cref{fig:proj-t-laws} can be used to derive
\ref{rule:wp-triv},
\ref{rule:wp-const}, and \ref{rule:wp-proj}. 

\begin{figure}[!htbp]
  \adjustfigure[\small]
  \begin{proofrules}
    \infer*[lab=wp-skip]{}{
  \WP {(\m{t} \m. \m[i: \code{skip}])} {Q}
  \lequiv
  \WP {\m{t}} {Q}
}     \label{rule:wp-skip}

    \infer*[lab=wp-empty]{}{
  P
  \lequiv
  \WP {\m[]} {P}
}
     \label{rule:wp-empty}

    \infer*[lab=wp-impl-l]{
  \pvar(P) \inters \mods(\m{t}) = \emptyset
}{
  \V
  \proj(\m{t}),
  \WP{\m{t}}{\ret. Q(\ret) \implies P}
  |-
  \WP{\m{t}}{Q} \implies P
}     \label{rule:wp-impl-l}

\infer*[lab=wp-rename]{
  \V \Gamma |- \WP{\m{t}}{Q}
  \\
  \pi\ \mathrm{bij.}
}{
  \V \Gamma\isub*{\pi} |- \WP{\m{t}\isub*{\pi}}{Q\isub*{\pi}}
}     \label{rule:wp-rename}

    \infer*[lab=wp-seq-plus]{
  \V \Gamma
  |- \WP{(\m[i: t_i | i \in I] \m. \m{t}_1)}{
        \WP{(\m[i: t_i' | i \in I] \m. \m{t}_2)}{Q}
     }
}{
  \V \Gamma
  |- \WP{(\m[i: (t_i\code{;}\ t'_i) | i \in I] \m. \m{t}_1 \m. \m{t}_2)}{Q}
}     \label{rule:wp-seq-plus}

\infer*[
  lab=wp-indirect,
  Right={$\begin{smallmatrix*}[l]
    j \notin \idx(\Gamma)
    \\
    j \notin \supp(\m{t})
  \end{smallmatrix*}$}
]{
  \V \Gamma |- \WP{\m[i: t_1]}{\E v. \at{A(v)}{i}}
  \\
  \forall v\st\ \V
  \Gamma, \at{A(v)}{j}
  |-
  \WP{(\m[i: t_1, j: t_2] \m. \m{t})}{
    \at{A(v)}{i} \implies Q(v)
  }
}{
  \V
  \Gamma
  |-
  \WP{(\m[i: t_1] \m. \m{t})}*{
      \E v.
      \at{A(v)}{i}
      \land
      \WP{\m[i: t_2]}{
        \bigl( \PP i. Q(v) \bigr)\isub{j->i}
      }
  }
}
     \label{rule:wp-indirect}
\end{proofrules}
  \caption{Derived rules for weakest precondition.}
  \label{fig:wp-derived-laws}
\end{figure}

The rules in \cref{fig:wp-derived-laws} collect some useful derived rules.
The most interesting is \ref{rule:wp-indirect} which represents
the general pattern that converts an indirect-style hyper-triple
into a direct style one.
 \section{Soundness}
\label{sec:soundness}

The soundness of many rules is standard or straightforward.
We detail the proofs of the non-standard rules.

We first prove some simple auxiliary lemmas.
We use the following notation.
For~$X \subs (\PVar \times \Idx)$,
and hyper-stores~$\m{s},\m{s}'$,
the relation $ \m{s} =_X \m{s}' $ states that
$ \forall (\p{x}, i) \in X \st \m{s}(i)(\p{x}) = \m{s}'(i)(\p{x}) $.

\begin{lemma}
\label{lm:pvar-P-agree}
  If $P(\m{s})$ and $\m{s} =_{\pvar(P)} \m{s}'$
  then $P(\m{s}').$
\end{lemma}
\begin{proof}
  Suppose $P(\m{s})$ holds but $P(\m{s}')$ does not.
  Then there exists a $(\p{x},i)$ such that
  $ \m{s}(i)(\p{x}) \ne \m{s}'(i)(\p{x}) $.
  Then, by definition, $ (\p{x},i) \in \pvar(P) $.
  But the hypothesis that $\m{s} =_{\pvar(P)} \m{s}'$
  would then imply that
  $ \m{s}(i)(\p{x}) = \m{s}'(i)(\p{x}) $,
  which is a contradiction.
\end{proof}

\begin{lemma}
\label{lm:non-mods-t-agree}
  If $ X \inters \mods(\m{t}) = \emptyset $ and
  $\bigstep{\m{t}}{\m{s}}{\m{v}}{\m{s}'}$,
  then $\m{s} =_{X} \m{s}'$.
\end{lemma}
\begin{proof}
  Straightforward.
\end{proof}

We start with the projectability rules in \cref{fig:proj-t-laws}
as they are useful to derive other rules.

\begin{lemma}[Soundness of~\ref{rule:proj-elim}]
\label{lm:rule:proj-elim}
  If $
    \proj(\m{t})(\m{s}) \implies (\WP {\m{t}} {Q})(\m{s})
  $ then $
    (\WP {\m{t}} {Q})(\m{s}).
  $
\end{lemma}
\begin{proof}
  We have to prove $
    \forall \m{v}, \m{s}' \st
      \bigstep{\m{t}}{\m{s}}{\m{v}}{\m{s}'} \implies Q(\m{v})(\m{s}')
  $.
  Assume $ \bigstep{\m{t}}{\m{s}}{\m{v}}{\m{s}'} $.
  Then, by definition, $ \proj(t)(\m{s}) $ holds.
  We can therefore apply our hypothesis and get
  $ Q(\m{v})(\m{s}') $ as desired.
\end{proof}

\begin{lemma}[Soundness of~\ref{rule:wp-elim}]
\label{lm:rule:wp-elim}
  If
  $\pvar(P) \inters \mods(\m{t}) = \emptyset$
  then,
  $
    {\WP {\m{t}} {P}
    \lequiv
    \proj(\m{t}) \implies P}.
  $
\end{lemma}
\begin{proof}
  We prove the two directions separately.
  \begin{casesplit}
    \case*[Direction $(\proves)$]
      From $\proj(\m{t})(\m{s})$ we obtain that
      $\bigstep{\m{t}}{\m{s}}{\m{v}}{\m{s}'}$
      for some~$\m{v}$ and~$\m{s}'$.
      From the hypothesis $(\WP {\m{t}} {P})(\m{s})$
      we get that~$ P(\m{s}') $.
      By \cref{lm:non-mods-t-agree}~$\m{s} =_{\pvar(P)} \m{s}'$,
      so by \cref{lm:pvar-P-agree}, $ P(\m{s}) $.
    \case*[Direction $(\provedby)$]
      Assume
      $\bigstep{\m{t}}{\m{s}}{\m{v}}{\m{s}'}$.
      Then by definition $ \proj(\m{t})(\m{s}) $ holds.
      We can then apply the hypothesis and obtain~$P(\m{s})$.
      By \cref{lm:non-mods-t-agree}~$\m{s} =_{\pvar(P)} \m{s}'$,
      so by \cref{lm:pvar-P-agree}, $ P(\m{s}) $.
      \qedhere
  \end{casesplit}
\end{proof}

\begin{lemma}[Soundness of~\ref{rule:reindex}]
\label{lm:rule:reindex}
  If $ \forall\hat{\m{s}}\st \Gamma(\hat{\m{s}}) \implies P(\hat{\m{s}}) $,
  then $ \forall\m{s}\st \Gamma(\m{s}\isub*{\pi}) \implies P(\m{s}\isub*{\pi}) $
\end{lemma}
\begin{proof}
  Follows immediately by applying the hypothesis
  with $ \hat{\m{s}} = \m{s}\isub*{\pi} $.
\end{proof}

\begin{lemma}[Soundness of~\ref{rule:idx-ex}]
\label{lm:rule:idx-ex}
  $(\E x.P(x))\isub*{\pi} \lequiv \E x.(P(x)\isub*{\pi})$.
\end{lemma}
\begin{proof}
  By definition,
  $
    \bigl((\E x.P(x))\isub*{\pi}\bigr)(\m{s})
    \iff
    \exists v. P(v)(\m{s}\isub*{\pi})
    \iff
    \bigl(\E x.(P(x)\isub*{\pi})\bigr)(\m{s}).
  $
\end{proof}

\begin{lemma}[Soundness of~\ref{rule:proj-intro}]
\label{lm:rule:proj-intro}
  If $ \forall\hat{\m{s}}\st \Gamma(\hat{\m{s}}) \implies P(\hat{\m{s}}) $,
  then \[
    \forall\m{s}\st
      (\exists \m{s}'\st \Gamma(\m{s}\m[i: \m{s}'(i) | i\in I]))
      \implies
      (\exists \m{s}'\st P(\m{s}\m[i: \m{s}'(i) | i\in I])).
  \]
\end{lemma}
\begin{proof}
  Assume $\Gamma(\m{s}\m[i: \m{s}'(i) | i\in I])$.
  Then by the hypothesis with~$\hat{\m{s}} = \m{s}\m[i: \m{s}'(i) | i\in I]$
  we get~$ P(\m{s}\m[i: \m{s}'(i) | i\in I]) $, which proves the statement.
\end{proof}

\begin{lemma}[Soundness of~\ref{rule:wp-const}]
\label{lm:rule:wp-const}
  If $ \pvar(P) \inters \mods(\m{t}) = \emptyset $
  then
  \[
    P \land \WP{\m{t}}{Q}
    \proves
    \WP{\m{t}}{P \land Q}.
  \]
\end{lemma}
\begin{proof}
  By the following derivation:
  \begin{derivation}
    \infer*[right=\ref{rule:wp-conj}]{
      \infer*{
        P \proves \proj(\m{t}) \implies P
        \\
        \infer*[Right=\ref{rule:wp-elim}]{
          \pvar(P) \inters \mods(\m{t}) = \emptyset
        }{
          \proj(\m{t}) \implies P \proves \WP{\m{t}}{P}
        }
      }{
        P \proves \WP{\m{t}}{P}
      }
      \\
      \infer*{}{
        \WP{\m{t}}{Q} \proves \WP{\m{t}}{Q}
      }
    }{
      P \land \WP{\m{t}}{Q}
      \proves
      \WP{\m{t}}{P \land Q}
    }
    \qedhere
  \end{derivation}
\end{proof}

\begin{lemma}[Soundness of~\ref{rule:wp-impl-r}]
\label{lm:rule:wp-impl-r}
  If $\pvar(P) \inters \mods(\m{t}) = \emptyset$
  then \[
    P \implies \WP{\m{t}}{Q}
    \lequiv
    \WP{\m{t}}{\ret. P \implies Q(\ret)}.
  \]
\end{lemma}
\begin{proof}
  We prove the two directions separately.
  \begin{casesplit}
  \case*[Direction $(\proves)$]
    Assume $\bigstep{\m{t}}{\m{s}}{\m{v}}{\m{s}'}$.
    We have to prove
      $P(\m{s}') \implies Q(\m{v})(\m{s}')$.
    Assume~$P(\m{s}')$, then by \cref{lm:non-mods-t-agree,lm:pvar-P-agree}
    we have~$P(\m{s})$.
    By assumption then we have $(\WP{\m{t}}{Q})(\m{s})$
    which gives us~$ Q(\m{v})(\m{s}') $ as desired.
  \case*[Direction $(\provedby)$]
    Assume $P(\m{s})$ and $\bigstep{\m{t}}{\m{s}}{\m{v}}{\m{s}'}$.
    By the hypotesis $ (\WP{\m{t}}{\ret. P \implies Q(\ret)})(\m{s}) $
    we then have that if $P(\m{s}')$ then $ Q(\m{v})(\m{s}') $.
    By \cref{lm:non-mods-t-agree,lm:pvar-P-agree}
    $P(\m{s}')$ holds, and we get the desired postcondition holds.
    \qedhere
  \end{casesplit}
\end{proof}

\begin{lemma}[Soundness of~\ref{rule:wp-idx}]
\label{lm:rule:wp-idx}
  Let $\pi$ be a bijective reindexing.
  Then, \[
    \V (\WP{\m{t}}{Q})\isub*{\pi} |- \WP{\m{t}\isub*{\pi}}{Q\isub*{\pi}}.
  \]
\end{lemma}
\begin{proof}
  Assume $ \bigstep{\m{t}\isub*{\pi}}{\m{s}}{\m{v}}{\m{s}'} $.
  We have to prove~$ Q(\m{v}\isub*{\pi})(\m{s}\isub*{\pi}) $.
  By definition, the hypotesis reads $ (\WP{\m{t}}{Q})(\m{s}\isub*{\pi}) $
  which means that for any $ \hat{\m{v}},\hat{\m{s}}'$,
  if $ \bigstep{\m{t}}{\m{s}\isub*{\pi}}{\hat{\m{v}}}{\hat{\m{s}}'} $,
  then $ Q(\hat{\m{v}})(\hat{\m{s}}') $.
  Since $\pi$ is a bijection, we have
  $ \bigstep{\m{t}\isub*{\pi}}{\m{s}}{\m{v}}{\m{s}'} $
  implies
  $ \bigstep{\m{t}}{\m{s}\isub*{\pi}}{\m{v}\isub*{\pi}}{\m{s}'\isub*{\pi}} $.
  Therefore, by setting
    $\hat{v} = \m{v}\isub*{\pi}$ and
    $\hat{\m{s}}'=\m{s}'\isub*{\pi}$
  we have proven our goal.
\end{proof}

\begin{lemma}[Soundness of~\ref{rule:wp-while}]
\label{lm:rule:wp-while}
  Let
  \begin{align*}
    \m{w} &= \m[i: \code{while}\ g_i\ \code{do}\ t_i | i \in I]  &
    \m{t} &= \m[i: t_i | i \in I]  &
    \m{g} &= \m[i: g_i | i\in I]
  \end{align*}
  If $
    \V P |-
    \WP {\m{g}}[\big]{
      \fun\m{b}.
      (\m{b} =_I 0 \land R)
      \lor
      (\m{b} \ne_I 0 \land \WP {\m{t}} {P})
    }
  $ then $
    \V P |-
    \WP {\m{w}} {R}.
  $
\end{lemma}
\begin{proof}
  Assume $P(\m{s})$ and
  $
    \bigstep{\m{w}}{\m{s}_0}{\wtv}{\m{s}_2}
  $.
  We proceed by induction on the big-step semantics derivation.
  By definition of the big-step semantics,
  that can hold only if, for some $\m{b}$ and~$\m{s}_1$,
  $ \bigstep{\m{g}}{\m{s}}{\m{b}}{\m{s}_1} $.
  By hypothesis, we have two cases.
  \begin{casesplit}
  \case[$ \m{b} =_I 0 \land R(\m{s}_1) $]
    Then, by the big-step semantics, we have $\m{s}_1 = \m{s}_2$
    and $ R(\m{s}_2) $ holds as desired.
  \case[$ \m{b} \ne_I 0 \land (\WP {\m{t}} {P})(\m{s}_1) $]
    Then, by the big-step semantics,
    we have $\bigstep{\m{t}}{\m{s}_1}{\wtv}{\m{s}_1'}$
    and $\bigstep{\m{w}}{\m{s}_1'}{\wtv}{\m{s}_2}$.
    The latter has a lower derivation depth than the judgment we started with,
    so we get, by induction hypothesis, that
    $
      P(\m{s}_1') \implies R(\m{s}_2).
    $.
    From $(\WP {\m{t}} {P})(\m{s}_1)$ we get that
    $ P(\m{s}_1') $, which concludes the proof.
    \qedhere
  \end{casesplit}
\end{proof}

\begin{lemma}[Soundness of~\ref{rule:wp-refine}]
\label{lm:rule:wp-refine}
  If $(t_1 \semleq t_2)(\m{s}(i))$
  and  $(\WP {(\m[i: t_2] \m. \m{t})} {Q})(\m{s})$,
  then $(\WP {(\m[i: t_1] \m. \m{t})} {Q})(\m{s})$.
\end{lemma}
\begin{proof}
  Assume $ \bigstep{\m[i: t_1] \m. \m{t}}{\m{s}}{\m{v}}{\m{s}'} $.
  Then by the refinement, we have
  $ \bigstep{t_2}{\m{s}(i)}{\m{v}(i)}{\m{s}'(i)} $,
  which implies
  $ \bigstep{\m[i: t_2] \m. \m{t}}{\m{s}}{\m{v}}{\m{s}'} $.
  Then the statement follows from the assumed WP.
\end{proof}

\begin{lemma}[Soundness of~\ref{rule:wp-nest}]
\label{lm:rule:wp-nest}
  $
    \WP{\m{t}_1}{\fun\m{v}.
      \WP{\m{t}_2}{\fun\m{w}.
        Q(\m{v}\m.\m{w})}}
    \lequiv
    \WP{(\m{t}_1 \m. \m{t}_2)}{Q}.
  $
\end{lemma}
\begin{proof}
  Let~$I_1 = \supp(\m{t}_1)$ and $I_2 = \supp(\m{t}_2)$;
  since the two hyper-terms are disjoint,
  we have $ I_1 \inters I_2 = \emptyset $.
  We prove the two directions of the statement separately.
  \begin{casesplit}
  \case*[Direction $(\proves)$]
    Assume $ \bigstep{\m{t}_1\m.\m{t}_2}{\m{s}}{\m{v}}{\m{s}'} $.
    Then we have $
      \bigstep
        {\m{t}_1}
        {\m{s}}
        {\m[i: \m{v}(i) | i\in I_1]}
        {\m{s}'\m[i: \m{s}(i) | i\in I_2]}
    $ and $
      \bigstep
        {\m{t}_2}
        {\m{s}'\m[i: \m{s}(i) | i\in I_2]}
        {\m[i: \m{v}(i) | i\in I_2]}
        {\m{s}'}
    $.
    By the hypothesis WP we get $
      Q(\m[i: \m{v}(i) | i\in I_1]\m.\m[i: \m{v}(i) | i\in I_2])(\m{s}')
    $.
    Since $\m{v} = \m[i: \m{v}(i) | i\in I_1]\m.\m[i: \m{v}(i) | i\in I_2]$
    we get the desired postcondition holds.
  \case*[Direction $(\provedby)$]
    Assume $ \bigstep{\m{t}_1}{\m{s}_0}{\m{v}_1}{\m{s}_1} $.
    From that we know~$ \m{s}_0 =_{I_2} \m{s}_1 $ and $\supp(\m{v}_1)=I_1$.
    We have to prove that
    for any $ \m{v}_2$ and $\m{s}_2 $,
    $  \bigstep{\m{t}_2}{\m{s}_1}{\m{v}_2}{\m{s}_2} $
    implies
    $ Q(\m{v}_1\m.\m{v}_2)(\m{s}_2) $.
    Since $\m{s}_1 =_{I_1} \m{s}_2$ and $\supp(\m{v}_2)=I_2$,
    we have
    $
      \bigstep{\m{t}_1\m.\m{t}_2}{\m{s}_0}{\m{v}_1\m.\m{v}_2}{\m{s}_2}.
    $
    By the hypothesis WP we get $
      Q(\m{v}_1\m.\m{v}_2)(\m{s}_2)
    $ which completes the proof.
    \qedhere
  \end{casesplit}
\end{proof}

\begin{lemma}[Soundness of~\ref{rule:wp-conj}]
\label{lm:rule:wp-conj}
  Assume
    $\idx(Q_1) \inters \supp(\m{t}_2) \subs \supp(\m{t}_1)$, and
    $\idx(Q_2) \inters \supp(\m{t}_1) \subs \supp(\m{t}_2)$.
  Then $
    \bigl(
      \WP{\m{t}_1}{Q_1}
      \land
      \WP{\m{t}_2}{Q_2}
    \bigr)
    \proves
    \WP{(\m{t}_1 \m+ \m{t}_2)}{Q_1 \land Q_2}.
  $
\end{lemma}
\begin{proof}
  Assume $
    \bigstep{\m{t}_1 \m+ \m{t}_2}{\m{s}}{\m{v}}{\m{s}'}
  $.
  We then have:
  \begin{align*}
    &\bigstep{\m{t}_1}{\m{s}}
             {\m[i: \m{v}(i) | i\in I_1]}
             {\m{s}'\m[i: \m{s}(i) | i \in I_2 \setminus I_1 ]}
    &
    &\bigstep{\m{t}_2}{\m{s}}
             {\m[i: \m{v}(i) | i\in I_2]}
             {\m{s}'\m[i: \m{s}(i) | i \in I_1 \setminus I_2 ]}
  \end{align*}
  By the two WP in the assumptions,
  we have:
  \begin{align*}
    & Q_1 (\m[i: \m{v}(i) | i\in I_1])
          (\m{s}'\m[i: \m{s}(i) | i \in I_2 \setminus I_1 ])
    &
    & Q_2 (\m[i: \m{v}(i) | i\in I_2])
          (\m{s}'\m[i: \m{s}(i) | i \in I_1 \setminus I_2 ])
  \end{align*}
  By the assumptions on~$\idx(Q_1)$ and~$\idx(Q_2)$,
  we have $ \idx(Q_1) \inters (I_2\setminus I_1)  = \emptyset $
      and $ \idx(Q_2) \inters (I_1\setminus I_2)  = \emptyset $.
  This together with upward-closedness of postconditions on the hyper-return-value, gives us
  $Q_1 (\m{v})(\m{s}')$ and $Q_2 (\m{v})(\m{s}')$
  as desired.
\end{proof}

\begin{lemma}[Soundness of~\ref{rule:wp-proj}]
\label{lm:rule:wp-proj}
  Let $I = \supp(\m{t}_1)$.
  Then
  $
    \P I. \bigl(
      \proj(\m{t}_2) \implies \proj(\m{t}_1)
      \land
      \WP{(\m{t}_1 \m. \m{t}_2)}{Q}
    \bigr)
    \proves
    \WP{\m{t}_2}{\PP I.Q}
  $.
\end{lemma}
\begin{proof}
  By the following derivation:
  \begin{derivation}
    \infer*[right=\labelstep{wp-proj:step1}]{
    \infer*[Right=\ref{rule:wp-elim}]{
    \infer*[Right=\ref{rule:proj-irrel}]{
    \infer*[Right=\ref{rule:proj-intro}]{
    \infer*[Right=\ref{rule:proj-elim}]{
      \proj(\m{t}_2) \implies \proj(\m{t}_1)
      \land
      \proj(\m{t}_1) \implies \WP{\m{t}_2}{\PP I.Q}
      \proves
      \proj(\m{t}_2) \implies \WP{\m{t}_2}{\PP I.Q}
    }{
      \proj(\m{t}_2) \implies \proj(\m{t}_1)
      \land
      \proj(\m{t}_1) \implies \WP{\m{t}_2}{\PP I.Q}
      \proves
      \WP{\m{t}_2}{\PP I.Q}
    }}{
      \P I. \bigl(
        \proj(\m{t}_2) \implies \proj(\m{t}_1)
        \land
        \proj(\m{t}_1) \implies \WP{\m{t}_2}{\PP I.Q}
      \bigr)
      \proves
      \P I. \WP{\m{t}_2}{\PP I.Q}
    }}{
      \P I. \bigl(
        \proj(\m{t}_2) \implies \proj(\m{t}_1)
        \land
        \proj(\m{t}_1) \implies \WP{\m{t}_2}{\PP I.Q}
      \bigr)
      \proves
      \WP{\m{t}_2}{\PP I.Q}
    }}{
      \P I. \bigl(
        \proj(\m{t}_2) \implies \proj(\m{t}_1)
        \land
        \WP{\m{t}_1}{\WP{\m{t}_2}{\PP I.Q}}
      \bigr)
      \proves
      \WP{\m{t}_2}{\PP I.Q}
    }}{
      \P I. \bigl(
        \proj(\m{t}_2) \implies \proj(\m{t}_1)
        \land
        \WP{(\m{t}_1 \m. \m{t}_2)}{Q}
      \bigr)
      \proves
      \WP{\m{t}_2}{\PP I.Q}
    }
  \end{derivation}
  Step~\eqref{wp-proj:step1} weakens the assumption
  using \ref{rule:wp-nest} and \ref{rule:proj-weak}.
\end{proof}

\begin{lemma}[Soundness of~\ref{rule:wp-idx-post}]
\label{lm:rule:wp-idx-post}
  If
  $\V \Gamma |- \WP{\m{t}}{Q}$
  and
  $j \notin \supp(\m{t}) \union \idx(\Gamma)$
  then
  $\V \Gamma |- \WP{\m{t}}{Q\isub{j->i}}$.
\end{lemma}
\begin{proof}
  Take an arbitrary~$\m{s}$.
  Let~$s_0$ and~$s_1$ be such that~$ \m{s} = \m{s}\m[i: s_0, j: s_1] $.
  Assume
    $\Gamma(\m{s}\m[i: s_0, j: s_1])$ and
    $
      \bigstep{\m{t}}{\m{s}\m[i: s_0, j: s_1]}
              {\m{v}}{\m{s}'\m{s}\m[i: s_0', j: s_1']}
    $.
    Since $j\notin\supp(\m{t})$, $j\notin\supp(\m{v})$.
    Note that
    $\m{s}'\isub{j->i} = \m{s}'\m[i: s_0', j: s_0']$
    and
    $\m{v}\isub{j->i} = \m{v}\m[j: \m{v}(i)]$,
    so our goal is to prove that
    $Q(\m{v}\m[j: \m{v}(i)])(\m{s}'\m[i: s_0', j: s_0'])$.

    The assumption gives us that
    for all $ \hat{\m{s}},\hat{\m{s}}',\hat{\m{v}} $,
    if~$ \Gamma(\hat{\m{s}}) $
    and $ \bigstep{\m{t}}{\hat{\m{s}}}{\hat{\m{v}}}{\hat{\m{s}}'} $
    then $ Q(\hat{\m{v}})(\hat{\m{s}}') $.

    Let $\hat{\m{s}}  = \m{s}\m[i: s_0, j: s_0']$,
        $\hat{\m{s}}' = \m{s}'\m[i: s_0', j: s_0']$, and
        $\hat{\m{v}}  = \m{v}$.
    By $j\notin\idx(\Gamma)$, we have $\Gamma(\m{s}\m[i: s_0, j: s_0'])$.
    We also have
    $
      \bigstep{\m{t}}{\m{s}\m[i: s_0, j: s_0']}
              {\m{v}}{\m{s}'\m[i: s_0', j: s_0']}.
    $
    We therefore get that $
      Q(\m{v})(\m{s}'\m[i: s_0', j: s_0'])
    $.
    By upward-closedness of~$Q$ we get $
      Q(\m{v}\m[j: \m{v}(i)])(\m{s}'\m[i: s_0', j: s_0'])
    $.
\end{proof}

\begin{lemma}[Soundness of~\ref{rule:wp-idx-swap}]
\label{lm:rule:wp-idx-swap}
  Assume $i \notin \idx(Q)$.
  Then, \[
    \bigl(\WP{(\m[j: t]\m.\m{t}')}{Q}\bigr)\isub{j->i}
    \proves
    \WP{(\m[i: t]\m.\m{t}')}{Q\isub{j->i}}.
  \]
\end{lemma}
\begin{proof}
  For the assertions to be well-defined we have~$i,j \notin \m{t}'$.
  Take an arbitrary~$\m{s}$.
  Let~$s_0$ and~$s_1$ be such that~$ \m{s} = \m{s}\m[i: s_0, j: s_1] $.
  Assume $
    \bigstep{\m[i: t]\m.\m{t}'}{\m{s}\m[i: s_0, j: s_1]}
            {\m{v}\m[i:v_0,j:\bot]}{\m{s}'\m[i: s_0', j: s_1']}
  $.
  We have to prove that
  $ Q(\m{v}\m[i:v_0,j:v_0])(\m{s}'\m[i: s_0',j: s_0']) $.

  We have $
    \bigstep{\m[j: t]\m.\m{t}'}{\m{s}\m[i: s_0, j: s_0]}
            {\m{v}\m[i:\bot,j:v_0]}{\m{s}'\m[i: s_0, j: s_0']}
  $.
  Since $\m{s}\isub{j->i} = \m{s}\m[i: s_0, j: s_0]$,
  the WP in the assumptions gives us
  $ Q(\m{v}\m[i:\bot,j:v_0])(\m{s}'\m[i: s_0, j: s_0']) $.
  By upward-closedness of~$Q$ we get
  $ Q(\m{v}\m[i:v_0,j:v_0])(\m{s}'\m[i: s_0, j: s_0']) $.
  By $i \notin \idx(Q)$ we get
  $ Q(\m{v}\m[i:v_0,j:v_0])(\m{s}'\m[i: s_0', j: s_0']) $
  as required.
\end{proof}

\begin{lemma}[Soundness of~\ref{rule:wp-idx-pass}]
\label{lm:rule:wp-idx-pass}
  Assume $i,j \notin \supp(\m{t})$.
  Then, \[
    \V (\WP{\m{t}}{Q})\isub{j->i} |- \WP{\m{t}}{Q\isub{j->i}}.
  \]
\end{lemma}
\begin{proof}
  Assume $
    \bigstep{\m{t}}{\m{s}}{\m{v}}{\m{s}'}
  $.
  Since $i,j \notin \supp(\m{t})$,
  we have $i,j\notin \supp(\m{v})$ and $
    \bigstep{\m{t}}{\m{s}\isub{j->i}}{\m{v}}{\m{s}'\isub{j->i}}
  $.
  Therefore, by the WP in the assumptions,
  we have~$Q(\m{v})(\m{s}'\isub{j->i})$
  which proves the statement.
\end{proof}

\begin{lemma}[Soundness of~\ref{rule:wp-idx-merge}]
\label{lm:rule:wp-idx-merge}
  Assume $ i,j \notin \m{t}' $. Then,
  \[
    \bigl(\WP{(\m[i: t, j: t]\m.\m{t}')}{Q}\bigr)\isub{j->i}
    \proves
    \WP{(\m[i: t]\m.\m{t}')}{Q\isub{j->i}}.
  \]
\end{lemma}
\begin{proof}
  Assume $
    \bigstep{\m[i: t]\m.\m{t}'}{\m{s}}{\m{v}}{\m{s}'}
  $.
  We have $ \m{s}\isub{j->i} = \m{s}\m[j: \m{s}(i)] $.
  Then $
    \bigstep{\m[i: t, j: t]\m.\m{t}'}{\m{s}\m[j: \m{s}(i)]}
            {\m{v}\m[j: \m{v}(i)]}{\m{s}'\m[j: \m{s}'(i)]}
  $.
  By the WP in the assumptions, we get
  $
    Q(\m{v}\m[j: \m{v}(i)])(\m{s}'\m[j: \m{s}'(i)])
  $
  which is the desired postcondition.
\end{proof}

\subsection{Derived Rules}

In \cref{fig:wp-derived-laws} we show some derived rules which we use in
examples.
We give a derivation for the non trivial ones.

\begin{lemma}[Soundness of~\ref{rule:wp-seq-plus}]
  If\/ $
    \V \Gamma
    |- \WP{(\m[i: t_i | i \in I] \m. \m{t}_1)}{
          \WP{(\m[i: t_i' | i \in I] \m. \m{t}_2)}{Q}
       }
  $, then $
    \V \Gamma
    |- \WP{(\m[i: (t_i\code{;}\ t'_i) | i \in I] \m. \m{t}_1 \m. \m{t}_2)}{Q}.
  $
\end{lemma}
\begin{proof}
  By the following derivation:
  \begin{derivation}
    \infer{
    \infer{
    \infer{
      \V \Gamma
      |- \WP{(\m[i: t_i | i \in I] \m. \m{t}_1)}{
            \WP{(\m[i: t_i' | i \in I] \m. \m{t}_2)}{Q}
         }
    }{
      \V \Gamma
      |- \WP{\m[i: t_i | i \in I]}{
            \WP{\m[i: t_i' | i \in I]}{
              \WP{(\m{t}_1 \m. \m{t}_2)}{Q}
            }
         }
    }}{
      \V \Gamma
      |- \WP{\m[i: (t_i\code{;}\ t'_i) | i \in I]}{\WP{(\m{t}_1 \m. \m{t}_2)}{Q}}
    }}{
      \V \Gamma
      |- \WP{(\m[i: (t_i\code{;}\ t'_i) | i \in I] \m. \m{t}_1 \m. \m{t}_2)}{Q}
    }
  \end{derivation}
\end{proof}

\begin{lemma}[Soundness of~\ref{rule:wp-impl-l}]
\label{lm:rule:wp-impl-l}
  Assume $\pvar(P) \inters \mods(\m{t}) = \emptyset$.
  Then the following entailment is admissible:
  $
    \proj(\m{t}),
    \WP{\m{t}}{\ret. Q(\ret) \implies P}
    \proves
    \WP{\m{t}}{Q} \implies P.
  $
\end{lemma}
\begin{proof}
  By the following derivation:
  \begin{derivation}
    \infer*{
    \infer*[Right=\ref{rule:wp-conj}]{
    \infer*[Right=\ref{rule:wp-cons}]{
    \infer*[Right=\ref{rule:wp-elim}]{
      \pvar(P) \inters \mods(\m{t}) = \emptyset
    }{
      \proj(\m{t}),
      \WP{\m{t}}{P}
      \proves
      P
    }}{
      \proj(\m{t}),
      \WP{\m{t}}{\ret.
        Q(\ret) \land
        (Q(\ret) \implies P)
      }
      \proves
      P
    }}{
      \proj(\m{t}),
      \WP{\m{t}}{\ret. Q(\ret) \implies P},
      \WP{\m{t}}{Q}
      \proves
      P
    }}{
      \proj(\m{t}),
      \WP{\m{t}}{\ret. Q(\ret) \implies P}
      \proves
      \WP{\m{t}}{Q} \implies P
    }
  \end{derivation}
\end{proof}

\begin{lemma}[Soundness of~\ref{rule:wp-indirect}]
\label{lm:rule:wp-indirect}
  Assume $j \notin \idx(\Gamma)$ and $j \notin \supp(\m{t})$.
  Then the following rule is admissible:
  \begin{proofrule}
  \infer{
    \V \Gamma |- \WP{\m[i: t_1]}{\E v. \at{A(v)}{i}}
    \\
    \forall v\st\ \V
    \Gamma, \at{A(v)}{j}
    |-
    \WP{(\m[i: t_1, j: t_2] \m. \m{t})}{
      \at{A(v)}{i} \implies Q(v)
    }
  }{
    \V \Gamma |-
    \WP{(\m[i: t_1] \m. \m{t})}*{
      \ret.
        \E v.
        \at{A(v)}{i}
        \land
        \WP{\m[i: t_2]}{
          \bigl( \PP i. Q(v) \bigr)\isub{j->i}
        }
    }
  }
  \end{proofrule}
\end{lemma}
\begin{proof}
  By the following derivation:
  \begin{derivation}
  \infer*[Right=\labelstep{wp-indirect:step4}]{
  \infer*[Right=\labelstep{wp-indirect:step3}]{
  \infer*[Right=\labelstep{wp-indirect:step2}]{
  \infer*[Right=\labelstep{wp-indirect:step1}]{
    \infer*{}{
      \V \Gamma |- \WP{\m[i: t_1]}{\E v. \at{A(v)}{i}}
    }
    \and
    \infer*[right=\labelstep{wp-indirect:prep}]{
      \forall v\st\ \V \Gamma, \at{A(v)}{j} |-
      \WP{(\m[i: t_1, j: t_2] \m. \m{t})}{
        \at{A(v)}{i} \implies Q(v)
      }
    }{
      \V \Gamma |-
      \WP{(\m[i: t_1]\m.\m{t})}{
      \A v.
        (\at{A(v)}{i} \land \at{A(v)}{j}) \implies
          \WP{\m[j: t_2]}{Q(v)}
      }
    }
  }{
    \V \Gamma |-
    \WP{(\m[i: t_1]\m.\m{t})}*{
      \E v. \bigl(
        \at{A(v)}{i}
        \land
        (\at{A(v)}{j} \implies \WP{\m[j: t_2]}{Q(v)})
      \bigr)
    }
  }}{
    \V \Gamma |-
    \WP{(\m[i: t_1]\m.\m{t})}*{
      \E v. \bigl(
        \at{A(v)}{i}
        \land
        (\at{A(v)}{j} \implies \WP{\m[j: t_2]}{\PP i. Q(v)})
      \bigr)
    }
  }}{
    \V \Gamma |-
    \WP{(\m[i: t_1]\m.\m{t})}*{
      \E v. \left(
        \at{A(v)}{i}
        \land
        \left(\WP{\m[j: t_2]}*{\bigl( \PP i. Q(v) \bigr)}\right)\isub{j->i}
      \right)
    }
  }}{
    \V \Gamma |-
    \WP{(\m[i: t_1] \m. \m{t})}*{
        \E v.
        \at{A(v)}{i}
        \land
        \WP{\m[i: t_2]}{
          \bigl( \PP i. Q(v) \bigr)\isub{j->i}
        }
    }
  }
  \end{derivation}
  Step~\eqref{wp-indirect:prep} uses
    \ref{rule:wp-nest},
    \ref{rule:wp-impl-r}, and
    \ref{rule:wp-all},
  to push the relevant pieces inside the postcondition of the outer WP.
  Step~\eqref{wp-indirect:step1} uses \cref{rule:wp-conj} and consequence.
  Step~\eqref{wp-indirect:step2} uses consequence to introduce projection
  in the inner postcondition.
  Step~\eqref{wp-indirect:step3} is an application of \cref{rule:wp-idx-post};
  note that $\at{A(v)}{i}$ is unaffected by the reindexing.
  Step~\eqref{wp-indirect:step4} is an application of \cref{rule:wp-idx-swap},
  justified because $i \notin \idx(\PP i. Q(v))$.
\end{proof}

\subsection{Relative Completeness}
\label{sec:completeness}

We prove that, given an oracle that provides a Hoare logic derivation
for every valid Hoare triple, we can construct a \thelogic{} derivation
for every valid hyper-triple.

\begin{theorem}[Relative Completeness]
\label{th:rel-completeness}
  Every valid hyper-triple has a \thelogic{} derivation.
\end{theorem}

\begin{proof}
  We start with some preliminary definitions:
  \begin{align*}
    \var{ST}(s_0) &\is
      \fun s.
        \A \p{x}.s(\p{x}) = s_0(\p{x})
    \\
    \var{SP}_{t}(s_0) &\is
      \fun v.
      \fun s.
        \bigstep{t}{s_0}{v}{s}
  \end{align*}
  Consider an arbitrary valid hyper-triple
  $ P \proves \WP{\m{t}}{Q} $.
  Wlog, assume~$ \idx(P) \subs \supp(t) = I $
  --- it is easy to derive the general case by using
  \ref{rule:wp-const}.
  By relative completeness of Hoare logic,
  we can obtain derivations of the valid triples
  $ \T {\var{ST}(s_0)}{\m{t}(i)}{\var{SP}_{\m{t}(i)}(s_0)} $;
  the triples are valid by definition, for any~$s_0$.
  Hoare triples are equivalent to hyper-triples of arity~1,
  so we obtain \thelogic{} derivations for
  $
    \at{\var{ST}(s_0)}{i}
    \proves
    \WP{\m[i: \m{t}(i)]}{\ret.
      \at{(\var{SP}_{t}(s_0)(\ret[i]))}{i}
    }
  $,
  for any $i \in I$.
  We can then build the derivation:
  \[
    \infer*[Right=\labelstep{step:complete-1}]{
    \infer*[Right=\labelstep{step:complete-2}]{
    \infer*[Right=\labelstep{step:complete-3}]{
    \infer*[Right=\labelstep{step:complete-4}]{
    \infer*[Right=\labelstep{step:complete-5}]{
    \infer*[Right=\labelstep{step:complete-6}]{
      \forall i \in I\st
      \forall \m{s_0}\st\;
      \at{\var{ST}(\m{s_0}(i))}{i}
      \proves
      \WP{\m[i: \m{t}(i)]}{\ret.
        \at{(\var{SP}_{\m{t}(i)}(\m{s_0}(i))(\ret[i]))}{i}
      }
    }{
      \forall \m{s_0}\st\;
      \LAnd_{i\in I}\at{\var{ST}(\m{s_0}(i))}{i}
      \proves
      \WP{\m{t}}{\ret.
        \LAnd_{i\in I}
          \at{(\var{SP}_{\m{t}(i)}(\m{s_0}(i))(\ret[i]))}{i}
      }
    }}{
      \forall \m{s_0}\st\;
      P \land \LAnd_{i\in I}\at{\var{ST}(\m{s_0}(i))}{i}
      \proves
      \WP{\m{t}}{\ret.
        \LAnd_{i\in I}
          \at{(\var{SP}_{\m{t}(i)}(\m{s_0}(i))(\ret[i]))}{i}
      }
    }}{
      \forall \m{s_0}\st\;
      P \land \LAnd_{i\in I}\at{\var{ST}(\m{s_0}(i))}{i}
      \land P(\m{s_0})
      \proves
      \WP{\m{t}}{\ret.
        P(\m{s_0}) \land
        \LAnd_{i\in I}
          \at{(\var{SP}_{\m{t}(i)}(\m{s_0}(i))(\ret[i]))}{i}
      }
    }}{
      \forall \m{s_0}\st\;
      P \land \LAnd_{i\in I}\at{\var{ST}(\m{s_0}(i))}{i}
      \land P(\m{s_0})
      \proves
      \WP{\m{t}}{Q}
    }}{
      P \land
      \E \m{s_0}.\bigl(
        \LAnd_{i\in I}\at{\var{ST}(\m{s_0}(i))}{i}
        \land P(\m{s_0})
      \bigr)
      \proves
      \WP{\m{t}}{Q}
    }}{
      P
      \proves
      \WP{\m{t}}{Q}
    }
  \]
Step~\eqref{step:complete-1} is an application of consequence
using the trivial
\[
  P \proves
  P \land
  \E \m{s_0}.\bigl(
    \LAnd_{i\in I}\at{\var{ST}(\m{s_0}(i))}{i}
    \land P(\m{s_0})
  \bigr).
\]
Step~\eqref{step:complete-2} eliminates the existential on $\m{s_0}$.
Step~\eqref{step:complete-3} uses the rule of consequence again,
by using
\[
  P(\m{s_0}) \land
  \LAnd_{i\in I}
    \at{(\var{SP}_{\m{t}(i)}(\m{s_0}(i))(\ret[i]))}{i}
  \proves
  Q(\m{r})
\]
which we know holds true by the validity of the hyper-triple
$ P \proves \WP{\m{t}}{Q} $.
Step~\eqref{step:complete-4} uses \ref{rule:wp-const} on the \emph{pure} assertion
$ P(\m{s_0}) $.
Note that the application of~$P$ to~$\m{s_0}$ makes it independent of the store.
Step~\eqref{step:complete-5} simply drops the assumption of~$P$ by \ref{rule:wp-cons}.
Step~\eqref{step:complete-6} is an application of \ref{rule:wp-conj}.
\end{proof}

\Cref{th:rel-completeness}
is also known to hold for CHL~\cite{SousaD16},
and we have given a direct proof that it holds for \thelogic{}
for completeness of exposition.
This kind of relative completeness, however,
is not very useful given the motivations that underly relational logics.
The reason is that the derivations produced by \cref{th:rel-completeness}
are always the least relational ones:
each component is analysed in isolation,
its most precise (functional) specification is derived,
and then the overall relational hyper-triple is obtained
through an implication between the functional specifications and the overall
post-condition.
This is essentially the same proof strategy as reducing hypersafety to
Hoare triples via self-composition.

 \section{Case studies}
\label{sec:case-studies}
\subsection{Loop hoisting example}
\label{sec:ex-loop-hoisting}

\emph{Loop hoisting} is a compiler optimization that removes redundant code from a loop,
transforming\footnote{Many variations are possible, but the substance of the needed reasoning for correctness is the same.}
 the program
\code{$(t_1$;while$\;g\;$do$\;$($t_1$; $t_2))$} into
\code{$(t_1$;while$\;g\;$do$\;t_2)$}.
Clearly the optimization is not valid for arbitrary~$t_1$ and~$t_2$.
Assume that the goal is to prove it correct when
$t_1$ is deterministic and \emph{idempotent},
i.e.~$t_1$ is equivalent to $(t_1\p;t_1)$---plus some mild restrictions on variables usage. In general~$t_1$ makes use of some variables~$\pv{x}$;
$t_2$ may depend on~$\pv{x}$ and use some additional
variables~$\pv{y}$.
We assume the guard of the loop~$g$ is deterministic and has no side-effects,
and that~$t_1$ does not modify the variables in~$\pv{y}$,
and~$t_2$ is deterministic and does not modify the ones in~$\pv{x}$.

Formally, we assume:
\begin{align}
  &\J |- { \pv{x}(\I1) = \pv{x}(\I2) }
         {\m[\I 1: t_1, \I 2: t_1]}
         { \pv{x}(\I1) = \pv{x}(\I2) }
  \tag{\textsc{Det}$_{t_1}$}
  \label{spec:t1-det}
  \\
  &\J |- { \pv{x}(\I2)=\vec{v} }
        {\m[\I 1: t_1, \I 2: t_1]}
        { \pv{x}(\I1)=\vec{v} \implies \pv{x}(\I1)=\pv{x}(\I2) }
  \tag{\textsc{Idem}$_{t_1}$}
  \label{spec:t1-idemp}
  \\
  &\J |- {\pv{x}(\I1) = \pv{x}(\I2) \\
          \pv{y}(\I1) = \pv{y}(\I2)}
         {\m<\I 1: g, \I 2: g>}
         {\ret.
         \ret[\I1] = \ret[\I2]
}
  \tag{\textsc{Det}$_g$}
  \label{spec:guard-det}
  \\
  &\J |- { \pv{y}(\I1) = \pv{y}(\I2) }
         {\m[\I 1: t_2, \I 2: t_2]}
         { \pv{y}(\I1) = \pv{y}(\I2) }
  \tag{\textsc{Det}$_{t_2}$}
  \label{spec:t2-det}
  \\
  &\mods(g) = \emptyset \land
  \pvar(t_1) \inters \pv{y} = \emptyset \land
  \mods(t_2) \inters \pv{x} = \emptyset
  \tag{\textsc{Mods}}
  \label{spec:mods}
\end{align}
In assertions, we write $\p x(i)$ to refer to the value stored in
the program variable $\p x$ in the state of the program at index~$i$.
Given a vector of pairwise distinct program variables~$
  \pv{x} = \p x_1 \dots \p x_n
$ we write $ \pv{x}(i) $ for the vector
$ \p x_1(i) \dots \p x_n(i) $.
The determinism judgment~\eqref{spec:t1-det} states that,
if the same input is given to two runs of~$t_1$,
then the outputs of the two runs would be the same.
We discuss the encoding of idempotence
thoroughly in~\cref{sec:discuss-idempotence}.

For deterministic programs\rlap{,}\footnote{We elaborate more on the assumption of determinism in \cref{sec:logic}.}
the equivalence between the optimization's source and target programs
can be encoded as a hyper-triple:
\begin{equation}
  \J |-
    {
      \pv{x}(\I1)=\pv{x}(\I2) \\
      \pv{y}(\I1)=\pv{y}(\I2)
    }{\m<
      \I 1: \code{$t_1$;while$\;g\;$do$\;$($t_1$; $t_2$)},
      \I 2: \code{$t_1$;while$\;g\;$do$\;t_2$}
    >}{
      \pv{x}(\I1)=\pv{x}(\I2) \\
      \pv{y}(\I1)=\pv{y}(\I2)
    }
  \tag{\textsc{Hoist}}
  \label{spec:opt-correct}
\end{equation}

The goal is to modularly use the assumed hyper-triples,
to prove the equivalence hyper-triple.
Unfortunately, the relational logics presented in prior work
all lack the reasoning principles
needed to validate the optimization at this level of generality.

\paragraph{What is missing in previous work}
To isolate the issue, let us first consider a simpler instance
of our example, that the other relational logics can handle:
instead of correctness of the optimization, we attempt translation validation.
That is, we verify that the optimization applied to a specific source program
produces a target that is equivalent to it---a classic application domain for relational logics~\cite{Benton04,BartheCK11}.

Concretely, let
$t_1 \is (\code{a=b*b})$,
$t_2 \is (\code{c=c-a})$, and
$g \is (\code{c$\;$>$\;$a+b})$.
The goal is to prove the equivalence of these two programs:
\begin{center}\begin{tabular}{c@{\hspace{10ex}}c}
{\begin{sourcecode}
src $\is$
  a = b*b;
  while (c > a+b) do
    a = b*b;
    c = c-a
\end{sourcecode}}
&
{\begin{sourcecode}
tgt $\is$
  a = b*b;
  while (c > a+b) do
    c = c-a
$ $
\end{sourcecode}}
\end{tabular}
\end{center}
An intuitive argument for the equivalence of \p{src} and \p{tgt}
is based on the observation that after line~2
the values~$v_{\p{a}}$ and~$v_{\p{b}}$ of~\p{a} and~\p{b}, respectively,
are stable throughout both loops.
More formally, we would proceed with a ``lockstep'' proof:
first, the fact that the input states of the two programs coincide
is used to conclude that the states after line~2 still coincide.
This can be proven using determinism of the assignment and
the rule~\ref{rule:bin-lockstep-seq}.
Then we would observe that, on identical states,
the guards of the two loops would evaluate either both to true or both to false,
so we can proceed by considering the two loops in lockstep.
We would then want to show that
$\p{a}(\I 1)=\p{a}(\I 2)=v_{\p{a}} \land \p{b}(\I 1)=\p{b}(\I 2)=v_{\p{b}}$
is a relational invariant of the loop.
This, however, cannot be proven unless we know that
at the beginning of both loops, \p{a} already stores \p{b*b},
i.e.~$v_{\p{a}} = (v_{\p{b}})^2$.

This intuitive proof strategy, which is the one supported by relational logics in the literature, relies on a \emph{functional specification} for the assignment to \p{a} in order to derive the stronger assertion
$v_{\p{a}} = (v_{\p{b}})^2$
in the precondition of the loops.

Unlike the simple case of \code{a=b*b},
a functional specification for the generic~$t_1$ may be neither available
nor easy to obtain, especially in the context of compiler optimizations.
Nor, intuitively, should such a functional specification be needed in order to validate the loop hoisting optimization:
the weaker, abstract properties of determinism and idempotence would suffice.

\paragraph{A proof in \thelogic}
Thanks to its new reasoning principles, \thelogic\ is capable
of expressing a fully modular proof for proving correctness
of the general optimization.

The first step is to apply \ref{rule:wp-seq},
reducing the goal to proving:
\begin{align}
  &\J |-
    {\pv{x}(\I1) = \pv{x}(\I2) \\
     \pv{y}(\I1) = \pv{y}(\I2) }
    {\m<
      \I 1: t_1,
      \I 2: t_1
    >}
    {P_0}
  &
  &\J |-
    {P_0}
    {\m<
      \I 1: \code{while$\;g\;$do$\;$($t_1$; $t_2$)},
      \I 2: \code{while$\;g\;$do$\;t_2$}
    >}
    {\pv{x}(\I1) = \pv{x}(\I2) \\
     \pv{y}(\I1) = \pv{y}(\I2) }
  \label{spec:opt-correct2}
\end{align}
for a suitable assertion $P_0$ capturing the intermediate hyper-state
that is both precise and does not require to know the full functionality of~$t_1$.
The following assertion fulfils these requirements:
\[
  P_0 \is
    \E \vec{v}. \bigl(
      \pv{x}(\I1) = \pv{x}(\I2) = \vec{v}
      \land
      \WP{\m[\I1: t_1]}{\pv{x}(\I1) = \vec{v}}
    \bigr)
    \land
    \pv{y}(\I1) = \pv{y}(\I2)
\]
$P_0$ states that, in addition to the states at $\I1$ and $\I2$ coinciding
on $\pv{x}$ and $\pv{y}$,
running $t_1$ again in component $\I1$ would preserve the value of $\pv{x}$.
This summarises the property of idempotence quite nicely.
Notice that $P_0$ implies the postcondition of \eqref{spec:opt-correct},
so we can use it as a loop invariant for the proof of the loops.

We now prove the two hyper-triples \eqref{spec:opt-correct2}.
For the first triple:
\[
\footnotesize
  \infer*[right=\labelstep{step:conj-idemp-det}]{
    \infer*[right=\labelstep{step:idemp-nest}]{
      \V
      { \pv{x}(\I2)=\vec{v} }
      |-
      \WP {\m[\I 1: t_1, \I 2: t_1]}
          { \pv{x}(\I1)=\vec{v} \implies \pv{x}(\I1)=\pv{x}(\I2) }
    }{
      \V |-
      \WP{\m[\I1:t_1]}*{
        \E \vec{v}.
          \begin{pmatrix}
            \pv{x}(\I1) = \vec{v} \\
            \WP{\m[\I1: t_1]}{\pv{x}(\I1) = \vec{v}}
          \end{pmatrix}
      }
    }
    \and
    \infer*[right=\labelstep{step:det-y}]{
      \V
      \pv{x}(\I1) = \pv{x}(\I2)
      |-
      \WP{\m[\I 1: t_1, \I 2: t_1]}
         { \pv{x}(\I1) = \pv{x}(\I2) }
    }{
      \V
      {\begin{pmatrix}
        \pv{x}(\I1) = \pv{x}(\I2) \\
        \pv{y}(\I1) = \pv{y}(\I2)
      \end{pmatrix}}
      |-
      \WP{\m<\I 1: t_1, \I 2: t_1>}*
         {\begin{pmatrix}
           \pv{x}(\I1) = \pv{x}(\I2) \\
           \pv{y}(\I1) = \pv{y}(\I2)
         \end{pmatrix}}
    }
  }{
    \V
    {\begin{pmatrix}
      \pv{x}(\I1) = \pv{x}(\I2) \\
      \pv{y}(\I1) = \pv{y}(\I2)
    \end{pmatrix}}
    |-
    \WP{\m<
      \I 1: t_1,
      \I 2: t_1
    >}*{
      \E \vec{v}.
      \begin{pmatrix}
        \pv{x}(\I1) = \pv{x}(\I2) = \vec{v} \\
        \WP{\m[\I1: t_1]}{\pv{x}(\I1) = \vec{v}}
      \end{pmatrix}
      \land
      \pv{y}(\I1) = \pv{y}(\I2)
    }
  }
\]

Step~\eqref{step:idemp-nest} transforms the idempotence triple
(in weakest precondition form) into a nested weakest precondition form.
The derivation of this step was presented in \cref{sec:overview:reindex},
and is an instance of \cref{rule:wp-indirect}.

Step \eqref{step:det-y} is an application of \ref{rule:wp-const}
to the determinism triple, using the hypothesis that
$t_1$ does not modify~$\pv{y}$.
Step \eqref{step:conj-idemp-det} is an application of \ref{rule:wp-conj}.

The proof of the loops proceeds by an application of \ref{rule:wp-while} and \ref{rule:wp-cons}.
We can apply \ref{rule:wp-const} to the triple \eqref{spec:guard-det} and obtain
that the evaluation of the guards preserves~$P_0$
(thanks to~$\mods(g) = \emptyset$).
For the loop bodies we have to prove:
\[
  \V P_0 |-
  \WP{\m[\I1: t_1\p; t_2, \I2: t_2]}{P_0}
\]
We first massage the goal to separate the task into proving
$t_1$ preserves~$P_0$, and
that the two instances of~$t_2$ together preserve~$P_0$:
\[
  \infer{\infer{\infer{\infer{
    \V P_0 |- \WP{\m[\I1: t_1]}{P_0}
    \\
    \V P_0 |- \WP{\m[\I1: t_2,\I2: t_2]}{P_0}
  }{
    \V P_0 |- \WP{\m[\I1: t_1]}[\big]{\WP{\m[\I1: t_2,\I2: t_2]}{P_0}}
  }
  }{
    \V P_0 |- \WP{\m[\I1: t_1]}[\big]{\WP{\m[\I1: t_2]}{\WP{\m[\I2: t_2]}{P_0}}}
  }}{
    \V P_0 |- \WP{\m[\I1: t_1\p; t_2]}[\big]{\WP{\m[\I2: t_2]}{P_0}}
  }}{
    \V P_0 |- \WP{\m[\I1: t_1\p; t_2, \I2: t_2]}{P_0}
  }
\]
The judgment about~$t_2$ is a simple consequence of \eqref{spec:t2-det}
and \ref{rule:wp-const}.
The side condition of \ref{rule:wp-const} is satisfied
thanks to the assumption that
$
  \pvar(t_1) \inters \pv{y} = \emptyset
$ and observing that $
  \pvar(\WP{\m[\I1: t_1]}{\pv{x}(\I1) = \vec{v}})
  =
  \pvar(\m[\I1: t_1]) \union \pv{x}(\I1)
$.

The judgment about~$t_1$ is again a consequence of idempotence:
we need to establish not only that~$\pv{x}(\I1)$ will still be~$\vec{v}$,
which is immediate from the assumptions in~$P_0$, but that running~$t_1$ after
that will still yield the same result.
\[
  \infer*{
  \infer*{
  \infer*[Right={\labelstep{step:cons-xv}}]{
  \infer*[Right={\labelstep{step:conj-wp1}}]{
  \infer*{
  \infer*{\infer*[Right={\eqref{step:idemp-nest}}]{
    \V
    { \pv{x}(\I2)=\vec{v}' }
    |-
    \WP {\m[\I 1: t_1, \I 2: t_1]}
        { \pv{x}(\I1)=\vec{v}' \implies \pv{x}(\I1)=\pv{x}(\I2) }
  }{
    \V |-
    \WP{\m[\I1:t_1]}*{
      \E \vec{v}'.
        \bigl(
          \pv{x}(\I1) = \vec{v}' \land
          \WP{\m[\I1: t_1]}{\pv{x}(\I1) = \vec{v}'}
        \bigr)
    }
  }}{
      \V
      \WP{\m[\I1: t_1]}{\pv{x}(\I1) = \vec{v}}
      |-
      \WP{\m[\I1:t_1]}*{
        \E \vec{v}'.
          \bigl(
            \pv{x}(\I1) = \vec{v}' \land
            \WP{\m[\I1: t_1]}{\pv{x}(\I1) = \vec{v}'}
          \bigr)
      }
  }}{
    \V
    \WP{\m[\I1: t_1]}{\pv{x}(\I1) = \vec{v}}
    |-
    {\begin{pmatrix*}[l]
      \WP{\m[\I1: t_1]}{\pv{x}(\I1) = \vec{v}}
      \\
      \WP{\m[\I1:t_1]}*{
        \E \vec{v}'.
          \bigl(
            \pv{x}(\I1) = \vec{v}' \land
            \WP{\m[\I1: t_1]}{\pv{x}(\I1) = \vec{v}'}
          \bigr)
      }
    \end{pmatrix*}}
  }}{
    \V
    \WP{\m[\I1: t_1]}{\pv{x}(\I1) = \vec{v}}
    |-
    \WP{\m[\I1:t_1]}*{
      \E \vec{v}'.
        \bigl(
          \pv{x}(\I1) = \vec{v}' \land
          \WP{\m[\I1: t_1]}{\pv{x}(\I1) = \vec{v}'}
        \bigr)
        \land
        \pv{x}(\I1) = \vec{v}
    }
  }}{
    \V
    \WP{\m[\I1: t_1]}{\pv{x}(\I1) = \vec{v}}
    |-
    \WP{\m[\I1:t_1]}*{
      \pv{x}(\I1) = \vec{v} \land
      \WP{\m[\I1: t_1]}{\pv{x}(\I1) = \vec{v}}
    }
  }}{
    \V
    \pv{x}(\I1) = \vec{v},
    \WP{\m[\I1: t_1]}{\pv{x}(\I1) = \vec{v}}
    |-
    \WP{\m[\I1:t_1]}*{
      \pv{x}(\I1) = \vec{v} \land
      \WP{\m[\I1: t_1]}{\pv{x}(\I1) = \vec{v}}
    }
  }}{
    \V
    {\begin{pmatrix}
    \pv{x}(\I1) = \vec{v}                    \land
    \WP{\m[\I1: t_1]}{\pv{x}(\I1) = \vec{v}} \\
    \pv{x}(\I2) = \vec{v}                    \land
    \pv{y}(\I1) = \pv{y}(\I2)
    \end{pmatrix}}
    |-
    \WP{\m[\I1:t_1]}*{
    {\begin{pmatrix}
      \pv{x}(\I1) = \vec{v} \land
      \WP{\m[\I1: t_1]}{\pv{x}(\I1) = \vec{v}} \\
      \pv{x}(\I2) = \vec{v} \land
      \pv{y}(\I1) = \pv{y}(\I2)
    \end{pmatrix}}
    }
  }
\]
Step \eqref{step:cons-xv} is justified by \ref{rule:wp-cons}
using the fact that
$
  \V
  \pv{x}(\I1)=\vec{v},
  \pv{x}(\I1)=\vec{v}'
  |-
  \vec{v}=\vec{v}'
$.
Step \eqref{step:conj-wp1} is an application of \ref{rule:wp-conj}.
 \subsection{Aligned loop distributivity}
\label{sec:ex-distrib-aligned}

Here we carry out the proof sketched in \cref{sec:discuss-modularity}.
The challenge is performing a lockstep alignment proof of
\p{f} in \cref{fig:f-code},
by not relying on a specific implementation of \p{op},
but only on the following abstract specifications:
\begin{align}
  & \p{i},\p{r} \not\in \mods(\p{op}(x_1,x_2))
  \tag{\textsc{Mods}$_{\p{op}}$}
  \label{ass:op-mods}
  \\
\proves&
  \proj(\p{op}(x_1,x_2))
  \tag{\textsc{Proj}$_{\p{op}}$}
  \label{ass:op-proj}
  \\
  \proves&
  \WP{\m[
    \I 1: \code{op}(x_1,x_2),
    \I 2: \code{op}(x_1,x_2)
  ]}{
  \ret.
    \ret[\I 1]=\ret[\I 2]
  }
  \tag{\textsc{Det}$_{\p{op}}$}
  \label{ass:op-det}
  \\
  \proves&
  \WP{\m[\I 1: \code{op}(0,0)]}{\ret. \ret[\I 1]=0}
  \tag{\textsc{Zero}$_{\p{op}}$}
  \label{ass:op-zero}
  \\
  \proves&
  \WP{\m[
    \I 1: \code{op}(x_1,x_2),
    \I 2: \code{op}(x_2,x_1)
  ]}{
  \ret.
    \ret[\I 1]=\ret[\I 2]
  }
  \tag{\textsc{Comm}$_{\p{op}}$}
  \label{ass:op-comm}
  \\
  \proves&
  \WP{\m<
    \I 1: \code{op}(x_1\p{,}x_2),*
    \I 3: \code{op}(d_1\p{,}x_3),
    \I 2: \code{op}(x_2\p{,}x_3),*
    \I 4: \code{op}(x_1\p{,}d_2)
  >}*{
  \ret.
    \begin{conj}
      {\ret[\I 1]=d_1}
      \and
      {\ret[\I 2]=d_2}
    \end{conj}
    \implies
      \ret[\I 3]=\ret[\I 4]
  }
  \tag{\textsc{Assoc}$_{\p{op}}$}
  \label{ass:op-assoc}
\end{align}
These assumptions formalize to the ones stated in \cref{sec:discuss-modularity}:
\begin{enumerate}[label=(\roman*)]
  \item not modifying variables of~\p{f},
    \eqref{ass:op-mods}
  \item a total function,
    that is,
    projectable \eqref{ass:op-proj} and
    deterministic \eqref{ass:op-det};\\
    with $\p{op}(0,0)=0$ \eqref{ass:op-zero},
  \item associative \eqref{ass:op-assoc}, and
  \item commutative \eqref{ass:op-comm}.
\end{enumerate}

We want to prove distributivity of \p{f} in its second argument:
\begin{equation*}
  \begin{judgement}[goal:fop-aligned]
    \tagHyp{$D$}
    \WP{\m[\I 4: \p{op}(b\code{,} c)]}{\ret.\ret[\I 4] = d}
    \\
    \tagHyp{$P_0$}
    \p{r}(\I1) = \p{r}(\I2) = \p{r}(\I3) = 0
    \land
    \p{i}(\I1) = \p{i}(\I2) = \p{i}(\I3)
  \proves
    \tagGoal{$G$}
    \label{fop-aligned:goal-post}
    \WPv
    {\m[ \I 1: \p{f}(a\code{,} d)
       , \I 2: \p{f}(a\code{,} b)
       , \I 3: \p{f}(a\code{,} c)
       ]}
    {
    \E v_1,v_2,v_3.
    \left(
      \begin{array}{l}
        \p{r}(\I1) = v_1 \land
        \p{r}(\I2) = v_2 \land
        \p{r}(\I3) = v_3
        \\
        \WP{\m[\I 5: \p{op}(v_2,v_3)]}{
          \ret. \ret[\I 5] = v_1
        }
      \end{array}
    \right)
    }
  \end{judgement}
\end{equation*}

We want to proceed with a proof where the loops at $\I1,\I2,\I3$ are
proved in lockstep.
To do so we will use the following loop invariant:
\[
  P =
  \begin{conj}
    \p{i}(\I1) = \p{i}(\I2) = \p{i}(\I3)
    \and
    \WP{\m[\I 4: \p{op}(b\code{,} c)]}{\ret.\ret[\I 4] = d}
    \and
    \E v_1,v_2,v_3.
\begin{conj}[c]
    \p{r}(\I1) = v_1 \land
    \p{r}(\I2) = v_2 \land
    \p{r}(\I3) = v_3
    \and
      \WP{\m[\I 5: \p{op}(v_2,v_3)]}{
        \ret. \ret[\I 5] = v_1
      }
    \end{conj}
  \end{conj}
\]

Establishing the invariant initially is a matter of applying consequence and using the assumption that
$ \p{op}(0,0) = 0 $:
\begin{derivation}
  \infer*{
    \infer*{
      \V |- \WP{\m[i: \p{op}(0,0)]}{\ret.\ret[i]=0}
    }{
      D, P_0 \proves P
    }
    \\
    \infer*{
      P
      \proves
      \WP{\m[ \I 1: \p{f}(a\code{,} d)
            , \I 2: \p{f}(a\code{,} b)
            , \I 3: \p{f}(a\code{,} c)
            ]}{P}
      \\
      P \proves Q
    }{
      P
      \proves
      \WP{\m[ \I 1: \p{f}(a\code{,} d)
            , \I 2: \p{f}(a\code{,} b)
            , \I 3: \p{f}(a\code{,} c)
            ]}{Q}
    }
  }{
    D, P_0
    \proves
    \WP{\m[ \I 1: \p{f}(a\code{,} d)
          , \I 2: \p{f}(a\code{,} b)
          , \I 3: \p{f}(a\code{,} c)
          ]}{Q}
  }
\end{derivation}
where~$Q$ is the postcondition of the goal \ref{fop-aligned:goal-post}.

An application of \ref{rule:wp-while} and \ref{rule:wp-cons} gives us:
\begin{derivation}[\footnotesize]
  \infer*{
    \infer*{}{
    P
    \proves
    \WP{\m<
         \I 1: \p{i}<a
       , \I 2: \p{i}<a
       , \I 3: \p{i}<a
       >}*{
         \begin{array}{@{}r@{}l@{}}
         \fun\m{b}.&
           P \\
           {} \land {}&
           \begin{conj}
           \phantom{{}\lor{}}\m{b}=_{\I1\I2\I3}0\\
           {}\lor \m{b}\ne_{\I1\I2\I3}0
           \end{conj}
         \end{array}
       }
    }
    \\
    \infer*{
      P
      \proves
      \WP{\m< \I 1: \code{r:=op(r,$d$)}
            , \I 2: \code{r:=op(r,$b$)}
            , \I 3: \code{r:=op(r,$c$)}
            >}[\big]{P}
      \\
      P
      \proves
      \WP{\m< \I 1: \code{i:=i+1}
            , \I 2: \code{i:=i+1}
            , \I 3: \code{i:=i+1}
            >}[\big]{P}
    }{
      P
      \proves
      \WP{\m< \I 1: \code{r:=op(r,$d$);i:=i+1}
            , \I 2: \code{r:=op(r,$b$);i:=i+1}
            , \I 3: \code{r:=op(r,$c$);i:=i+1}
            >}[\big]{P}
    }
  }{
    P
    \proves
    \WP{\m[ \I 1: \p{f}(a\code{,} d)
       , \I 2: \p{f}(a\code{,} b)
       , \I 3: \p{f}(a\code{,} c)
       ]}{P}
  }
\end{derivation}
The premise about the guards follows directly from \ref{rule:wp-const}.

For the second premise, what we need to do is to massage the judgment
until it becomes derivable from
the commutativity and associativity axioms on \p{op}.
\begin{derivation}[\footnotesize]
  \infer*{
  \infer*[Right=\labelstep{fop:reindex}]{
  \infer*[Right=\labelstep{fop:nest1}]{
  \infer*[Right=\labelstep{fop:nest2}]{
\proves
    \WP{\m< \I 5: \code{op($v_2$,$v_3$)}
          ,*\I 2: \code{op($v_2$,$b$)}
          , \I 4: \code{op($b$, $c$)}
          ,*\I 3: \code{op($v_3$,$c$)}
          >}*{
        \fun\m{v}.
\WP{\m< \I 1: \code{op($\m{v}({\I5})$,$\m{v}({\I4})$)}
          , \I 6: \code{op($\m{v}({\I2})$,$\m{v}({\I3})$)}
          >}*{
        \fun\m{r}'. \m{r}'(\I1) = \m{r}'(\I{6})
      }
    }
  }{
    {\begin{matrix*}[l]
    \p{r}(\I1) = v_1 \\
    \p{r}(\I2) = v_2 \\
    \p{r}(\I3) = v_3
    \end{matrix*}}
\proves
    \WP{\m<\I 5: \p{op}(v_2\code{,}v_3)
          ,\I 4: \p{op}(b\code{,} c)>}*{
        \fun\m{v}.
\WP{\m< \I 1: \code{op($\m{v}({\I{5}})$,$\m{v}({\I4})$)}
          , \I 2: \code{op(r,$b$)}
          , \I 3: \code{op(r,$c$)}
          >}*{
      {\begin{matrix*}[l]
        \fun\m{r}'.\E v_1',v_2',v_3'.
        \\\quad
        \m{r}'(\I1) = v_1' \land
        \m{r}'(\I2) = v_2' \land
        \m{r}'(\I3) = v_3' \land {}
        \\\quad
        \WP{\m[\I 6: \p{op}(v_2',v_3')]}{
          \ret. \ret[\I{6}] = v_1'
        }
      \end{matrix*}}
    }
    }
  }}{
    {\begin{matrix*}[l]
\p{r}(\I1) = v_1,
    \p{r}(\I2) = v_2,
    \p{r}(\I3) = v_3
    \\
    \WP{\m[\I 4: \p{op}(b\code{,} c)]}{\ret.\ret[\I 4] = d}
    \\
\WP{\m[\I 5: \p{op}(v_2,v_3)]}{
      \ret. \ret[\I 5] = v_1
    }
    \end{matrix*}}
    \proves
    \WP{\m< \I 1: \code{r:=op(r,$d$)}
          , \I 2: \code{r:=op(r,$b$)}
          , \I 3: \code{r:=op(r,$c$)}
          >}*{
      {\begin{matrix*}[l]
        \E v_1',v_2',v_3'.
        \\\quad
        \p{r}(\I1) = v_1' \land
        \p{r}(\I2) = v_2' \land
        \p{r}(\I3) = v_3' \land {}
        \\\quad
        \WP{\m[\I 6: \p{op}(v_2',v_3')]}{
          \ret. \ret[\I{6}] = v_1'
        }
      \end{matrix*}}
    }
  }}{
    {\begin{matrix*}[l]
\p{r}(\I1) = v_1,
    \p{r}(\I2) = v_2,
    \p{r}(\I3) = v_3
    \\
    \WP{\m[\I 4: \p{op}(b\code{,} c)]}{\ret.\ret[\I 4] = d}
    \\
\WP{\m[\I 5: \p{op}(v_2,v_3)]}{
      \ret. \ret[\I 5] = v_1
    }
    \end{matrix*}}
    \proves
    \WP{\m< \I 1: \code{r:=op(r,$d$)}
          , \I 2: \code{r:=op(r,$b$)}
          , \I 3: \code{r:=op(r,$c$)}
          >}*{
      {\begin{matrix*}[l]
        \E v_1',v_2',v_3'.
        \\\quad
        \p{r}(\I1) = v_1' \land
        \p{r}(\I2) = v_2' \land
        \p{r}(\I3) = v_3' \land {}
        \\\quad
        \WP{\m[\I 5: \p{op}(v_2',v_3')]}{
          \ret. \ret[\I{5}] = v_1'
        }
      \end{matrix*}}
    }
  }}{
    P
    \proves
    \WP{\m[ \I 1: \code{r:=op(r,$d$)}
          , \I 2: \code{r:=op(r,$b$)}
          , \I 3: \code{r:=op(r,$c$)}
          ]}[\big]{P}
  }
\end{derivation}

For \eqref{fop:reindex} we use \ref{rule:wp-idx-post} to apply~$\isub{6->5}$
to the postcondition of the top-level WP, and \ref{rule:wp-idx-swap}
to apply the reindexing to the inner WP.
For \eqref{fop:nest1} we use \ref{rule:wp-assign}, \ref{rule:wp-impl-l};
then we apply \ref{rule:wp-subst}
to substitute~$\p{r}$ and $d$ with~$ \m{v}(i) $ and~$ \m{v}(4) $ respectively.
For \eqref{fop:nest2} we just rearrange
the nesting of components using \ref{rule:wp-nest}.

The judgment we arrive at is essentially asking to prove
\begin{equation}
  \p{op}(\p{op}(v_2,b),\p{op}(v_3,c))
  =
  \p{op}(\p{op}(v_2,v_3), \p{op}(b,c)).
  \label{fop-aligned:informal-subgoal}
\end{equation}
which is justified by commutativity and associativity.

As a preparatory step, we transform the
associativity judgment \eqref{ass:op-assoc}
into a form that is more amenable to be applied:
\begin{equation}
  \proves
  \WP{\m<
      \I 1: \code{op($x_1$, $x_2$)},
      \I 2: \code{op($x_2$, $x_3$)}
    >}*{
    \fun\m{w}.
      \WP{\m<
        \I 3: \code{op($\m{w}({\I1})$, $x_3$)},
        \I 4: \code{op($x_1$, $\m{w}({\I2})$)}
      >}[\Big]{
      \fun\m{v}.
        \m{v}(\I3) = \m{v}(\I4)
      }
    }
    \label{fop-aligned:op-assoc}
\end{equation}
This is derived from \eqref{ass:op-assoc} as follows:
\[\small
  \infer*[Right=\labelstep{fop-align:op-assoc-rew}]{
  \infer*[Right=\ref{rule:wp-impl-r}]{
  \infer*[Right=\ref{rule:wp-nest}]{
    \proves
    \WP{\m<
      \I 1: \code{op}(x_1\p{,}x_2),*
      \I 3: \code{op}(d_1\p{,}x_3),
      \I 2: \code{op}(x_2\p{,}x_3),*
      \I 4: \code{op}(x_1\p{,}d_2)
    >}*{
    \ret.
      \begin{conj}
        {\ret[\I 1]=d_1}
        \and
        {\ret[\I 2]=d_2}
      \end{conj}
      \implies
        \ret[\I 3]=\ret[\I 4]
    }
  }{
    \proves
    \WP{\m{t}_{1,2}}*{
      \fun\m{w}.
        \WP{\m{t}_{3,4}^{d}}{
        \fun\m{v}.
          (\m{w}(\I 1)=d_1
          \land
          \m{w}(\I 2)=d_2)
          \implies
            \m{v}(\I 3)=\m{v}(\I 4)
        }
      }
  }}{
    \proves
    \WP{\m{t}_{1,2}}*{
      \fun\m{w}.
        (\m{w}(\I 1)=d_1
        \land
        \m{w}(\I 2)=d_2)
        \implies
        \WP{\m{t}_{3,4}^{d}}{
        \fun\m{v}.
            \m{v}(\I 3)=\m{v}(\I 4)
        }
      }
  }}{
    \proves
    \WP{\m{t}_{1,2}}*{
      \fun\m{w}.
        \WP{\m{t}_{3,4}^{\m{w}}}{
        \fun\m{v}.
          \m{v}(\I3) = \m{v}(\I4)
        }
      }
  }
\]
where
\begin{align*}
  \m{t}_{1,2} &\is
    \m<
      \I 1: \code{op($x_1$, $x_2$)},
      \I 2: \code{op($x_2$, $x_3$)}
    >
  &
  \m{t}_{3,4}^{d} &\is
    \m<
      \I 3: \code{op($d_1$, $x_3$)},
      \I 4: \code{op($x_1$, $d_2$)}
    >
  &
  \m{t}_{3,4}^{\m{w}} &\is
    \m<
      \I 3: \code{op($\m{w}({\I1})$, $x_3$)},
      \I 4: \code{op($x_1$, $\m{w}({\I2})$)}
    >
\end{align*}
Step~\eqref{fop-align:op-assoc-rew} picks
$\m{w}({\I1})$ and $\m{w}({\I2})$,
for the universally quantified
$d_1$ and $d_2$ respectively.

Algebraically,
we would prove \eqref{fop-aligned:informal-subgoal} using
commutativity and associativity:
\begin{align}
  &
  \p{op}(\p{op}(v_2,b),\p{op}(v_3,c))
  \tag{E1}\label{fop-aligned:e1}
  \\ = {} &
  \p{op}(\p{op}(v_2,b),\p{op}(c,v_3))
  \tag{E2}\label{fop-aligned:e2}
  \\ = {} &
  \p{op}(v_2,\p{op}(b,\p{op}(c,v_3)))
  \tag{E3}\label{fop-aligned:e3}
  \\ = {} &
  \p{op}(v_2,\p{op}(\p{op}(b,c),v_3))
  \tag{E4}\label{fop-aligned:e4}
  \\ = {} &
  \p{op}(v_2,\p{op}(v_3,\p{op}(b,c)))
  \tag{E5}\label{fop-aligned:e5}
  \\ = {} &
  \p{op}(\p{op}(v_2,v_3), \p{op}(b,c)).
  \tag{E6}\label{fop-aligned:e6}
\end{align}
We can perform the above reasoning within the logic by constructing the
following derivation:
\[
  \infer*[Right={\labelstep{fop:e1e6}}]{
    \infer*[Right={\labelstep{fop:e1e2}}]{\eqref{ass:op-comm}}{
      \ref{fop-aligned:e1=e2}
    }
    \\
    \infer*[Right={\labelstep{fop:e2e6}}]{
      \infer*[Right={\labelstep{fop:e2e3}}]{\eqref{fop-aligned:op-assoc}}{
        [\text{\ref{fop-aligned:e2}} = \text{\ref{fop-aligned:e3}}]
      }
      \\
      \infer*[Right={\labelstep{fop:e3e6}}]{
        \infer*[Right={\labelstep{fop:e3e4}}]{\eqref{fop-aligned:op-assoc}}{
          [\text{\ref{fop-aligned:e3}} = \text{\ref{fop-aligned:e4}}]
        }
        \\
        \infer*[Right={\labelstep{fop:e4e6}}]{
          \infer*[Right={\labelstep{fop:e4e5}}]{\eqref{ass:op-comm}}{
            [\text{\ref{fop-aligned:e4}} = \text{\ref{fop-aligned:e5}}]
          }
          \\
          \infer*[Right={\labelstep{fop:e5e6}}]{\eqref{fop-aligned:op-assoc}}{
            [\text{\ref{fop-aligned:e5}} = \text{\ref{fop-aligned:e6}}]
          }
        }{
          [\text{\ref{fop-aligned:e4}} = \text{\ref{fop-aligned:e6}}]
        }
      }{
        [\text{\ref{fop-aligned:e3}} = \text{\ref{fop-aligned:e6}}]
      }
    }{
      \ref{fop-aligned:e2=e6}
    }
  }{
    \ref{fop-aligned:e1=e6}
  }
\]
where we write
$ [\text{E}j = \text{E}k] $
for the judgment representation of the equality,
and the leaves are the commutativity or associativity assumptions on~\p{op}.
The derivation starts from
  $[\text{\ref{fop-aligned:e1}} = \text{\ref{fop-aligned:e6}}]$
which is the judgment we obtained from \eqref{fop:nest2}:
\begin{equation}
\small
  \proves
  \WP{\m<
    \I 5: \code{op($v_2$,$v_3$)},*
    \I 2: \code{op($v_2$,$b$)},
    \I 4: \code{op($b$, $c$)},*
    \I 3: \code{op($v_3$,$c$)}
  >}*{
  \fun\m{v}.
    \WP{\m<
      \I 1: \code{op($\m{v}({\I5})$,$\m{v}({\I4})$)},
      \I 6: \code{op($\m{v}({\I2})$,$\m{v}({\I3})$)}
    >}[\Big]{
    \ret.
      \ret[\I1] = \ret[\I6]
    }
  }
  \tag*{[\text{\ref{fop-aligned:e1}} = \text{\ref{fop-aligned:e6}}]}
  \label{fop-aligned:e1=e6}
\end{equation}

We reduce the goal to proving the two judgments:
{\small
\begin{align}
  \proves &
  \WP{\m[
    \I 3: \code{op($v_3$,$c$)},
    \I 7: \code{op($c$,$v_3$)}
  ]}{
  \fun\m{v}.
    \m{v}(\I3) = \m{v}(\I7)
  }
  \tag*{[\text{\ref{fop-aligned:e1}} = \text{\ref{fop-aligned:e2}}]}
  \label{fop-aligned:e1=e2}
  \\
  \proves &
  \WP{\m<
    \I 5: \code{op($v_2$,$v_3$)},*
    \I 2: \code{op($v_2$,$b$)},
    \I 4: \code{op($b$, $c$)},*
    \I 7: \code{op($c$,$v_3$)}
  >}*{
  \fun\m{v}.
    \WP{\m<
      \I 1: \code{op($\m{v}({\I5})$,$\m{v}({\I4})$)},
      \I 6: \code{op($\m{v}({\I2})$,$\m{v}({\I7})$)}
    >}[\Big]{
    \ret.
      \ret[\I1] = \ret[\I6]
    }
  }
  \tag*{[\text{\ref{fop-aligned:e2}} = \text{\ref{fop-aligned:e6}}]}
  \label{fop-aligned:e2=e6}
\end{align}
}
by means of the following derivation (which justifies \eqref{fop:e1e6}):
\[\small
  \infer*[Right={\ref{rule:wp-proj}$\scriptstyle(\I7)$}]{
  \infer*[Right=\ref{rule:wp-cons}]{
  \infer*[Right=\ref{rule:wp-conj}]{
    \ref{fop-aligned:e1=e2}
    \\
    \ref{fop-aligned:e2=e6}
  }{
    \proves
    \WP{\m<
      \I 5: \code{op($v_2$,$v_3$)},*
      \I 2: \code{op($v_2$,$b$)},
      \I 4: \code{op($b$, $c$)},*
      \I 3: \code{op($v_3$,$c$)},
      \I 7: \code{op($c$,$v_3$)}
    >}*{
    \fun\m{v}.
    \begin{pmatrix*}[l]
      {\m{v}(\I3) = \m{v}(\I7)}
      \\
      {\WP{\m<
        \I 1: \code{op($\m{v}({\I5})$,$\m{v}({\I4})$)},
        \I 6: \code{op($\m{v}({\I2})$,$\m{v}({\I7})$)}
      >}[\Big]{
      \ret.
        \ret[\I1] = \ret[\I6]
      }}
    \end{pmatrix*}
    }
  }}{
    \proves
    \WP{\m<
      \I 5: \code{op($v_2$,$v_3$)},*
      \I 2: \code{op($v_2$,$b$)},
      \I 4: \code{op($b$, $c$)},*
      \I 3: \code{op($v_3$,$c$)},
      \I 7: \code{op($c$,$v_3$)}
    >}*{
    \fun\m{v}.
      \WP{\m<
        \I 1: \code{op($\m{v}({\I5})$,$\m{v}({\I4})$)},
        \I 6: \code{op($\m{v}({\I2})$,$\m{v}({\I3})$)}
      >}[\Big]{
      \ret.
        \ret[\I1] = \ret[\I6]
      }
    }
  }
  }{
    \ref{fop-aligned:e1=e6}
  }
\]
The last step projects out component~$\I7$ from the triple,
which is justified by \eqref{ass:op-proj}.

While \ref{fop-aligned:e1=e2} is a direct instance of \eqref{ass:op-comm},
\ref{fop-aligned:e2=e6} is reduced to proving:
{\small
\begin{align}
  \proves&
  \WP{\m<
    \I 7: \code{op($c$,$v_3$)},
    \I 2: \code{op($v_2$,$b$)}
  >}*{
  \fun\m{v}.
    \WP{\m[\I 8: \code{op($b$, $\m{v}({\I7})$)}]}*{
    \fun\m{w}.
      \WP{\m<
        \I 9: \code{op($v_2$,$\m{w}({\I8})$)},
        \I 6: \code{op($\m{v}({\I2})$,$\m{v}({\I7})$)}
      >}[\Big]{
      \ret.
        \ret[\I9] = \ret[\I6]
      }
    }
  }
  \tag*{[\text{\ref{fop-aligned:e2}} = \text{\ref{fop-aligned:e3}}]}
  \label{fop-aligned:e2=e3}
  \\
  \proves&
  \WP{\m<
    \I 7: \code{op($c$,$v_3$)},
    \I 5: \code{op($v_2$,$v_3$)}\mkern-5mu,
    \I 4: \code{op($b$, $c$)}
  >}*{
  \fun\m{v}.
    \WP{\m[\I 8: \code{op($b$, $\m{v}({\I7})$)}]}*{
    \fun\m{w}.
      \WP{\m<
        \I 1: \code{op($\m{v}({\I5})$,$\m{v}({\I4})$)},
        \I 9: \code{op($v_2$,$\m{w}({\I8})$)}
      >}[\Big]{
      \ret.
        \ret[\I1] = \ret[\I9]
      }
    }
  }
  \tag*{[\text{\ref{fop-aligned:e3}} = \text{\ref{fop-aligned:e6}}]}
  \label{fop-aligned:e3=e6}
\end{align}
}
Again the reduction is driven by \ref{rule:wp-conj} and \ref{rule:wp-proj}:
\begin{derivation}
  \infer*[Right=\ref{rule:wp-empty}]{
  \infer*[Right={\ref{rule:wp-proj}$\scriptstyle(\I8,\I9)$}]{
  \infer*[Right=\ref{rule:wp-conj}]{
    \ref{fop-aligned:e2=e3}
    \\
    \ref{fop-aligned:e3=e6}
  }{
    \proves
    \WP{\m<
      \I 7: \code{op($c$,$v_3$)},
      \I 5: \code{op($v_2$,$v_3$)}\mkern-5mu,
      \I 4: \code{op($b$, $c$)},
      \I 2: \code{op($v_2$,$b$)}
    >}*{
    \fun\m{v}.
      \WP{\m[\I 8: \code{op($b$, $\m{v}({\I7})$)}]}*{
      \fun\m{w}.
        \WP{\m<
          \I 1: \code{op($\m{v}({\I5})$,$\m{v}({\I4})$)},
          \I 9: \code{op($v_2$,$\m{w}({\I8})$)},
          \I 6: \code{op($\m{v}({\I7})$,$\m{v}({\I2})$)}
        >}*{
        \ret.
        \begin{conj}
          {\ret[\I1] = \ret[\I9]}
          \and
          {\ret[\I9] = \ret[\I6]}
        \end{conj}
        }
      }
    }
  }}{
    \proves
    \WP{\m<
      \I 5: \code{op($v_2$,$v_3$)},*
      \I 2: \code{op($v_2$,$b$)},
      \I 4: \code{op($b$, $c$)},*
      \I 7: \code{op($c$,$v_3$)}
    >}*{
    \fun\m{v}.
      \WP{\m[]}*{
\WP{\m<
          \I 1: \code{op($\m{v}({\I5})$,$\m{v}({\I4})$)},
          \I 6: \code{op($\m{v}({\I7})$,$\m{v}({\I2})$)}
        >}*{
        \ret.
          \ret[\I1] = \ret[\I6]
        }
      }
    }
  }}{
    \ref{fop-aligned:e2=e6}
  }
\end{derivation}

The proof of \ref{fop-aligned:e2=e3}
uses associativity:
\begin{derivation}
  \infer*[Right=\labelstep{fop-aligned:assoc23}]{
  \infer*[Right=\labelstep{fop-aligned:nest23}]{
    \infer*{ }{
      \proves
      \WP{\m[\I 7: \code{op($c$,$v_3$)}]}{\True}
    }
    \and
    \infer*{\eqref{fop-aligned:op-assoc}}{
    \proves
\WP{\m<
        \I 2: \code{op($v_2$,$b$)},
        \I 8: \code{op($b$, $\m{v}({\I7})$)}
      >}*{
      \fun\m{w}.
        \WP{\m<
          \I 9: \code{op($v_2$,$\m{w}({\I8})$)},
          \I 6: \code{op($\m{w}({\I2})$,$\m{v}({\I7})$)}
        >}[\Big]{
        \ret.
          \ret[\I9] = \ret[\I6]
        }
      }
    }
  }{
    \proves
    \WP{\m[\I 7: \code{op($c$,$v_3$)}]}*{
    \A\m{v}.
      \WP{\m<
        \I 2: \code{op($v_2$,$b$)},
        \I 8: \code{op($b$, $\m{v}({\I7})$)}
      >}*{
      \fun\m{w}.
        \WP{\m<
          \I 9: \code{op($v_2$,$\m{w}({\I8})$)},
          \I 6: \code{op($\m{w}({\I2})$,$\m{v}({\I7})$)}
        >}[\Big]{
        \ret.
          \ret[\I9] = \ret[\I6]
        }
      }
    }
  }}{
    \proves
    \WP{\m<
      \I 7: \code{op($c$,$v_3$)},
      \I 2: \code{op($v_2$,$b$)}
    >}*{
    \fun\m{v}.
      \WP{\m[\I 8: \code{op($b$, $\m{v}({\I7})$)}]}*{
      \fun\m{w}.
        \WP{\m<
          \I 9: \code{op($v_2$,$\m{w}({\I8})$)},
          \I 6: \code{op($\m{v}({\I2})$,$\m{v}({\I7})$)}
        >}[\Big]{
        \ret.
          \ret[\I9] = \ret[\I6]
        }
      }
    }
  }
\end{derivation}
Step \eqref{fop-aligned:nest23} uses \ref{rule:wp-nest} and \ref{rule:wp-cons}
to rearrange the components and generalize the postcondition
to universally quantify over $\m{v}$ (which is then unrelated to the return values).

Step \eqref{fop-aligned:assoc23} uses \ref{rule:wp-const}
to get the postcondition out of the outer weakest precondition assertion,
leaving us with an application of \ref{rule:wp-triv}
and an instance of associativity in the form of \eqref{fop-aligned:op-assoc}.

The proof of \ref{fop-aligned:e2=e6} is further reduced to proving:
{\small
\begin{align}
  \proves&
  \WP{\m<
    \I 4: \code{op($b$, $c$)},
    \I 7: \code{op($c$,$v_3$)}
  >}*{
  \fun\m{v}.
    \WP{\m<
      \I10: \code{op($\m{v}({\I4})$, $v_3$)},
      \I 8: \code{op($b$, $\m{v}({\I7})$)}
    >}*{
    \fun\m{w}.
      \WP{\m<
        \I 9: \code{op($v_2$,$\m{w}({\I8})$)},
        \I 11: \code{op($v_2$,$\m{w}({\I10})$)}
      >}[\Big]{
      \ret.
        \ret[\I9] = \ret[\I{11}]
      }
    }
  }
  \tag*{[\text{\ref{fop-aligned:e3}} = \text{\ref{fop-aligned:e4}}]}
  \label{fop-aligned:e3=e4}
  \\
  \proves&
  \WP{\m<
    \I 5: \code{op($v_2$,$v_3$)}\mkern-5mu,
    \I 4: \code{op($b$, $c$)}
  >}*{
  \fun\m{v}.
    \WP{\m[
      \I 10: \code{op($\m{v}({\I4})$, $v_3$)}
    ]}*{
    \fun\m{w}.
      \WP{\m<
        \I 1: \code{op($\m{v}({\I5})$,$\m{v}({\I4})$)},
        \I 11: \code{op($v_2$,$\m{w}({\I10})$)}
      >}[\Big]{
      \ret.
        \ret[\I1] = \ret[\I{11}]
      }
    }
  }
  \tag*{[\text{\ref{fop-aligned:e4}} = \text{\ref{fop-aligned:e6}}]}
  \label{fop-aligned:e4=e6}
\end{align}
}

These can again be combined using \ref{rule:wp-conj} to obtain
\ref{fop-aligned:e3=e6} by projecting out components $\I{10}$ and $\I{11}$
using \ref{rule:wp-proj}.

The judgment \ref{fop-aligned:e3=e4} can be proved by appealing to
associativity and determinism:
\begin{derivation}
  \infer*[Right=\ref{rule:wp-cons}]{
    \proves
    \WP{\m<
      \I 4: \code{op($b$, $c$)},
      \I 7: \code{op($c$,$v_3$)}
    >}*{
    \fun\m{v}.
      \WP{\m<
        \I10: \code{op($\m{v}({\I4})$, $v_3$)},
        \I 8: \code{op($b$, $\m{v}({\I7})$)}
      >}[\Big]{
      \fun\m{w}.
        \m{w}(\I{10}) = \m{w}(\I{8})
      }
    }
    \\
    \m{w}(\I{10}) = \m{w}(\I{8})
    \proves
    \WP{\m<
      \I 9: \code{op($v_2$,$\m{w}({\I8})$)},
      \I 11: \code{op($v_2$,$\m{w}({\I10})$)}
    >}[\Big]{
    \ret.
      \ret[\I9] = \ret[\I{11}]
    }
  }{
    \ref{fop-aligned:e3=e4}
  }
\end{derivation}

Finally, we reduce the remaining \ref{fop-aligned:e4=e6}
to proving:
{\small
\begin{align}
  \proves&
  \WP{\m[
    \I 4: \code{op($b$, $c$)}
  ]}*{
  \fun\m{v}.
    \WP{\m<
      \I10: \code{op($\m{v}({\I4})$, $v_3$)},
      \I12: \code{op($v_3$, $\m{v}({\I4})$)}
    >}*{
    \fun\m{w}.
      \WP{\m<
        \I 11: \code{op($v_2$,$\m{w}({\I10})$)},
        \I 13: \code{op($v_2$,$\m{w}({\I12})$)}
      >}[\Big]{
      \ret.
        \ret[\I{13}] = \ret[\I{11}]
      }
    }
  }
  \tag*{[\text{\ref{fop-aligned:e4}} = \text{\ref{fop-aligned:e5}}]}
  \label{fop-aligned:e4=e5}
  \\
  \proves&
  \WP{\m<
    \I 5: \code{op($v_2$,$v_3$)}\mkern-5mu,
    \I 4: \code{op($b$, $c$)}
  >}*{
  \fun\m{v}.
    \WP{\m[
      \I12: \code{op($v_3$, $\m{v}({\I4})$)}
    ]}*{
    \fun\m{w}.
      \WP{\m<
        \I 1: \code{op($\m{v}({\I5})$,$\m{v}({\I4})$)},
        \I 13: \code{op($v_2$,$\m{w}({\I12})$)}
      >}[\Big]{
      \ret.
        \ret[\I1] = \ret[\I{13}]
      }
    }
  }
  \tag*{[\text{\ref{fop-aligned:e5}} = \text{\ref{fop-aligned:e6}}]}
  \label{fop-aligned:e5=e6}
\end{align}
}
These can once more be combined using \ref{rule:wp-conj} to obtain
\ref{fop-aligned:e4=e6} by projecting out components $\I{12}$ and $\I{13}$
using \ref{rule:wp-proj}.

The judgment \ref{fop-aligned:e4=e5} follows from commutativity applied to components~$\I{10}$ and~$\I{12}$:
\begin{derivation}
  \infer*{
  \infer*{
\infer*{
      \proves
      \WP{\m<
        \I10: \code{op($\m{v}({\I4})$, $v_3$)},
        \I12: \code{op($v_3$, $\m{v}({\I4})$)}
      >}[\Big]{
      \fun\m{w}.
        \m{w}(\I{10}) = \m{w}(\I{12})
      }
      \\
      \m{w}(\I{10}) = \m{w}(\I{12})
      \proves
      \WP{\m<
        \I 11: \code{op($v_2$,$\m{w}({\I10})$)},
        \I 13: \code{op($v_2$,$\m{w}({\I12})$)}
      >}[\Big]{
      \ret.
        \ret[\I{13}] = \ret[\I{11}]
      }
    }{
      \WP{\m<
        \I10: \code{op($\m{v}({\I4})$, $v_3$)},
        \I12: \code{op($v_3$, $\m{v}({\I4})$)}
      >}*{
      \fun\m{w}.
        \WP{\m<
          \I 11: \code{op($v_2$,$\m{w}({\I10})$)},
          \I 13: \code{op($v_2$,$\m{w}({\I12})$)}
        >}[\Big]{
        \ret.
          \ret[\I{13}] = \ret[\I{11}]
        }
      }
    }
  }{
    \proves
    \WP{\m[
      \I 4: \code{op($b$, $c$)}
    ]}*{
    \A\m{v}.
      \WP{\m<
        \I10: \code{op($\m{v}({\I4})$, $v_3$)},
        \I12: \code{op($v_3$, $\m{v}({\I4})$)}
      >}*{
      \fun\m{w}.
        \WP{\m<
          \I 11: \code{op($v_2$,$\m{w}({\I10})$)},
          \I 13: \code{op($v_2$,$\m{w}({\I12})$)}
        >}[\Big]{
        \ret.
          \ret[\I{13}] = \ret[\I{11}]
        }
      }
    }
  }}{
    \ref{fop-aligned:e4=e5}
  }
\end{derivation}
The leaves are instances of commutativity and determinism.

The judgment \ref{fop-aligned:e5=e6} is proven by a first application of
\ref{rule:wp-nest} and then by appealing to associativity,
as in~\eqref{fop-aligned:assoc23}:
\begin{derivation}
  \infer*{
  \infer*{
    \infer*{\eqref{fop-aligned:op-assoc}}{
      \proves
      \WP{\m<
        \I 5: \code{op($v_2$,$v_3$)}\mkern-5mu,
        \I12: \code{op($v_3$, $\m{v}({\I4})$)}
      >}*{
      \fun\m{w}.
        \WP{\m<
          \I 1: \code{op($\m{w}({\I5})$,$\m{v}({\I4})$)},
          \I 13: \code{op($v_2$,$\m{w}({\I12})$)}
        >}[\Big]{
        \ret.
          \ret[\I1] = \ret[\I{13}]
        }
      }
    }
  }{
    \proves
    \WP{\m[\I 4: \code{op($b$, $c$)}]}*{
    \A\m{v}.
      \WP{\m<
        \I 5: \code{op($v_2$,$v_3$)}\mkern-5mu,
        \I12: \code{op($v_3$, $\m{v}({\I4})$)}
      >}*{
      \fun\m{w}.
        \WP{\m<
          \I 1: \code{op($\m{w}({\I5})$,$\m{v}({\I4})$)},
          \I 13: \code{op($v_2$,$\m{w}({\I12})$)}
        >}[\Big]{
        \ret.
          \ret[\I1] = \ret[\I{13}]
        }
      }
    }
  }}{
    \ref{fop-aligned:e5=e6}
  }
\end{derivation}

 \subsection{Loop splitting}
\label{sec:ex-loop-splitting}

Here we pick up the proof sketched in \cref{sec:discuss-refinement},
where we reduce distributivity on the first argument of the program
\p{f} in \cref{fig:f-code} (with \p{op} fixed as addition),
to proving that we can split a while loop running~$a+b$ times
into a loop running~$a$ times followed by one running~$b$ times.

Formally, we need to prove the equivalence between component~\I{1} and~\I{4},
which can be encoded by the auxiliary hyper-triple~\eqref{spec:mult-split}:
\begin{equation}
  \small
  \J |- {\p{r}(\I1)=\p{r}(\I4) \land \p{i}(\I1)=\p{i}(\I4)}
        {\m[\I 1: \p{f}(a+b, c),
          \I 4: \p{f}(a, c)\p;\p{f}(a+b, c)
        ]}
        {\p{r}(\I1)=\p{r}(\I4) \land \p{i}(\I1)=\p{i}(\I4)}
  \tag{\ref{spec:mult-split}}
\end{equation}

We start the derivation by using nesting and sequence to break up the three while loops.
Letting
\[
  P = \bigl(
    \p{i}(\I1) \leq \p{i}(\I4)
    \land
    \WP
      {\m[\I 4: \p{f}(a+b, c),
        \I 1: \p{f}(a+b, c)
      ]}
      {
        \p{r}(\I1)=\p{r}(\I4) \land \p{i}(\I1)=\p{i}(\I4)
      }
  \bigr)
\]
we build the derivation:
\[
  \small
  \infer*{
    \infer*{
      \V
      \p{r}(\I1)=\p{r}(\I4),
      \p{i}(\I1)=\p{i}(\I4)
      |-
      P
      \\
      \V P |- \WP {\m[\I 4: \p{f}(a, c)]} {P}
    }{
      \infer*{
        \V
        \p{r}(\I1)=\p{r}(\I4),
        \p{i}(\I1)=\p{i}(\I4)
        |-
        \WP
          {\m[\I 4: \p{f}(a, c)
          ]}
          {
          \WP
            {\m[\I 4: \p{f}(a+b, c),
              \I 1: \p{f}(a+b, c)
            ]}
            {
              \p{r}(\I1)=\p{r}(\I4) \land \p{i}(\I1)=\p{i}(\I4)
            }
          }
      }{
        \V
        \p{r}(\I1)=\p{r}(\I4),
        \p{i}(\I1)=\p{i}(\I4)
        |-
        \WP
          {\m[\I 4: \p{f}(a, c)\p;\p{f}(a+b, c)
          ]}
          {
          \WP
            {\m[\I 1: \p{f}(a+b, c)
            ]}
            {
              \p{r}(\I1)=\p{r}(\I4) \land \p{i}(\I1)=\p{i}(\I4)
            }
          }
      }
    }
  }{
    \V
    \p{r}(\I1)=\p{r}(\I4),
    \p{i}(\I1)=\p{i}(\I4)
    |-
    \WP
      {\m[\I 1: \p{f}(a+b, c),
        \I 4: \p{f}(a, c)\p;\p{f}(a+b, c)
      ]}
      {
        \p{r}(\I1)=\p{r}(\I4) \land \p{i}(\I1)=\p{i}(\I4)
      }
  }
\]
The left-hand leaf
$\V \p{r}(\I1)=\p{r}(\I4), \p{i}(\I1)=\p{i}(\I4) |- P$
encodes determinism of \p{f}
which can be easily proven by lockstep reasoning.

The right-hand leaf is in the form required to apply
the \ref{rule:wp-while} rule, with~$P$
as loop invariant:
\[
  \infer*{
    \V P |-
      \WP {\m[\I 4: i < a]} {
        \fun\m{b}.
          (\m{b}(\I4) = 0 \land P)
          \lor
          (\m{b}(\I4) \ne 0 \land
            \WP {
              \m[\I4: (\code{r:=r+$c$;i:=i+1})]
            }{P}
          )
      }
  }{
    \V P |- \WP {\m[\I 4: \p{f}(a, c)]} {P}
  }
\]

Let $
  R(C) = (C \land P)
          \lor
          (\neg C \land
            \WP {
              \m[\I4: (\code{r:=r+$c$;i:=i+1})]
            }{P}
          )
$.
Since ${\mods(\p{i<}a)=\emptyset}$,
we can eliminate the outer weakest precondition by:

\[
  \small
  \infer*[right=\ref{rule:wp-conj}]{
    \V P |-\WP {\m[\I 4: i < a]} {\fun\m{b}.\m{b}(\I4) = 0 \iff (i(\I4)<a)}
    \and
    \infer*[Right=\ref{rule:wp-triv}]{
        \V P |- R(i(\I4)<a)
}{
      \infer*[Right=\ref{rule:wp-const}]{
        \V P |-R(i(\I4)<a) \land \WP {\m[\I 4: i < a]} {\True}
      }{
        \V P |-\WP {\m[\I 4: i < a]} {R(i(\I4)<a)}
      }
    }
  }{
    \V P |-\WP {\m[\I 4: i < a]} {\fun\m{b}.R(\m{b}(\I4) = 0)}
  }
\]

Then we apply the refine rule with
\begin{equation}
  \code{f($a$+$b$,$c$)}
  \semeq
  \code{if i<$a$+$b$ then r:=r+$c$;i:=i+1;f($a$+$b$,$c$) else skip}
  \label{loop-split:refine}
\end{equation}
and apply \ref{rule:wp-if} plus some simplification you get:
\[
  \begin{judgement}
    (\p{i}(\I4)<a+b \land P)
    \lor
    (\p{i}(\I4)\geq a+b \land
      \WP {
        \m[\I4: (\code{r:=r+$c$;i:=i+1})]
      }{P}
    )
  \proves
    (\p{i}(\I4)<a \land P)
    \lor
    (\p{i}(\I4)\geq a \land
      \WP {
        \m[\I4: (\code{r:=r+$c$;i:=i+1})]
      }{P}
    )
  \end{judgement}
\]

We prove the judgment by case analysis:
the cases where
$\p{i}(\I4) < a$ or
$\p{i}(\I4) \geq a+b$
are trivial.
We are left with the interesting case:
$
  \V
    a \leq \p{i}(\I4) < a+b \land P
  |-
    \WP {
      \m[\I4: (\code{r:=r+$c$;i:=i+1})]
    }{P}
$.
Expanding~$P$ we have to show:

\[
\begin{judgement}
    a \leq \p{i}(\I4) < a+b \land
    \p{i}(\I1) \leq \p{i}(\I4)
    \\
    \WP
      {\m[\I 4: \underline{\p{f}(a+b, c)},
        \I 1: \underline{\p{f}(a+b, c)}
      ]}
      {
        \p{r}(\I1)=\p{r}(\I4) \land \p{i}(\I1)=\p{i}(\I4)
      }
  \proves
    \WPv {
      \m[\I4: (\code{r:=r+$c$;i:=i+1})]
    }{
    \p{i}(\I1) \leq \p{i}(\I4)
    \land
    \WP
      {\m[\I 4: \p{f}(a+b, c),
        \I 1: \underline{\p{f}(a+b, c)}
      ]}
      {
        \p{r}(\I1)=\p{r}(\I4) \land \p{i}(\I1)=\p{i}(\I4)
      }
    }
\end{judgement}
\]
Conceptually,
since $\p{i}(\I1) < a+b$,
component \I1 in the conclusion will perform at least one iteration.
Similarly, both components \I1 and \I4 in the assumptions will perform at least one iteration.
By using~\eqref{loop-split:refine} on the underlined code,
followed by some trivial simplifications to eliminate the guards, we get:
\[
  \begin{judgement}
    a \leq \p{i}(\I4) < a+b \land
    \p{i}(\I1) \leq \p{i}(\I4)
    \\
    \WPv
      {\m[\I 4: (\code{r:=r+$c$;i:=i+1}),
        \I 1: (\code{r:=r+$c$;i:=i+1})
      ]}
      {\WP
        {\m[\I 4: \p{f}(a+b, c),
          \I 1: \p{f}(a+b, c)
        ]}
        {
          \p{r}(\I1)=\p{r}(\I4) \land \p{i}(\I1)=\p{i}(\I4)
        }
      }
  \proves
    \WPv {
      \m[\I4: (\code{r:=r+$c$;i:=i+1})]
    }{
      \p{i}(\I1) \leq \p{i}(\I4) \land {}
      \\&
      \WPv {
        \m[\I1: (\code{r:=r+$c$;i:=i+1})]
      }{
        \WP
        {\m[\I 4: \p{f}(a+b, c),
          \I 1: \p{f}(a+b, c)
        ]}
        {
          \p{r}(\I1)=\p{r}(\I4) \land \p{i}(\I1)=\p{i}(\I4)
        }
      }
    }
  \end{judgement}
\]

By using \ref{rule:wp-conj} with the trivial
$
  \V
  \p{i}(\I1) \leq \p{i}(\I4)
  |-
  \WP {
    \m[\I4: (\code{r:=r+$c$;i:=i+1})]
  }{
    \p{i}(\I1) \leq \p{i}(\I4)
  }
$
we can discharge the $\p{i}(\I1) \leq \p{i}(\I4)$ in the postcondition,
leaving us, through \ref{rule:wp-nest}, with the trivial judgment:
\[
  \begin{judgement}
    a \leq \p{i}(\I4) < a+b \land
    \p{i}(\I1) \leq \p{i}(\I4)
    \\
    \WPv
      {\m[\I 4: (\code{r:=r+$c$;i:=i+1}),
        \I 1: (\code{r:=r+$c$;i:=i+1})
      ]}
      {\WP
        {\m[\I 4: \p{f}(a+b, c),
          \I 1: \p{f}(a+b, c)
        ]}
        {
          \p{r}(\I1)=\p{r}(\I4) \land \p{i}(\I1)=\p{i}(\I4)
        }
      }
  \proves
    \WPv {
      \m[\I4: (\code{r:=r+$c$;i:=i+1}),
         \I1: (\code{r:=r+$c$;i:=i+1})]
      }{
        \WP
        {\m[\I 4: \p{f}(a+b, c),
          \I 1: \p{f}(a+b, c)
        ]}
        {
          \p{r}(\I1)=\p{r}(\I4) \land \p{i}(\I1)=\p{i}(\I4)
        }
      }
  \end{judgement}
\]
 \subsection{Distributivity of \p{f} on both arguments}
\label{sec:distr-both}
Here we sketch the derivation of distributivity of \p{f} on both arguments
we described in \cref{sec:discuss-divide-and-conquer}.

Take again the example in \cref{fig:f-code}
fixing $\p{op}(x,y)\is x + y$.
In \cref{sec:ex-distrib-aligned,sec:ex-loop-splitting}
we obtained a proof of distributivity on the second argument,
and one of distributivity on the first.
We want to prove distributivity over \emph{both} arguments;
as a hyper-triple:
\begin{equation}
  \small
  \J |- {\LAnd_{j=1}^5 \p{r}(j)=\p i(j)=0}
        {\m[ \I 1: \p{f}(a{+}b,c{+}d)
           , \I 2: \p{f}(a,c)
           , \I 3: \p{f}(b,c)
           , \I 4: \p{f}(a,d)
           , \I 5: \p{f}(b,d)
           ]}
        {\p{r}(\I 1) = \sum_{j=2}^5 \p{r}(\I j)}
  \label{spec:mult-full-distr}
\end{equation}
By using \ref{rule:wp-conj} and \ref{rule:wp-proj} it is possible to decompose this hyper-triple into two proofs of distributivity, respective on the first and the second arguments which we discussed in the previous sections.

By means of \ref{rule:wp-proj} we introduce the two intermediate terms
$ \p{f}(a+b,c) $ and $\p{f}(a+b,d)$,
at indices~\I{6} and~\I{7} respectively:
\begin{equation}
  \Jv |- {\LAnd_{j=1}^7 \p{r}(j)=\p i(j)=0}
        {\m[ \I 1: \p{f}(a{+}b,c{+}d)
           , \I 2: \p{f}(a,c)
           , \I 3: \p{f}(b,c)
           , \I 4: \p{f}(a,d)
           , \I 5: \p{f}(b,d)
           , \I 6: \p{f}(a{+}b,c)
           , \I 7: \p{f}(a{+}b,d)
           ]}
        {
          \p{r}(\I 1) = \p{r}(\I 6)+\p{r}(\I 7) \land
          \p{r}(\I 6) = \p{r}(\I 1)+\p{r}(\I 3) \land
          \p{r}(\I 7) = \p{r}(\I 2)+\p{r}(\I 4)
        }
  \label{spec:mult-full-distr7}
\end{equation}
The side condition of projectability is satisfied by our assumption
that \p{op} is total (i.e.~terminating).

Finally, \ref{rule:wp-conj} is used (twice) to break the judgment in three subgoals,
grouping indices
  $\set{\I1,\I6,\I7}$,
  $\set{\I6,\I1,\I3}$, and
  $\set{\I7,\I2,\I4}$
respectively.
The first subgoal is an instance of distributivity on the second argument,
the other two of distributivity on the first argument. \subsection{Specifications for Idempotence}
\label{sec:specs-for-idemp}

\begin{lemma}
  $\eqref{spec:t-idemp3} \implies \eqref{spec:t-idemp}$.
\end{lemma}
\begin{proof}
  By the following derivation:
  \begin{derivation}
  \infer*[right=\ref{rule:wp-idx-merge}]{
  \infer*[Right=\ref{rule:idx-intro}]{
    \V \pv{x}(\I1)=\pv{x}(\I2) \land \pv{x}(\I3)=\vec{v}
    |- \WP{\m[\I 1: t, \I 2: t, \I 3: t]}
          { \pv{x}(\I2)=\vec{v} \implies \pv{x}(\I1)=\pv{x}(\I3) }
  }{
    \V
    \bigl(
    \pv{x}(\I1)=\pv{x}(\I2) \land \pv{x}(\I3)=\vec{v}
    \bigr)
    \isub{2->1}
    |-
    \bigl(
    \WP{\m[\I 1: t, \I 2: t, \I 3: t]}
          { \pv{x}(\I2)=\vec{v} \implies \pv{x}(\I1)=\pv{x}(\I3) }
    \bigr)
    \isub{2->1}
  }}{
    \V
    \pv{x}(\I3)=\vec{v}
    |-
    \WP{\m[\I 1: t, \I 3: t]}
          { \pv{x}(\I1)=\vec{v} \implies \pv{x}(\I1)=\pv{x}(\I3) }
  }
  \qedhere
  \end{derivation}
\end{proof}

An interesting fact is that
the derivation above that uses \ref{rule:wp-idx-merge}
can be implemented using \ref{rule:wp-conj}
and \ref{rule:wp-proj} if~$t$ is deterministic and $\pvar(t) \subs \pv{x}$.
What is of note is that in this application of \ref{rule:wp-proj}
the projectability condition is trivially true even if~$t$ diverges.
This is because another copy of~$t$ starting from the same state
is present in the hyper-triple in the conclusion.

\begin{lemma}
  $\eqref{spec:t-idemp3} \implies \eqref{spec:t-idemp-seq}$.
\end{lemma}
\begin{proof}
  By the following derivation, which mirrors the one presented
  in \cref{sec:overview:reindex}:
  \begin{derivation}
    \infer*[Right=\ref{rule:wp-nest}]{
    \infer*[Right=\ref{rule:wp-cons}]{
    \infer*[Right=\ref{rule:wp-idx-post}]{
    \infer*[Right=\ref{rule:wp-cons}]{
    \infer*[Right=\labelstep{step:idemp3-seq}]{
      \V \pv{x}(\I1)=\pv{x}(\I2) \land \pv{x}(\I3)=\vec{v}
      |- \WP{\m[\I 1: t, \I 2: t, \I 3: t]}
            { \pv{x}(\I2)=\vec{v} \implies \pv{x}(\I1)=\pv{x}(\I3) }
    }{
      \V \pv{x}(\I1)=\pv{x}(\I2)
      |- \WP{\m[\I 1: t, \I 2: t]}
            { \A \vec{v}.
              \pv{x}(\I3)=\vec{v} \implies
              \pv{x}(\I2)=\vec{v} \implies
              \WP{\m[\I 3: t]}{\pv{x}(\I1)=\pv{x}(\I3)}
            }
    }}{
      \V \pv{x}(\I1)=\pv{x}(\I2)
      |- \WP{\m[\I 1: t, \I 2: t]}
            { \pv{x}(\I3)=\pv{x}(\I2) \implies
              \WP{\m[\I 3: t]}{\pv{x}(\I1)=\pv{x}(\I3)}
            }
    }}{
      \V \pv{x}(\I1)=\pv{x}(\I2)
      |- \WP{\m[\I 1: t, \I 2: t]}
            { \bigl(
                \pv{x}(\I3)=\pv{x}(\I2) \implies
                \WP{\m[\I 3: t]}{\pv{x}(\I1)=\pv{x}(\I3)}
              \bigr)
              \isub{3->2}
            }
    }}{
      \V \pv{x}(\I1)=\pv{x}(\I2)
      |- \WP{\m[\I 1: t, \I 2: t]}
            {\WP{\m[\I 2: t]}{\pv{x}(\I1)=\pv{x}(\I3)}}
    }}{
      \V \pv{x}(\I1)=\pv{x}(\I2)
      |- \WP{\m[\I 1: t, \I 2: (t\p;t)]}
            { \pv{x}(\I1)=\pv{x}(\I2) }
    }
  \end{derivation}
  Step~\eqref{step:idemp3-seq} uses \cref{rule:wp-nest,,rule:wp-impl-r,,rule:wp-all}.
\end{proof}

\begin{lemma}
  $
    \bigl( \eqref{spec:t-idemp} \land (\textup{\textsc{Det}}_t) \bigr)
      \implies \eqref{spec:t-idemp3}
  $.
\end{lemma}
\begin{proof}
  By the following derivation:
  \[
    \infer*[Right=\ref{rule:wp-conj}]{
      \V
      { \pv{x}(\I1)=\pv{x}(\I2) }
      |- \WP {\m[\I 1: t, \I 2: t]}
             { \pv{x}(\I1)=\pv{x}(\I2) }
      \\
      \V
      { \pv{x}(\I3) = \vec{v}' }
      |- \WP {\m[\I 1: t, \I 3: t]}
             { \pv{x}(\I1) = \vec{v}' \implies \pv{x}(\I1)=\pv{x}(\I3) }
    }{
      \V
      { \pv{x}(\I1)=\pv{x}(\I2) \land \pv{x}(\I3) = \pr{\vec{v}} }
      |- \WP {\m[\I 1: t, \I 2: t, \I 3: t]}
             { \pv{x}(\I2) = \pr{\vec{v}} \implies \pv{x}(\I1)=\pv{x}(\I3) }
    }
   \qedhere
  \]
\end{proof}

\begin{lemma}
  $
    \bigl( \eqref{spec:t-idemp-seq} \land \proj(t) \bigr)
      \implies (\textup{\textsc{Det}}_t)
  $.
  However, $ \exists t. \eqref{spec:t-idemp-seq} \notimplies (\textup{\textsc{Det}}_t)$.
\end{lemma}
\begin{proof}
  For the first statement, we can build the following derivation:
  \begin{derivation}
    \infer*[right={\ref{rule:wp-proj}}]{
      \infer*[Right={\ref{rule:wp-conj}}]{
        \infer*{}{
          \V \pv{x}(\I1)=\pv{x}(\I2)
          |- \WP{\m[\I 1: t, \I 2: (t\p;t)]}
                { \pv{x}(\I1)=\pv{x}(\I2) }
        }
        \and
        \infer*[Right={\ref{rule:wp-rename}}]{
          \V \pv{x}(\I1)=\pv{x}(\I2)
          |- \WP{\m[\I 1: t, \I 2: (t\p;t)]}
                { \pv{x}(\I1)=\pv{x}(\I2) }
        }{
          \V \pv{x}(\I3)=\pv{x}(\I2)
          |- \WP{\m[\I 3: t, \I 2: (t\p;t)]}
                { \pv{x}(\I3)=\pv{x}(\I2) }
        }
      }{
        \V \pv{x}(\I1)=\pv{x}(\I2)=\pv{x}(\I3)
        |- \WP{\m[\I 1: t, \I 2: (t\p;t), \I 3: t]}
              { \pv{x}(\I1)=\pv{x}(\I2)=\pv{x}(\I3) }
      }
    }{
        \V \pv{x}(\I1)=\pv{x}(\I3)
        |- \WP{\m[\I 1: t, \I 3: t]}
              { \pv{x}(\I1)=\pv{x}(\I3) }
    }
  \end{derivation}
  The derivation is correct, with one important caveat:
  $(t\p;t)$ has to satisfy the projectability
  side condition of the \ref{rule:wp-proj} rule, which follows from $\proj(t)$.
  In fact, because our hyper-triples do not predicate
  over non-terminating traces,
  if~$t$ is not projectable, then the inference would be unsound.
  Take for example $t \is (\code{if x<0 then x=abs($*$) else diverge})$,
  where \p{abs} returns the absolute value of its argument.
  Since $t\p;t$ always diverges, $t$ trivially satisfies the idempotence hyper-triple, but is not deterministic on stores with $\p{x}<0$.
  This proves the second statement.
\end{proof}

\begin{lemma}
  $\exists t\st \eqref{spec:t-idemp} \notimplies \eqref{spec:t-idemp-seq}$
  and
  $\exists t\st \eqref{spec:t-idemp} \notimplies \eqref{spec:t-idemp3}$.
\end{lemma}
\begin{proof}
  Take $t=\bigl(\code{if x=0 then (if $\;*\;$ then x:=1 else x:=2)}\bigr)$.
  The result of a run from a store with $\p{x}=0$ would output a store with
  a random value~$\p{x} \in \set{1,2}$;
  executing the same program again would preserve the store,
  therefore the program satisfies~\eqref{spec:t-idemp}.
  It does not however satisfy~\eqref{spec:t-idemp-seq},
  nor $\eqref{spec:t-idemp3}$,
  as the programs at~\I{1} and~\I{2}
  may store different values in \p{x}.
\end{proof}

\begin{lemma}
  $\exists t\st \eqref{spec:t-idemp-seq} \notimplies \eqref{spec:t-idemp}$
  and
  $\exists t\st \eqref{spec:t-idemp-seq} \notimplies \eqref{spec:t-idemp3}$.
\end{lemma}
\begin{proof}
  The hyper-triple \eqref{spec:t-idemp-seq} is weaker than the other two
  in that it does not force~$\pv{x}$ to capture all the relevant input
  of the second run of~$t$.
  Take~$t = (\code{z:=z+y;y:=0})$ and~$\pv{x}=\p{z}$.
  Then~$t$ satisfies \eqref{spec:t-idemp-seq}:
  the second run of~$t$ always sees~$\p{y}=0$.
  This~$t$ however does not satisfy \eqref{spec:t-idemp},
  nor \eqref{spec:t-idemp3}.
  This is because the implication in the postcondition will only make sure
  that the value of \p{z} output at component~\I1 (resp.~\I2) is the same that
  was given as input to~\I2 (resp.~\I3); the value of \p{y} can however differ.
  To see this concretely,
  take stores $s_1$ with $\p{z}=5,\p{y}=1$ and $s_2$ with $\p{z}=6,\p{y}=1$.
  The variable \p{z} after running $t$ on $s_1$ as the same value as in $s_2$,
  but \p{y} would be zero.
  Running $t$ on $s_2$ would give $\p{z}=7$.

  Note that the choice of~$\pv{x}$ is application dependent:
  the notion of idempotence that is appropriate for an effectful~$t$ may
  require to ignore some portions of the state; for example~$t$ may log in \p{y}
  the number of times it is called, but acting in an idempotent way on the rest of the relevant state.
\end{proof}

\end{document}